\documentclass[11pt,a4paper]{article}

\usepackage{graphicx,amssymb,amsmath,amsthm,xcolor}
\usepackage{subfig}
\usepackage{floatrow}
\usepackage{fullpage,authblk}

\def\R{\mathbb R}
\def\d{\mathrm{d}}
\def\eps{\varepsilon}

\def\hdown{\mathrel{\raisebox{-\depth}{\rotatebox[origin=c]{180}{$ h $}}}}

\renewcommand{\vec}[1]{\mbox{\boldmath$#1$}}
\def\vecnabla{\vec{\nabla}}
\newtheorem{theorem}{Theorem}
\newtheorem{lemma}{Lemma}
\newtheorem{corollary}{Corollary}
\newtheorem{hypothesis}{Hypothesis}
\newtheorem{remark}{Remark}
\newtheorem{definition}{Definition}
\newtheorem{property}{Property}

\begin{document}
\title{Globally optimal stretching foliations  of dynamical systems reveal the organizing skeleton of intensive instabilities}

\author{Sanjeeva Balasuriya}
\affil{School of Mathematical Sciences, University of Adelaide, Adelaide SA 5005, Australia}

\author{Erik M.~Bollt}
\affil{Department of Electrical and Computer Engineering and $C^3S^2$ the Clarkson Center for Complex Systems Science, Clarkson University, Potsdam, New York 13699}

%\date{\today}

\maketitle

\begin{abstract}
Understanding instabilities in dynamical systems drives to the heart of modern chaos theory, whether forecasting or attempting to control future outcomes.  Instabilities in the sense of
locally maximal stretching in maps is well understood, and is connected to the concepts of Lyapunov exponents/vectors, Oseledec spaces and the Cauchy--Green tensor.  In this paper, we extend the concept to global optimization of stretching, as this forms a skeleton organizing the general instabilities. The `map' is general but incorporates the inevitability
of finite-time as in any realistic application: it can be defined via a finite sequence of discrete maps, or a finite-time flow associated with a continuous dynamical
system.  Limiting attention to two-dimensions, we formulate the global optimization problem as one over a restricted
class of foliations, and establish the foliations which both maximize and minimize global stretching.  A classification of nondegenerate singularities of the foliations is obtained.  Numerical issues in
computing optimal foliations are examined, in particular insights into special curves along which foliations appear to
veer and/or do not cross, and foliation behavior near singularities.  Illustrations and validations of the results to
the H\'{e}non map, the double-gyre flow and the standard map are provided.
  \end{abstract}

{{\bf Keywords}: \it {Lyapunov vector , finite-time flow, punctured foliation}} 

{ {\bf Mathematics Subject Classification}: \it{		37B55, 37C60, 53C12 }}

%%Research highlights

%\captionsetup[subfigure]{justification=left,singlelinecheck=false}

%\setlength{\topfraction}{0.9}

%\renewcommand{\topfraction}{1}
%\renewcommand{\floatpagefraction}{0.6}

%%%%%%%%%%%%%%%%

%%Graphical abstract
\newpage
\section{Graphical Abstract}
\noindent Sanjeeva Balasuriya, Erik Bollt
\medskip

\vspace*{0.4cm}
\includegraphics[scale=0.38]{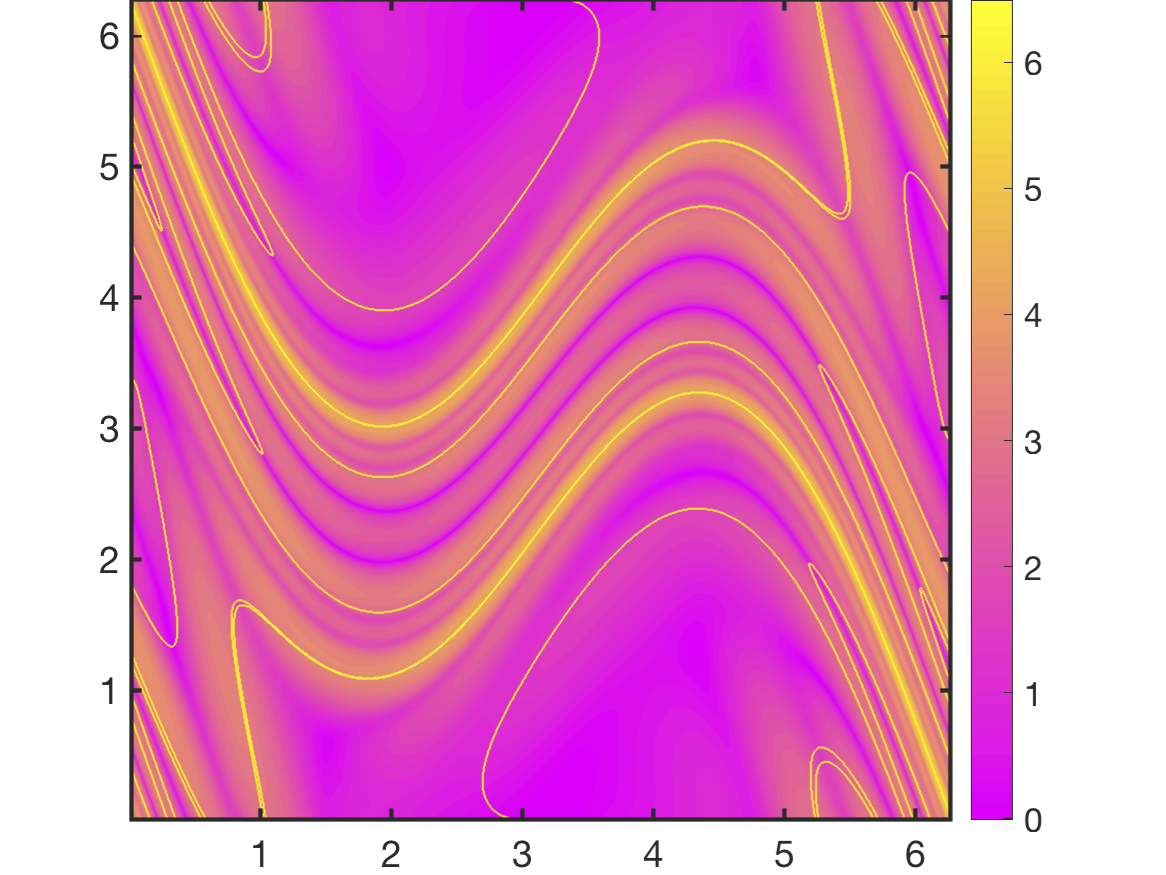} \hspace*{-1cm}
\includegraphics[scale=0.38]{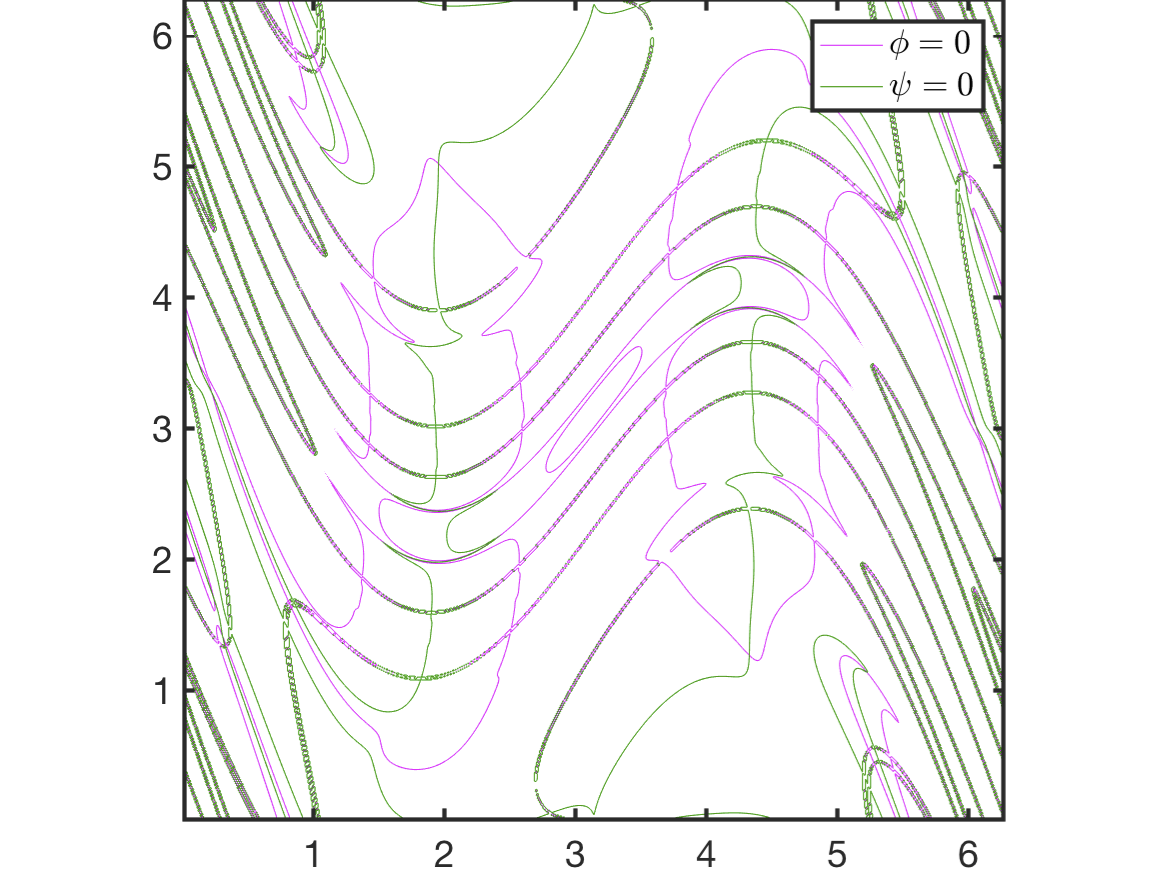} 

\vspace*{0.4cm}
\includegraphics[scale=0.39]{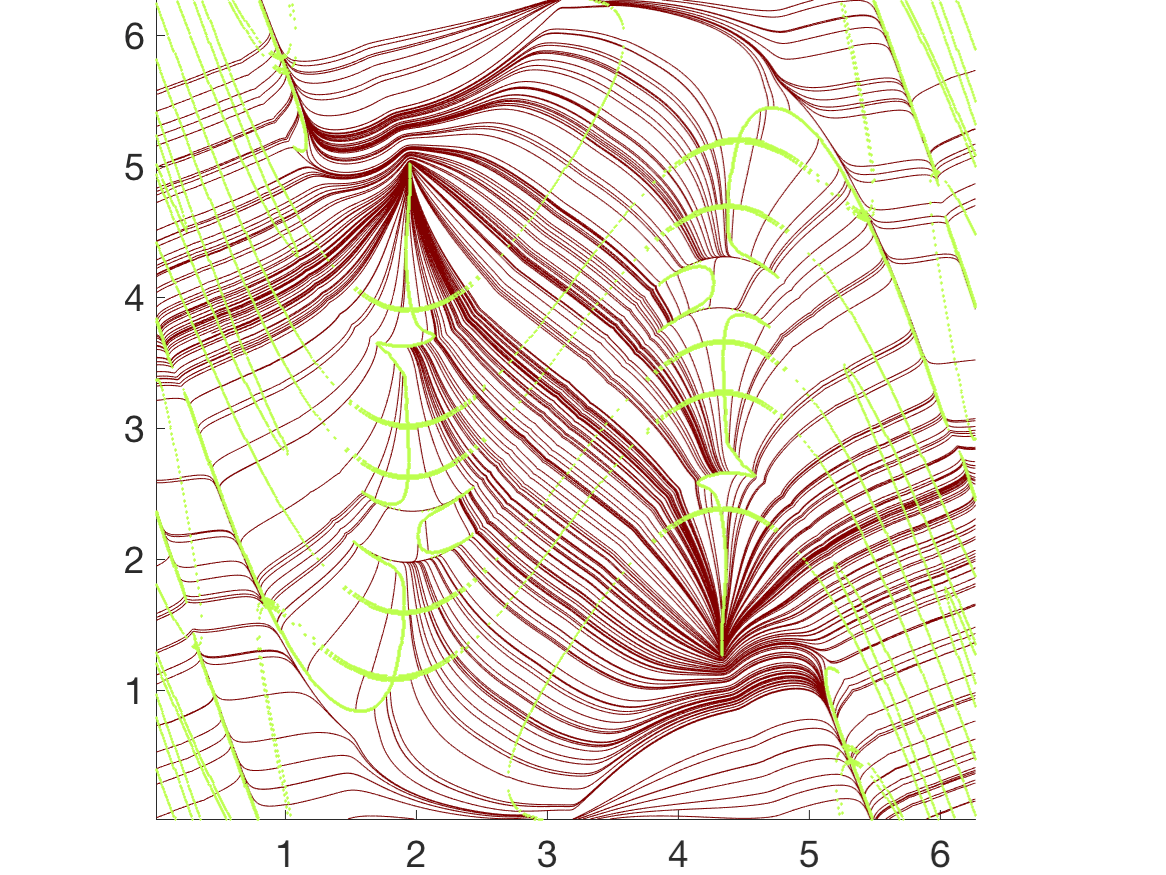} \hspace*{-1cm}
\includegraphics[scale=0.39]{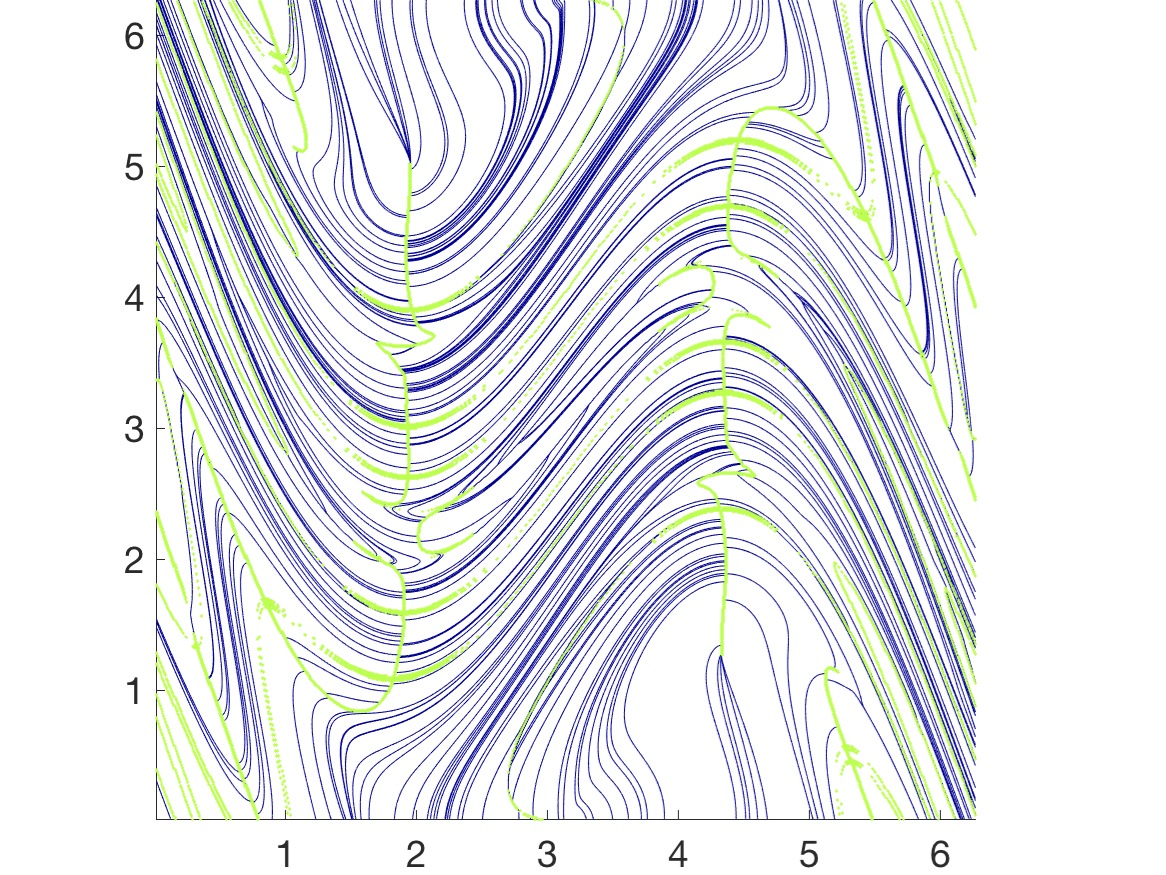}
%\includegraphics{grabs}
%\end{graphicalabstract}

\paragraph{Highlights}
\begin{enumerate}
\item Understanding the organizing skeleton of instability for orbits must be premised on  analysis of  globally optimal stretching.
\item Provides the theory to obtain foliation for globally optimizing stretching for any two-dimensional map (analytically
specified, derived from a finite-time flow or a sequence of maps, and/or given via data);
\item Classifies singularities and provides insight and solutions to spurious artefacts emerging when attempting
to numerically determine such a foliation;
\item Establishes connections with a range of well-established methods: locally optimizing stretching,
Cauchy--Green eigenvalues and singularities,
Lyapunov exponents, Lyapunov vectors, Oseledec spaces, and variational Lagrangian coherent structures.
\end{enumerate}

%\begin{keyword}
%%% keywords here, in the form: keyword \sep keyword
%Lyapunov vector \sep finite-time flow \sep punctured foliation
%%% PACS codes here, in the form: \PACS code \sep code
%
%%% MSC codes here, in the form: \MSC code \sep code
%%% or \MSC[2008] code \sep code (2000 is the default)
%\MSC[2010] 37B55 \sep 37C60 \sep 53C12
%\end{keyword}
%
%\end{frontmatter}

\newpage
\section{Introduction}
\label{sec:introduction}

A central topic of dynamical systems theory involves analysis of instabilities, since this is the central ideas behind the possibility of forecast time horizon, or even of ease of control of future outcomes.  The preponderance of work has involved analysis of local instability, whether by the Hartman-Grobman theorem and center manifold theorem \cite{guckenheimerholmes} for periodic orbits and similarly for invariant sets \cite{sackersell}.  For general orbits, local instability is characterized by Oseledec spaces \cite{oseledec} which
are identified via Lyapunov exponents \cite{shadden} and Lyapunov vectors \cite{wolfesamelson,ramasubramaniansriram}.  Via these techniques, 
{\em locally} optimizing stretching due to the operation of a map from subsets of $ \mathbb{R}^n $ to subsets
of $ \mathbb{R}^n $ is well-understood.  Computing the map's derivative matrix at each point is allows for 
computation of Oseledecs/Lyapunov information:  its singular values 
and corresponding singular vectors are respectively associated with stretching rates and relevant 
directions in the domain, and its (scaled) operator norm is the classical Lyapunov exponent of the orbit beginning at 
that point.

In this paper, we assert that understanding the global dynamics---how a system organizes orbits---is related to a global view of instabilities.  The related organizing skeleton of orbits must therefore be premised on analysis of  {\em globally optimal stretching}.  Here, orbits will be in relation to two-dimensional maps which can be derived from
various sources: a finite sequence of discrete maps, or a flow occurring over a finite time period.  The latter situation is
particularly relevant when seeking regions in unsteady flows which remain `coherent' over a given time period \cite{glcs}.  
In all these cases, we emphasize that we are {\em not} seeking to understand stretching in the infinite-time limit---which is
the focus in many classical approaches \cite{oseledec,sackersell}---but rather stretching associated with a one-step
map derived from any of these approaches.  From the applications perspective, the one-step map would be parametrized
by the discrete or continuous time over which the map operates, and this number would of necessity be finite in any
computational implementation.

When additionally seeking {\em global optimization}, the first issue is defining what this means with respect to a bounded open
domain on which the map operates.   In Section~\ref{sec:optimizing}, we pose this question as an
optimization over {\em foliations}, but need to restrict these foliations in a certain way because they would generically
have singularities.  We are able to characterize
the restricted foliations of optimal stretching (minimal or maximal)  in a straightforward geometric way, while establishing connections to
well-known local stretching optimizing entities.  We provide a complete classification of the nondegenerate singularities
using elementary arguments in Section~\ref{sec:singularity}, thereby easily identifying $ 1 $- and $ 3 $-pronged
singularities as the primary scenarios.  We argue in Section~\ref{sec:discontinuity} the inevitability of a `branch
cut' phenomenon if attempting to compute these restricted 
foliations using a vector field; this will generically possess discontinuities across one-dimensional curves which we can
characterize.  Other computational ramifications are addressed in Section~\ref{sec:computing_foliations}, which includes
issues of curves stopping abruptly when coming in horizontally or vertically, and veering along spurious curves.  We are able to give explicit insights into
the emergence of these issues as a result of standard numerical implementations, and we suggest an alternative integral-curve formulation which avoids these
difficulties.  In Section~\ref{sec:numerical}, we demonstrate computations of globally optimal restricted foliations for several well-known
examples:
the H\'{e}non map \cite{henon}, the Chirikov (standard) map \cite{chirikov}, and the double-gyre flow \cite{shadden},
each implemented over a finite time.
The aforementioned numerical issues are highlighted in these examples.

%%%%%%%%%%%%%%%%
\section{Globally optimizing stretching}
\label{sec:optimizing}

Let $ \Omega $ be a bounded two-dimensional subset of $ \R^2 $ consisting of a finite union of connected open sets, each
of whose closure has at most a finite number of boundary components.   So $ \Omega $ may, for example, consist of disconnected open sets 
and/or entities which are topologically equivalent to the interior of an annulus.   We will use
$ (x,y) $ to denote points in $ \Omega $.  Let 
$ \vec{F} $ be a map on $ \Omega $ to $ \R^2 $ which is given componentwise by
\begin{equation}
\vec{F} \left(x,y \right) = \left( \begin{array}{c} u(x,y) \\ v(x,y) \end{array} \right) \, .
\label{eq:map}
\end{equation}

\begin{hypothesis}[Smoothness of $ \vec{F} $]
\label{hypothesis:smooth}
Let the map $ \vec{F} \in {\mathrm{C}}^2(\Omega) $.
\end{hypothesis}

Physically, we note that $ \vec{F} $ can be generated in various ways.  It can be simply one iteration of a given map, multiple (finitely-many) iterations of a map, or even the application of a finite {\em sequence} of maps.  It can also be the flow-map generated from a nonautonomous flow in two-dimensions over a 
finite time.  In this sense, $ \vec{F} $ encapsulates the fact that {\em finiteness} is inevitable in any numerical, experimental or observational situation, 
while allowing for both discrete and continuous time, as well as nonautonomy.  The time over which the system operates 
can be thought of as a parameter which is encoded within $ \vec{F} $, and its effect can be investigated if needed by
varying this parameter.
 
The relative stretching of a tiny line (of length $ \delta > 0 $) placed at a point $ (x,y) $ in $ \Omega $, with an orientation given by
$ \theta \in [-\pi/2,\pi/2) $ due to the action of $ \vec{F} $ is
\[
\Lambda(x,y,\theta) = \lim_{\delta \rightarrow 0} \frac{\left\| \vec{F}\left( x + \delta \cos \theta, y + \delta \sin \theta \right)
- \vec{F}(x,y) \right\|}{\delta} \, .
\]
This is the magnitude of $ \vec{F} $'s directional derivative in the $ \theta $ direction. It is clear that 
\begin{equation}
\hspace*{-0.7cm} \Lambda(x,y,\theta) := \left\| \vecnabla \vec{F}(x,y) \left( \begin{array}{c} \cos \theta \\ \sin \theta \end{array} \right) \right\| 
= \left\| \left( \! \! \begin{array}{cc} u_x(x,y) & u_y(x,y) \\ v_x(x,y) & v_y(x,y) \end{array} \! \! \right) \,  \left( \! \begin{array}{c} \cos \theta \\ \sin \theta \end{array} \! \right) \right\| \, .
\label{eq:stretching}
\end{equation}
We refer to $ \Lambda(x,y,\theta) $ in (\ref{eq:stretching}) as the {\em local stretching} associated with a point
$ (x,y) \in \Omega $; note that this also depends on a choice of angle $ \theta $ in which an infinitesimal line is 
to be positioned.  
If we take the supremum over all $ \theta \in [-\pi/2,\pi/2) $ of the right-hand side of (\ref{eq:stretching}), we would 
get the operator
(matrix) norm $ \left\| \vecnabla \vec{F} \right\| $, computable for example via {\em Cauchy--Green tensor}
\begin{equation}
C(x,y) := \left[ \vecnabla \vec{F} (x,y) \right]^\top \vecnabla \vec{F}(x,y) \, .
\label{eq:cauchygreen}
\end{equation}
Thus, our development has close relationships to well-established
methods related to the Cauchy--Green tensor, finite-time Lyapunov exponents, and methods for determining Lagrangian
coherent structures, which we describe in more detail in \ref{sec:localstretching}.  However, at this stage
our local stretching definition in (\ref{eq:stretching}) is $ \theta $-dependent.

\begin{definition}[Isotropic and remaining sets]
\label{def:isotropic}
The {\em isotropic set} $ I \subset \Omega $ is defined by
\begin{equation}
I := \left\{ (x,y) \in \Omega \, \, : \, \, \frac{\partial \Lambda(x,y,\theta)}{\partial \theta} = 0 \, \right\} \, ,
\label{eq:puncture}
\end{equation}
and the {\em remaining set} is
\begin{equation}
\Omega_0 := \Omega \setminus I \, .
\label{eq:remaining}
\end{equation}
\end{definition}
 
 The isotropic set $ I $ consists of points at which the local stretching does not depend on directionality of a local
 line segment.  Given the smoothness we have assumed in $ \vec{F} $, $ I $ must be a `nice' closed set; it cannot, for example, be
 fractal.  In general, $ I $ may be empty, equal to $ \Omega $, or consist of a mixture of finitely many isolated
 points and closed regions of $ \Omega $. 
  
We are seeking a partition of $ \Omega $ into a family of nonintersecting curves, such that global stretching is optimized
in a way to be made specific.  
Since the local stretching at points in $ I $ is impervious to the directionality of lines passing through them, these
families of curves only need be defined on $ \Omega_0 = \Omega \setminus I $, with the understanding that this has nonempty interior.  In more formal language, we need to 
think of singular codimension-$ 1 $ foliations on $ \Omega $, whose singularities are restricted to $ I $.  We codify this
in terms of the required geometric properties of the family of curves:

\begin{definition}[Restricted foliation]
\label{definition:foliation}
A {\em restricted foliation}, $ f $,  on $ \Omega $ consists of a family of curves defined in the remaining set 
$ \Omega_0 $ such that
\begin{itemize}
\item[(a)] The curves of $ f $ (`the leaves of the foliation') are disjoint;
\item[(b)] The union of all these curves covers $ \Omega_0 $;
\item[(c)] The tangent vector varies in a $ {\mathrm{C}}^1 $-smooth fashion along each curve.
\end{itemize}
\end{definition}

Our definition is consistent with the local properties expected from a formal definition of foliations on manifolds
\cite{lawson}, but bears in mind that $ \Omega_0 $ is not a manifold because of the omission of the closed set
$ I $ from $ \Omega $.  
We remark that if $ I $ consists of a finite number of points, our restricted foliation definition is equivalent to that of a 
`punctured foliation' \cite{mosher} on $ \Omega $, where the punctures are at the points in $ I $.  This turns out to be a generic
expectation for $ I $, and we will examine this (both theoretically and numerically) in more detail later.

The properties of Definition~\ref{definition:foliation} ensure that every restricted foliation $ f $ is associated with a unique $ {\mathrm{C}}^1 $-smooth {\em angle field} on the remaining set $ \Omega_0 $  in the following sense.  Given a
point $ (x,y) \in \Omega_0 $, there exists a unique curve from $ f $ which passes through it.  
The tangent line drawn at this point makes an 
angle $ \theta_f $ with the positive $ x $-axis.  This angle can always be chosen uniquely modulo $ \pi $, from the set 
$ [-\pi/2,\pi/2) $: vertical lines have $ \theta_f = - \pi/2 $, while horizontal lines have $ \theta = 0 $.  Thus,
every foliation induces a unique angle field $ \theta_f: \Omega_0 \rightarrow [-\pi/2,\pi/2) $ (modulo $ \pi $).  The angle
field must be $ {\mathrm{C}}^1 $-smooth to complement the continuous variation in the tangent spaces of
$ f $'s leaves.  Conversely, suppose a $ {\mathrm{C}}^1 $-smooth angle field $ \theta_f: \Omega_0 \rightarrow [-\pi/2,\pi/2) $ (modulo $ \pi $)  is given.
Given an arbitrary point $ (x_\alpha,y_\alpha) \in \Omega_0 $,  the existence of solutions to the differential equation
\[
\left( \sin \theta_f(x,y) \right) \d x - \left( \cos \theta_f(x,y) \right) \d y = 0 \, 
\]
passing through the point $ (x_\alpha,y_\alpha) $ ensures that there is an integral curve of the form $ g_\alpha(x,y) = 0 $, in which
$ g_\alpha $ is $ {\mathrm{C}}^1 $-smooth in both arguments.  This is possible for each and every $ (x_\alpha,y_\alpha) \in  \Omega_0 $, and uniqueness ensures
that the curves $ g_\alpha(x,y) = 0 $ do not intersect one another.  Moreover, $ \Omega_0 $ is spanned by $ \bigcup_\alpha \left\{ (x,y) \, : \, g_\alpha(x,y) = 0 \right\} $ because $ \Omega_0 = \bigcup_\alpha \left\{ (x_\alpha,y_\alpha) \right\} $, ensuring that there is a curve passing through every point $ (x_\alpha,y_\alpha) $.  Hence, this process generates a unique restricted foliation $ f $ on $ \Omega_0 $.

We are now in a position to define the global stretching which we seek to optimize.

\begin{definition}[Global stretching]
{\em 
Given any restricted foliation $ f $, we define the {\em global stretching} on $ \Omega $ as the local stretching
integrated over $ \Omega $, i.e., 
\begin{equation}
\Sigma_f := \int \! \! \! \! \int_{\Omega_0} \Lambda \left( x, y, \theta_f(x,y) \right) \, \d x \, \d y 
+ \int \! \! \! \! \int_I \Lambda \left( x, y, \centerdot \right) \, \d x \, \d y \, , 
\label{eq:globalstretching}
\end{equation}
in which $ \theta_f $ is the angle field induced by a choice of restricted foliation $ f $.
}
\end{definition}
Notice that the integral over the full domain $ \Omega $ has been split into one
over $ \Omega_0 $ (on which $ f $ and thus $ \theta_f $ is well-defined) and over $ I $ (over which the directionality
has no influence on $ \Lambda $, and has thus been omitted).  Thus, any understanding of foliation leaves on
$ I $ is irrelevant to the global stretching, motivating our definition of restricted foliation defined only on $ \Omega_0 $.

As central premise of this work, 
we seek restricted foliations $ f $ which optimize (maximize, as well as minimize) $ \Sigma_f $.  Partitions of 
$ \Omega_0 $ which are extremal in this way represent the greatest instability or most stability associated with the dynamical system, and so orbits associated with these are distinguished for their corresponding difficulties in forecasting, or alternatively, relative coherence.
Before we state the main theorems, some notation is needed.  On $ \Omega $, we define the 
$ {\mathrm{C}}^1 $-smooth functions
\begin{equation}
\phi(x,y) = \frac{u_x(x,y)^2 + v_x(x,y)^2 - u_y(x,y)^2 - v_y(x,y)^2}{2}
\label{eq:phi}
\end{equation}
and
\begin{equation}
\psi(x,y) = u_x(x,y) u_y(x,y) + v_x(x,y) v_y(x,y)
\label{eq:psi}
\end{equation}
in terms of the partial derivatives $ u_x $, $ u_y $, $ v_x $ and $ v_y $ of the mapping $ \vec{F} $.  
First, we show the connection between zero level sets of $ \phi $ and $\psi $ and the isotropic set $ I $.

\begin{lemma}[Isotropic set]
\label{lemma:puncture}
The isotropic set $ I $ defined in (\ref{eq:puncture}) can be equivalently characterized by
\begin{equation}
I := \left\{ (x,y) \in \Omega \, : \, \phi(x,y) = 0 \, \, {\mathrm{and}} \, \, \psi(x,y) = 0 \right\} \, ,
\label{eq:B0}
\end{equation}
\end{lemma}

\begin{proof}
See \ref{sec:puncture}.
\end{proof}

We reiterate from this recharacterization of $ I $ that generically, it will consist of finitely many points (at which the curves $ \phi(x,y) = 0 $ intersect the
curves $ \psi(x,y) = 0 $), but may contain curve segments (if the two curves are tangential in a region), or 
areas (if both $ \phi $ and $ \psi $ are zero in two-dimensional regions).   Even for the generic case (finitely
many isolated points), we will see that $ I $ will strongly influence the nature of the optimal foliations in $ \Omega_0 $.  

Next,
we define the angle field $ \theta^+: \Omega_0 \rightarrow [-\pi/2,\pi/2) $ by
\begin{equation}
\theta^+(x,y) := \frac{1}{2} \, \tilde{\tan}^{-1} \left( \psi(x,y), \phi(x,y) \right) \qquad  ({\mathrm{mod}} \, \pi) \, , 
\label{eq:thetaplus}
\end{equation}
in terms of the four-quadrant inverse tangent function $ \tilde{\tan}^{-1} (\tilde{y},\tilde{x}) $ (sometimes called atan2 in computer science applications, which assigns the angle in 
$ [-\pi,\pi) $ associated with the quadrant in $ \left( \tilde{x}, \tilde{y} \right) $-space when computing $ \tan^{-1}(\tilde{y}/\tilde{x}) $). 
We also define the angle field $ \theta^-: \Omega_0 \rightarrow [-\pi/2,\pi/2) $ by
\begin{equation}
\theta^-(x,y) = \frac{\pi}{2} + \frac{1}{2} \, \tilde{\tan}^{-1} \left( \psi(x,y), \phi(x,y) \right) \qquad  ({\mathrm{mod}} \, \pi) \, ,
\label{eq:thetaminus}
\end{equation}
and observe that
\begin{equation}
\theta^+(x,y) - \theta^-(x,y)  = - \frac{\pi}{2} \qquad ({\mathrm{mod}} \, \pi) \, .
\label{eq:thetaorthogonal}
\end{equation}

\begin{lemma}[Equivalent characterizations of angle fields, $ \theta^\pm $]
\label{lemma:theta}
On $ \Omega_0 $, $ \theta^\pm \in [-\pi/2,\pi/2) $ are representable as
\begin{equation}
\theta^+(x,y)  := \tan^{-1}  \frac{ - \phi(x,y) + \sqrt{ \phi(x,y)^2 + \psi(x,y)^2}}{\psi(x,y)} \qquad  ({\mathrm{mod}} \,  \pi) \label{eq:thetaplus2}
\end{equation}
and
\begin{equation}
\theta^-(x,y) := \tan^{-1}  \frac{ - \phi(x,y) - \sqrt{ \phi(x,y)^2 + \psi(x,y)^2}}{\psi(x,y)} \qquad  ({\mathrm{mod}}  \, \pi)  \, . 
\label{eq:thetaminus2}
\end{equation}
\end{lemma}

\begin{proof}
See \ref{sec:theta}.
\end{proof}

\begin{remark}[Removable singularities at $ \psi = 0 $ and $ \phi \ne 0 $]
\label{remark:removable}
{\em 
While it appears that points where $ \psi = 0 $ but $ \phi \ne 0 $ are not in the domain of $ \theta^+ $ as written
in (\ref{eq:thetaplus2}) and (\ref{eq:thetaminus2}),
these turn out to be removable singularities, and thus can be thought of in the sense of keeping $ \phi $ constant
and taking the limit $ \psi \rightarrow 0 $.  More specifically, this implies that
\begin{equation}
\theta^+(x,y) \Big|_{\psi=0} = \left\{ 
\begin{array}{ll}
-\pi/2 & ~~~~{\mathrm{if}} \, \phi < 0 \\
0  & ~~~~{\mathrm{if}} \, \phi > 0
\end{array} \right. \, .
\label{eq:thetaplus_zero}
\end{equation}
With this understanding of dealing with the removable singularities, we will simply view (\ref{eq:thetaplus2}) as being
defined on $ \Omega_0 $.  Similarly, 
\begin{equation}
\theta^-(x,y) \Big|_{\psi=0} = \left\{ 
\begin{array}{ll}
0 & ~~~~{\mathrm{if}} \, \phi < 0 \\
-\pi/2  & ~~~~{\mathrm{if}} \, \phi > 0
\end{array} \right. \, .
\label{eq:thetaminus_zero}
\end{equation}
}
\end{remark}

\begin{remark}[Smoothness of $ \theta^\pm $ in $ \Omega_0 $]
\label{remark:thetasmooth}
{\em Subject to the removable singularity understanding of Remark~\ref{remark:removable}, $ \theta^+ $ and $ \theta^- $ are $ {\mathrm{C}}^1 $-smooth
in $ \Omega_0 $, and thereby respectively induce well-defined foliations $ f^+ $ and $ f^- $ on $ \Omega_0 $.
}
\end{remark}

While theoretically, the alternative expressions in (\ref{eq:thetaplus2})-({\ref{eq:thetaminus2})  for $ \theta^\pm $ are  equivalent
to the definitions in (\ref{eq:thetaplus})-(\ref{eq:thetaminus}), practically in fact, which of these is chosen {\em will} cause differences when
performing {\em numerical} optimal foliation computations.  We will highlight similarities
and differences between their usage in Section~\ref{sec:computing_foliations}, and demonstrate
these issues numerically in Section~\ref{sec:numerical}.  

We can now state our
first main result:
\begin{theorem}[Stretching Optimizing Restricted Foliation - Maximum ($ \mbox{SORF}_{max}$)]
\label{theorem:max}
The restricted foliation $ f^+ $ which maximizes the global stretching (\ref{eq:globalstretching}) is
that associated with the angle field $ \theta^+ $.  The corresponding maximum of the global stretching 
 (\ref{eq:globalstretching}) is 
\begin{equation}
\Sigma^+= \int \! \! \! \! \int_{\Omega} \left[ \frac{\left| \vecnabla u \right|^2 +
\left| \vecnabla v \right|^2}{2} + \sqrt{\phi^2 + \psi^2} \right]^{1/2} \, \d x \, \d y \, .
\label{eq:maxstretching}
\end{equation}
\end{theorem}

\begin{proof}
See \ref{sec:proof}.
\end{proof}

\begin{remark}[Lyapunov exponent field]
\label{remark:lyapunov}
{\em The integand of (\ref{eq:maxstretching}) is the $ \Lambda $ field associated with maximizing stretching,
and is given by
\begin{equation}
\Lambda^+(x,y) = \left[ \frac{\left| \vecnabla u \right|^2 +
\left| \vecnabla v \right|^2}{2} + \sqrt{\phi^2 + \psi^2} \right]^{1/2}  = \left\| \vecnabla \vec{F}(x,y) \right\| \, .
\label{eq:ftle}
\end{equation}
This is (a scaled version of) the standard {\em Lyapunov exponent field}.  We avoid a time-scaling here since, for example, $ \vec{F} $ may be derived
from a sequence of application of various forms of maps (indeed, any sequential combination of discrete maps and
continuous flows).  Neither will we take a logarithm, since we do not necessarily want to think of the
stretching field as an exponent because the finite `amount of time' associated with $ \vec{F} $ depends on
its discrete/continuous nature, which is flexible in our implementation.
}
\end{remark}

\begin{remark}[Stretching on the isotropic set $ I $]
{\em
The value of the global stretching restricted to $ I $ (i.e., the second integral in
(\ref{eq:globalstretching})) is, from (\ref{eq:maxstretching}), 
\begin{eqnarray}
\int \! \! \! \! \int_{I} \left[ \frac{\left| \vecnabla u \right|^2 +
\left| \vecnabla v \right|^2}{2} \right]^{1/2} \! \!  \d x \, \d y  \! \! & = & \! \frac{1}{2} \int \! \! \! \! \int_{I}  \left\| \vecnabla \vec{F}
\right\|_{\mathrm{Frob}} \, \d x \, \d y  \nonumber \\
&  = & \frac{1}{2} \int \! \! \! \! \int_{I} \left\{ {\mathrm{Tr}} \left[ 
\vecnabla \vec{F} \left( \vecnabla \vec{F} \right)^\top \right]  \right\}^{1/2} \! \!  \d x \, \d y \nonumber \\, 
&  = & \frac{1}{2} \int \! \! \! \! \int_{I} \left\{ {\mathrm{Tr}} \left[ 
C(x,y)  \right]  \right\}^{1/2}   \d x \, \d y \, , 
\label{eq:frobenius}
\end{eqnarray}
expressed in terms of the Frobenius norm $ \left\| \centerdot \right\|_{\mathrm{Frob}} $ or trace $ {\mathrm{Tr}} \left[ 
\centerdot \right] $ of the Cauchy--Green tensor 
(\ref{eq:cauchygreen}). 
}
\end{remark}

Similar to the maximizing result, we also have the minimal foliation:

\begin{theorem}[Stretching Optimizing Restricted Foliation - Minimum ($ \mbox{SORF}_{min}$)]
\label{theorem:min}
The restricted foliation $ f^- $ which minimizes the global stretching (\ref{eq:globalstretching}) is
that associated with the angle field $ \theta^- $.  The corresponding minimum of the global stretching 
 (\ref{eq:globalstretching}) is 
\begin{equation}
\Sigma^- = \int \! \! \! \! \int_{\Omega} \left[ \frac{\left| \vecnabla u \right|^2 +
\left| \vecnabla v \right|^2}{2} - \sqrt{\phi^2 + \psi^2} \right]^{1/2} \, \d x \, \d y \, .
\label{eq:minstretching}
\end{equation}
\end{theorem}

\begin{proof}
See \ref{sec:proof}.
\end{proof}

\begin{corollary}[$\mbox{SORF}_{max}$ and $ \mbox{SORF}_{min}$ are orthogonal]
\label{corollary:orthogonal}
If any curve from $ \mbox{SORF}_{max}$ intersects a curve from $ \mbox{SORF}_{min}$ in $ \Omega_0 $, then it does so orthogonally.
\end{corollary}

\begin{proof}
The $ \mbox{SORF}_{max}$ and $ \mbox{SORF}_{min}$ curves are respectively tangential to
the angle fields $ \theta^+ $ and $ \theta^- $, which are known to be orthogonal by (\ref{eq:thetaorthogonal}).
\end{proof}

There is clearly a strong interaction between {\em local} properties and quantities related to {\em global} 
stretching optimization.  We summarize some properties below.  We do not discuss them in detail, but provide
additional explanations in \ref{sec:localstretching}.

\begin{remark}[Maximal or minimal {\em local} stretching]
\label{remark:local}
{\em 
\mbox{}
\begin{itemize}
\setlength{\itemsep}{-1pt}
\item[(a)] Given a point $ (x,y) \in \Omega_0 $, if we pose the question of determining the orientation of
an infinitesimal line positioned here in order to experience the maximum stretching, then this is at an angle $ \theta^+ $.  
\item[(b)] The local maximal stretching associated with choosing the angle of orientation $ \theta^+ $ is exactly the operator norm
of the gradient of the map $ \vec{F} $, which is expressible in terms of the Cauchy--Green tensor (\ref{eq:cauchygreen}).
\item[(c)] The above quantity is associated with the {\em Lyapunov exponent field}, given in (\ref{eq:ftle}), which 
is defined on all of $ \Omega $ despite having the above interpretation only on $ \Omega_0 $.
\item[(d)] In $ \Omega_0 $, the $ \mbox{SORF}_{max}$ leaves (curves) lie along streamlines of the eigenvector field of the Cauchy--Green
tensor corresponding to the larger eigenvalue.  This eigenvector field can also be thought of as the
{\em Lyapunov} or {\em Oseledec vector field} associated with $ \vec{F} $.
\item[(e)] If the question is instead to find the orientation of an infinitesimal line positioned at $ (x,y) $ in order to experience
the minimum stretching, then the angle of this line is $ \theta^- $.  Compare this statement to observation (a) together with Corollary \ref{corollary:orthogonal}.
\item[(f)] In $ \Omega_0 $, Eq.~(\ref{eq:remaining}), the $ \mbox{SORF}_{min}$ leaves lie along streamlines of the eigenvector field of the Cauchy--Green
tensor corresponding to the smaller eigenvalue.
\item[(g)]  The set $ I $ corresponds to points in $ \Omega $ at which the two eigenvalues of the Cauchy--Green
tensor coincide.  
\end{itemize}
}
\end{remark}

%\begin{remark}
%{\em 
%The second integral in (\ref{eq:globalstretching}) (over $ I $) is irrelevant to the global optimization problem solved via Theorems~\ref{theorem:max} and \ref{theorem:min}.  Thus, if we  pose  our optimization problem as instead seeking foliations $ \tilde{f}^\pm $ on $ \Omega $ which maximize or  minimize the global stretching, since the choice of curves within $ I $ has no impact on the global stretching, we can choose any foliation $ \tilde{f}^\pm $ such that $ \tilde{f}^\pm \Big|_{\Omega_0} = f^\pm $, that is, which extends $ f^\pm $ from $ \Omega_0 $ to $ \Omega $ in any way which ensures that the extension is a foliation on $ \Omega $.
%}
% NOT CLEAR THAT AN EXTENSION AS A GENUINE FOLIATION IS POSSIBLE IN THE PUNCTURED CASE
%\end{remark}

%%%%%%%%%%%
\section{Behavior near singularities}
\label{sec:singularity}

The previous section's optimization ignored the isotropic set $ I $, since the local stretching within $ I $ was
independent of direction.    In this section, we analyze the topological structure of our optimal 
foliations near generic points in $ I $, which can be thought of as {\em singularities} with respect to optimal foliations. 
By Lemma~\ref{lemma:puncture}, these are points where both $ \phi $ and $ \psi $ are zero.

\begin{definition}[Nondegenerate singularity]
\label{definition:nondegenerate}
If a point $ \vec{p} \in I $ is such that 
\begin{equation}
\Big[ \vecnabla \phi \times \vecnabla \psi \Big]_{\vec{p}} \ne \vec{0} \quad {\mathrm{or~equivalently}} \quad
\mathrm{det} \, \frac{\partial(\phi,\psi)}{\partial(x,y)} \Big|_{\vec{p}} \ne 0 \, , 
\label{eq:nondegenerate}
\end{equation}
then $ \vec{p} $ is a {\em nondegenerate singularity}. 
\end{definition}

\begin{figure}
{\includegraphics[width=0.47\textwidth,height=0.17\textheight]{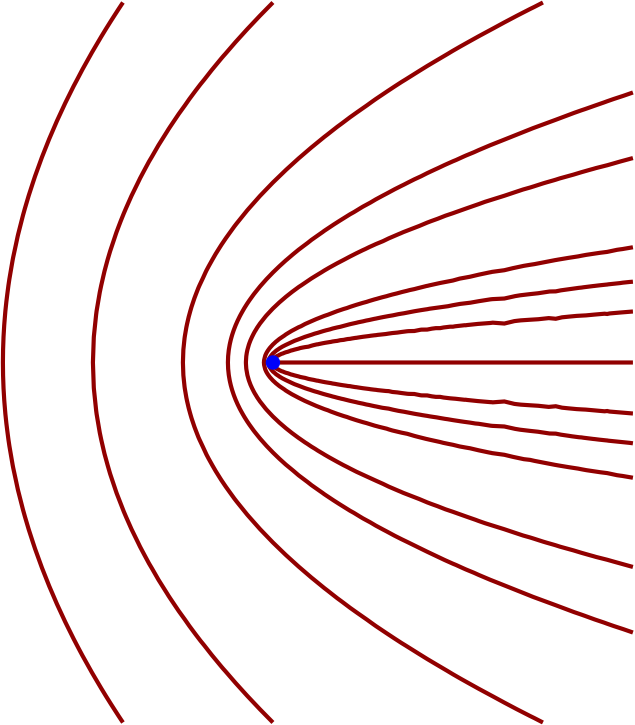}} 
%[]
{\includegraphics[width=0.47\textwidth,height=0.17\textheight]{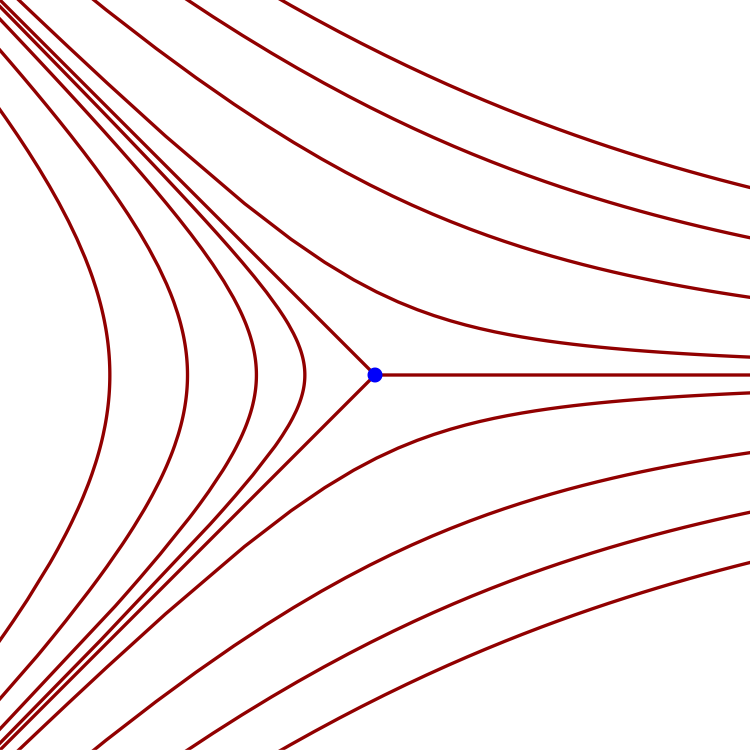}} 
\caption{Topological classification of nondegenerate singularities with respect to $\mbox{SORF}_{max}$ or -min (a) a $ 1 $-pronged (intruding)
point, and (b) a $ 3 $-pronged (separating) point.  See Property \ref{property:singularities}.  Compare to Fig.~\ref{fig:fourquadrant}. }
\label{fig:classify}
\end{figure}

Since by Hypothesis~\ref{hypothesis:smooth} both  $ \phi $ and $ \psi $ are $ {\mathrm{C}}^1 $-smooth in $ \Omega $, their gradients are well-defined on $ \Omega $. 
Nondegeneracy precludes either $ \phi $ or 
$ \psi $ possessing critical points at $ \vec{p} $; thus, we cannot get self-intersections of either $ \phi = 0 $ or $ \psi = 0 $
contours at $ \vec{p} $, have local extrema of $ \phi $ or $ \psi $ at $ \vec{p} $, or have a situation where $ \phi $ or $ \psi $ is constant in an open
neighborhood around $ \vec{p} $.  Nondegeneracy also precludes
$ \phi = 0 $ and $ \psi = 0 $ contours intersecting tangentially at $ \vec{p} $ (although we will be able to make some
remarks about this situation later).  Thus, at nondegenerate points $ \vec{p} $, the curves $ \phi = 0 $ and $ \psi = 0 $ 
intersect transversely.  
We explain in \ref{sec:fourquadrant} how we obtain the following {\em complete} classification for nondegenerate singularities, as illustrated in Fig.~\ref{fig:classify}:

\begin{property}[$ 1 $- and $ 3 $-pronged singularities]
\label{property:singularities}
Let $ \vec{p} \in I $ be a nondegenerate singularity, and let $ \vec{\hat{k}} $ be the unit-vector in the $ + z $-direction
(i.e., `pointing out of the page' for a standard right-handed Cartesian system).  Then,
\begin{itemize}
\item If $ \vec{p} $ is {\em right-handed}, i.e., if 
\begin{equation}
\Big[ \vecnabla \phi \times \vecnabla \psi \Big]_{\vec{p}} \cdot \vec{\hat{k}}  = \mathrm{det} \, \frac{\partial(\phi,\psi)}{\partial(x,y)} \Big|_{\vec{p}} > 0 \, , 
\label{eq:righthanded}
\end{equation}
then $ \vec{p} $ is a {\em 1-pronged singularity} (an `intruding point'), with nearby foliation of both $ f^+ $ and
$ f^- $ topologically equivalent to Fig.~\ref{fig:classify}(a);
and
\item If $ \vec{p} $ is {\em left-handed}, i.e., if 
\begin{equation}
\Big[ \vecnabla \phi \times \vecnabla \psi \Big]_{\vec{p}} \cdot \vec{\hat{k}} = \mathrm{det} \, \frac{\partial(\phi,\psi)}{\partial(x,y)} \Big|_{\vec{p}} < 0 \, , 
\label{eq:lefthanded}
\end{equation}
then $ \vec{p} $ is a {\em 3-pronged singularity} (a `separating point'), with nearby foliation of both $ f^+ $ and 
$ f^- $ topologically equivalent to Fig.~\ref{fig:classify}(b).
\end{itemize}
The intrusions/separations occur in 
opposite directions for the two orthogonal foliations $ f^\pm $.
\end{property}

We use the `$ 1 $-pronged' and `$3 $-pronged' terminology from the theory of singularities of measured foliations \cite{thurston,hubbardmasur}.  We also note that in the case of all singularities being nondegenerate, the
curves on $ \Omega_0 $ may be thought of as a punctured foliation \cite[e.g.]{mosher} on $ \Omega $.  These
two singularities also correspond to the index of the foliation being $ + 1/2 $ and $ - 1/2 $ respectively
(for e.g., see Fig.~1 in \cite{rykken}).  These two topologically distinct singularities serve as the organizing skeleton around which the rest of the SORF smoothly vary.
These topologies have been observed numerically \cite{tchon,farazmand} but
apparently not classified before.  

We have claimed in Property~\ref{property:singularities} that the topology of $ f^- $ is similar to that of $ f^+ $ as illustrated in Fig.~\ref{fig:classify}. To
see why this is so, imagine reflecting these curves about the
vertical line going through $ \vec{p} $.  This generates an orthogonal set of curves, which are the
complementary (orthogonal) foliation.  Thus, $ f^+ $ and $ f^- $ have the {\em same} topology near $ \vec{p} $.  

At the next-order of degeneracy, we will have $ \phi = 0 $ and $\psi = 0 $ contours continuing to be curves,
but now intersecting at $ \vec{p} $ {\em nontangentially}.  In that case, it turns out that Fig.~\ref{fig:degenerate} gives the possible topologies for
$ \mbox{SORF}_{max}$,
which are explained in detail in \ref{sec:fourquadrant}.  
If $ \vec{p} $ is
not an isolated point in $ I $, then many other possibilities exist.
The $ \mbox{SORF}_{min}$ in the mildly degenerate situations in Fig.~\ref{fig:degenerate} represent curves which are orthogonal to
the pictured ones, by Corollary \ref{corollary:orthogonal}.  Their topology will be identical. 

\begin{figure}
\centering
{\includegraphics[width=0.3\textwidth,height=0.21\textheight]{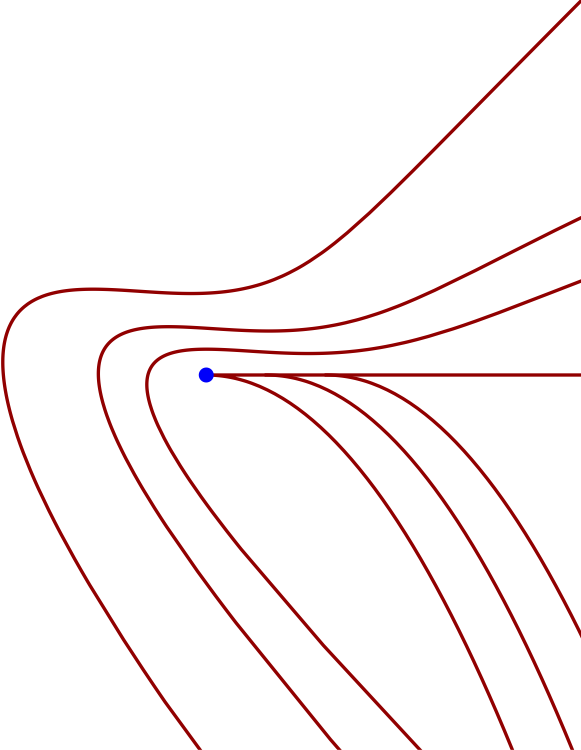}} 
{\includegraphics[width=0.3\textwidth,height=0.21\textheight]{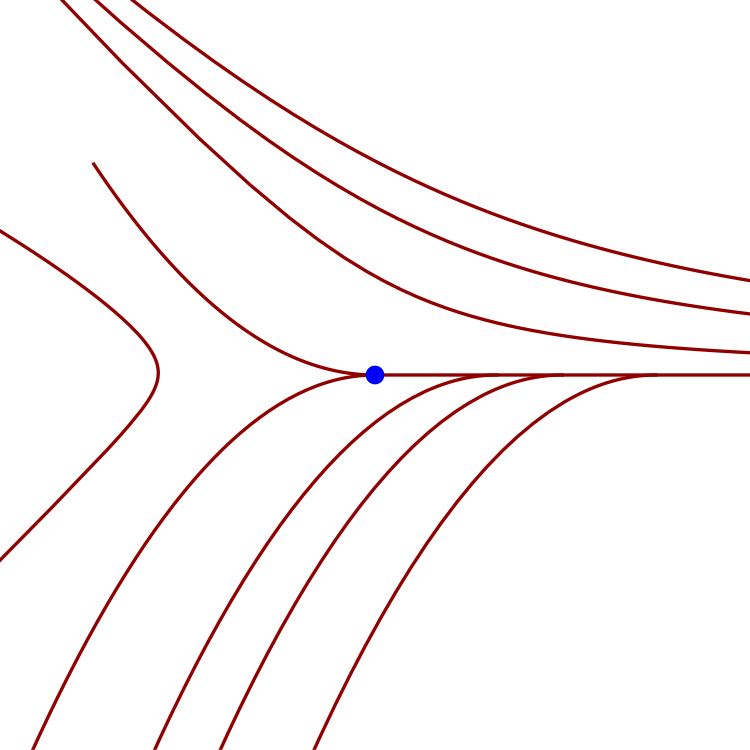}} 
{\includegraphics[width=0.3\textwidth,height=0.21\textheight]{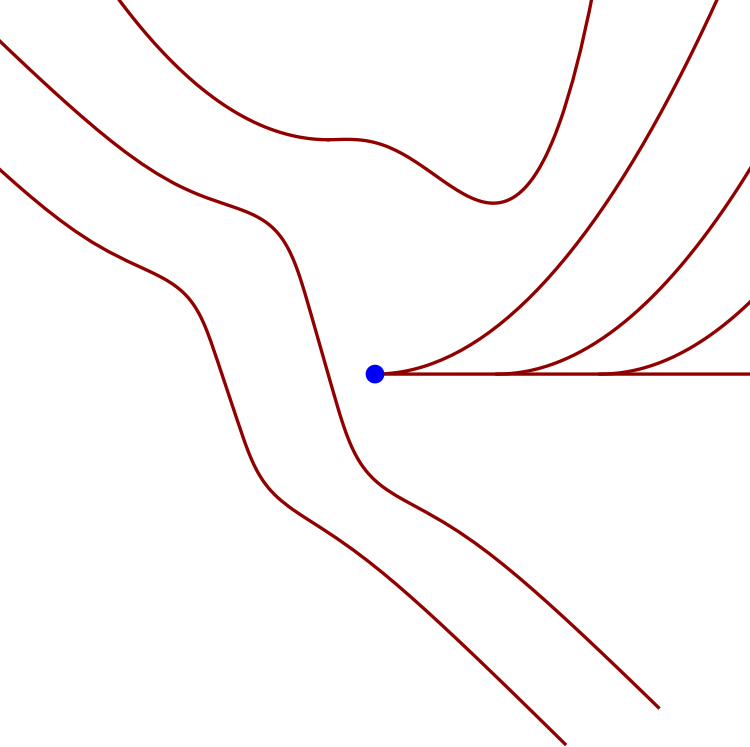}} 
\caption{Some possible topologies for $ \mbox{SORF}_{max}$ near $ \vec{p} $ when transversality is relaxed (see \ref{sec:fourquadrant} for explanations of these structures).}
\label{fig:degenerate}
\end{figure}

%%%%%%%%%%%
\section{Discontinuity in Lyapunov vectors}
\label{sec:discontinuity}

We have determined slope fields $ \theta^+ $ and $\theta^- $ corresponding to maximizing and minimizing
the {\em global} stretching.  By Remark~\ref{remark:local}, maximizing the {\em local} stretching at a point in $ \Omega_0 $ also results 
in an angle corresponding to $ \theta^+ $.  Such local stretching is well-studied;
it is related to the Lyapunov exponent, and the directions are associated with {\em Lyapunov vectors} \cite{wolfesamelson}
or {\em Oseledec spaces} \cite{oseledec}.  Additionally, the direction associated with $ \theta^+ $ can be characterized in 
terms of the eigenvector associated with the larger Cauchy--Green eigenvalue.  See \ref{sec:localstretching}
for a more extensive discussion of these connections.

Here, we analyze the {\em vector fields} associated with $ \theta^\pm $ in some detail, using the behavior in the
$ (\phi,\psi) $-plane introduced in the previous section.  The main observation is that, generically, it is {\em not}
possible to express a $ {\mathrm{C}}^0 $-vector field on the closure of $ \Omega_0  $ from the $ \theta^\pm $ angle fields.  This has implications in
numerically computing curves in the optimal foliations, where we give insight into spurious effects that arise.

The  $ \theta^+ $ field in $ \Omega_0 $ is given by (\ref{eq:thetaplus}).   To determine a curve from the $ \mbox{SORF}_{max}$,
we need to pick an initial point in $ \Omega_0 $, and evolve it according to `the' vector
field generated from $ \theta^+ $.   A simple possibility would be to take the (unit) vector field
\begin{equation}
\vec{w}^+(x,y) := \left( \begin{array}{c} 
\cos \left[ \theta^+(x,y) \right] \\ \sin \left[ \theta^+(x,y) \right] \end{array} \right) \, , 
\label{eq:vector_plus}
\end{equation}
in which $ \theta^+ $ is computed from (\ref{eq:thetaplus}).   In evolving trajectories associated with this 
vector field---i.e., in determining streamlines of (\ref{eq:thetaplus})---one can of course multiply $ \vec{w}^+ $
by a scalar function $ m(x,y) $, which simply changes the parametrization along the trajectory/streamline.
As verified in \ref{sec:localstretching}, (\ref{eq:vector_plus}) is indeed the eigenvector
associated with the larger eigenvalue of the Cauchy--Green tensor at $ (x,y) $, with the understanding that it
can be multiplied by a nonzero scalar. The fact that the
eigenvector at each point is unique, modulo a constant multiple, is of course directly related to these observations.  

Exactly the same arguments hold when attempting to compute the $ \mbox{SORF}_{min}$: from the angle field $ \theta^- $
we can construct the vector field
\begin{equation}
\vec{w}^-(x,y) := \left( \begin{array}{c} 
\cos \left[ \theta^-(x,y) \right] \\ \sin \left[ \theta^-(x,y) \right] \end{array} \right) \, , 
\label{eq:vector_minus}
\end{equation}
where $ \theta^- $ is defined from (\ref{eq:thetaminus}). 

\begin{property}[Generating foliation curves using vector fields]
\label{property:curvegenerate}
If generating a $ \mbox{SORF}_{max}$  or $ \mbox{SORF}_{min}$ curve in $ \Omega_0 $, we can in general find solutions to
\begin{equation}
\renewcommand{\arraystretch}{1.5}
\hspace*{-1cm} \frac{d}{d s} \left( \begin{array}{c} x \\ y \end{array} \right) =  \vec{w} \left( x(s) ,y(s) \right) \quad ; \quad \left( \begin{array}{c} x(0) \\ y(0) \end{array} \right) = \left( \begin{array}{c}
x_0 \\ y_0 \end{array} \right) \, ,
\label{eq:curvegenerate}
\end{equation}
where $ s $ is the parameter along the curve and $ (x_0,y_0) \in \Omega_0 $, and we can choose a {\em Lyapunov vector field} in the form
\begin{equation}
\vec{w}(x,y) = m(x,y) \, \vec{w}^\pm(x,y)
\label{eq:w}
\end{equation}
for a suitable scalar function $ m $.
\end{property}

 If we use $ m \equiv 1 $ on $ \Omega_0 $, the parametrization $ s $ along the trajectory is 
exactly the arclength.
However, more general scalar  functions $ m $ can be used in (\ref{eq:curvegenerate}), reflecting the
fact that the vector fields which generate the foliations are actually {\em direction fields}, and thus
can be multiplied at each point by a scalar.  The only restrictions are (i) $ m $ can never be zero, because if
it is, we introduce a spurious fixed point in the system (\ref{eq:curvegenerate}) which `stops' the curve, and (ii) $ m $ is sufficiently
smooth to ensure that the equation (\ref{eq:curvegenerate}) has unique $ {\mathrm{C}}^1 $-smooth
solutions.     From the perspective of a SORF curve, making a choice of the function
$ m $ simply adjusts the parametrization
along the curve.  Notice that if we flip the sign of $ m $ we would be going along the curve
in the opposite direction. 

\begin{figure}
\centering
\includegraphics[scale=1]{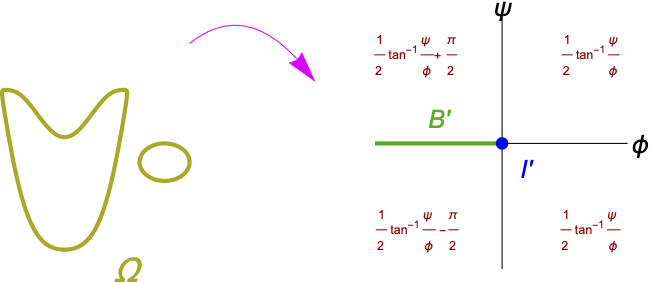}
\caption{The map from $ \Omega $ to $ (\phi,\psi) $-space, illustrating the sets $ I' $ and $ B' $ to
which the sets $ I $ and $ B $ map.  In red, we have stated the
value of the field $ \theta^+ $ in (\ref{eq:thetaplus}) in each quadrant.}
\label{fig:phipsimap}
\end{figure}

To understand the generation of curves from (\ref{eq:w}), it helps to think of the mapping from $ \Omega $ to  $ (\phi,\psi) $-space,
illustrated in Fig.~\ref{fig:phipsimap}.   We have already characterized an important subset of $ \Omega $ in relation to this
mapping: the isotropic set $ I $ is the kernel of this mapping (by Lemma~\ref{lemma:puncture}).  Its image is denoted
by $ I' $, the origin in $ (\phi,\psi) $-space.    

Another important set that we require is

\begin{definition}[Branch cut]
\label{definition:branchcut}
The {\em branch cut} $ B $  is the set of points $ (x,y) \in \Omega $ such that
\begin{equation}
B := \left\{ (x,y) \in \Omega \, : \, \phi(x,y) < 0 ~~{\mathrm{and}}~~ \psi(x,y) = 0 \, \right\} \, .
\label{eq:branchcut}
\end{equation}
\end{definition}

The image $ B' $ of the branch cut is also shown in Fig.~\ref{fig:phipsimap} as the negative $ \phi $-axis.  In each of the four quadrants of Fig.~\ref{fig:phipsimap}, we have carefully
stated the value of the $ \theta^+ $ field in terms of the {\em standard} inverse tangent function.
We focus here near a nondegenerate
singularity $ \vec{p} $, where the $ \phi = 0 $ and
$ \psi = 0 $ contours must cross $ \vec{p} $ transversely, given that the Jacobian determinant of $ (\phi,\psi) $ with
respect to $ (x,y) $ is nonzero. The axis-crossings in Fig.~\ref{fig:phipsimap} will have the same topology as
these contours if the determinant is positive (the map is orientation-preserving). 

The relevant set $ B $ in 
$ \Omega_0 $, near $ \vec{p} $, must therefore have the structure as seen in Fig.~\ref{fig:branchcut}(a).   Consider a  small circle
around $ \vec{p} $ as drawn in Fig.~\ref{fig:branchcut}(a), and indicated via arrows the directions of the vector field
$ \vec{w}^+ $ along it.  The reasons for these directions stems directly from Fig.~\ref{fig:phipsimap};
we need to take the cosine (for the $ x $-component) and the sine (for the $ y $-component) of the angle field
defined therein.  While $ \vec{w}^+ $ must vary smoothly along the circle, it exhibits a discontinuity across the
branch cut $ B $, because the angle has rotated around from $ - \pi/2 $ to $ +\pi/2 $.  Clearly, the same behavior
occurs for left-handed $ \vec{p} $: in this case we need to consider Fig.~\ref{fig:phipsimap} with the 
$ \psi $-axis flipped (this orientation-reversing case is indeed pictured in Fig.~\ref{fig:fourquadrant}(b)). Once again, it is the $ \phi_- $
axis to which the branch cut $ B \in \Omega_0 $ gets mapped.  The intuition of Fig.~\ref{fig:branchcut} gives
us a {\em theoretical} issue related to using a vector field to find curves:

\begin{figure}
\centering
{\includegraphics[width=0.4\textwidth,height=0.17\textheight]{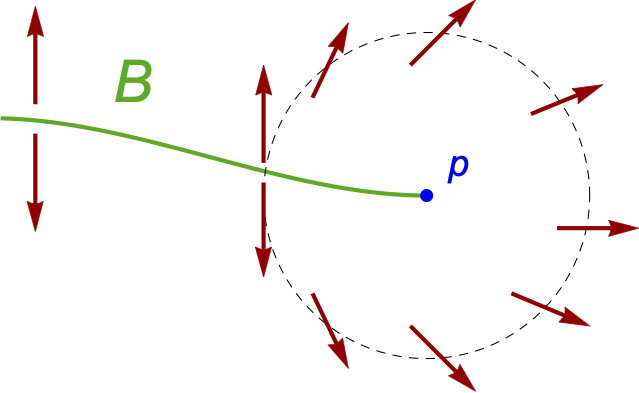}} \hfill 
{\includegraphics[width=0.4\textwidth,height=0.17\textheight]{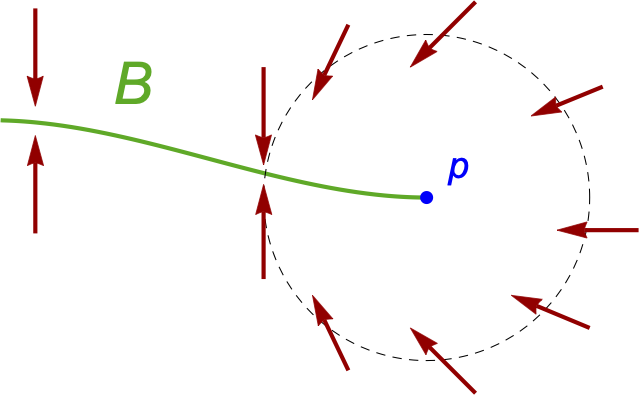}} 
\caption{Vector field of (\ref{eq:curvegenerate}) using $ \vec{w}^+ $, 
near a nondegenerate singularity $ \vec{p} $, with the
branch cut $ B $ shown in green: (a) right-handed $ \vec{p} $ and (b) left-handed $ \vec{p} $.}
\label{fig:branchcut}
\end{figure}

\begin{theorem}[Impossibility of continuous Lyapunov vector field]
\label{theorem:nonsmooth}
If there exists at least one nondegenerate singularity $ \vec{p} \in \Omega $, then no
nontrivial scalar
function $ m $ in (\ref{eq:curvegenerate}) exists such that the right-hand side (i.e., vector field associated with
the angle field $ \theta^+ $) is a $ {\mathrm{C}}^0 $-smooth nonzero vector
field in $ \Omega_0 $.  The same conclusion holds for vector fields generated from $ \theta^- $.
\end{theorem}

\begin{proof}
See \ref{sec:nonsmooth}.
\end{proof}

%%%%%%%%%%%%%%%%%%
\section{Computational issues of finding foliations}
\label{sec:computing_foliations}

In the previous section, we have outlined a {\em theoretical} concern in defining a vector field for computing optimal foliations. 
We show here related {\em numerical} issues which emerge when attempting to compute foliating curves.  

First, we remark that using a vector field to generate curves of streamlines of eigenvector fields of a tensor---which
as seen here are equivalent to $ \mbox{SORF}_{max}$ and $ \mbox{SORF}_{min} $ curves---is standard practice.  
Numerical issues in doing so have been observed previously, and {\em ad hoc} remedies proposed: 
\begin{itemize}
\item In generating trajectories following `smooth' fields from grid-based data, one suggested approach is to keep checking the direction of the vector field within each cell a trajectory ventures into, and
then flip the vector field at the bounding gridpoints to all be in the same direction before interpolating \cite{farazmand}.
\item In dealing with points at which the eigenvector field is not defined, an approach is to mollify the field by multiplying with a sufficiently smooth field which is zero
at such points (e.g., the square of the difference in the two eigenvalues \cite{tchon}). 
\end{itemize}
Our Theorem~\ref{theorem:nonsmooth} gives explicit insights into the nature of both
these issues.   Both {\em ad hoc} numerical methods relate to choosing the function $ m $
(respectively as $ \pm 1 $, or a smooth scalar field which is zero at singularities).  In either case, {\em actual} behavior near the singularities gets blurred by this process.

The branch cut near singularities also leads to more subtle---and apparently hitherto unidentified in the literature
of following streamlines of tensor fields---issues when performing numerical computations.  In \ref{sec:branchcut}, we explain why the following occur.

\begin{property}[Numerical computation of optimal foliations using vector fields]
\label{property:branchcut}
{\em 
Suppose we numerically compute a $ \mbox{SORF}_{max}$ (resp.\ $ \mbox{SORF}_{min}$) curve by using (\ref{eq:curvegenerate}) with $ m = 1 $
and the vector field $ \vec{w}^+ $ (resp.\ $ \vec{w}^- $), 
by allowing the parameter $ s $ to evolve in both directions.  Then
\begin{itemize}
\item[(a)] $ \mbox{SORF}_{max}$ curves will not cross a one-dimensional part of $ B $ vertically, and may also veer along $ B $ even though $ B $ may not be 
a genuine $ \mbox{SORF}_{max}$ curve;
\item[(b)] $ \mbox{SORF}_{min}$ curves will not cross a one-dimensional part of $ B $ horizontally, and may also veer along $ B $ even though $ B $ may not be 
a genuine $ \mbox{SORF}_{min}$ curve.
\end{itemize}
}
\end{property}

These problems are akin to branch splitting issues arising when applying curve continuation methods
in instances such as bifurcations \cite{auto}.
Is it possible to choose a function $ m $ which is not identically $ 1 $ to remove these difficulties?  The proof of Theorem~\ref{theorem:nonsmooth} tells us that the answer is no.  Either the branch cut gets moved to a different
curve connected to $ \vec{p} $ across which there is a similar discontinuity, or it gets converted to a curve which has spurious fixed points (i.e., a center manifold curve) because $ m = 0 $ on it.  In either case, the numerical evaluation
will give problems.  

Thus, there are several numerical issues in computing foliations using the vector fields $ \vec{w}^\pm $.  
 Lemma~\ref{lemma:theta} suggests a straightfoward alternative method for numerically computing such curves in generic
 situations, while systematically avoiding {\em all} these issues.  Let
 \begin{eqnarray*}
\Phi_- & := & \left\{ (x,y): \phi(x,y) < 0~~{\mathrm{and}}~~\psi(x,y) = 0 \right\} \quad {\mathrm{and}} \\
\Phi_+ & := & \left\{ (x,y): \phi(x,y) > 0~~{\mathrm{and}}~~\psi(x,y) = 0 \right\} \, ;
\end{eqnarray*} 
these are points mapping to the `negative $ \phi $-axis'
and the `positive $ \phi $-axis' (see Figs.~\ref{fig:phipsimap} and
\ref{fig:fourquadrant}), and we also note that $ \Phi_- = B $.  In seeking the maximizing foliation, 
 we define on $ \Omega_0 \setminus \Phi_- $,
\begin{equation}
h^+(x,y) =  \left\{ \begin{array}{ll} 
\frac{- \phi(x,y) + \sqrt{\phi^2(x,y) + \psi^2(x,y)}}{\psi(x,y)} & ~~ {\mathrm{if}} ~~ \psi(x,y) \ne 0 \\
0 & ~~{\mathrm{if}} ~~ \psi(x,y) = 0 ~{\mathrm{and}}~ \phi(x,y) > 0 
\end{array} \right. \, .
\label{eq:fplus}
\end{equation}
This is essentially the function $ \tan \theta^+ $ as defined in (\ref{eq:thetaplus2}), and is $ {\mathrm{C}}^1 $ in
$ \Omega_0 \setminus  \Phi_- $ because of Remark~\ref{remark:removable}.
The reason for not defining $ h^+ $ on $ \Phi_- $ is because 
the relevant tangent line becomes vertical.  Hence we define its reciprocal, $ {\mathrm{C}}^1 $ on $ \Omega_0 \setminus \Phi_+ $,  by 
\begin{equation}
\hdown^+ (x,y):=  \left\{ \begin{array}{ll}
\frac{\phi(x,y) + \sqrt{\phi^2(x,y) + \psi^2(x,y)}}{\psi(x,y)} 
& ~~ {\mathrm{if}} ~~ \psi(x,y) \ne 0 \\
0 & ~~{\mathrm{if}} ~~ \psi(x,y) = 0 ~{\mathrm{and}}~ \phi(x,y) < 0 
\end{array} \right. \, .
\label{eq:gplus}
\end{equation}
The minimizing foliation is associated with the angle field $ \theta^- $.  Thus we define on 
$ \Omega_0 \setminus \Phi_+ $, 
\begin{equation}
h^-(x,y) :=  \left\{ \begin{array}{ll}
\frac{- \phi(x,y) - \sqrt{\phi^2(x,y) + \psi^2(x,y)}}{\psi(x,y)} 
& ~~ {\mathrm{if}} ~~ \psi(x,y) \ne 0 \\
0 & ~~{\mathrm{if}} ~~ \psi(x,y) = 0 ~{\mathrm{and}}~ \phi(x,y) < 0 
\end{array} \right. \, , 
\label{eq:fminus}
\end{equation}
which gives the slope field associated with $ \theta^- $, and on $ \Omega_0 \setminus \Phi_- $ its reciprocal
\begin{equation}
\hdown^-(x,y) := \left\{ \begin{array}{ll} \frac{\phi(x,y) - \sqrt{\phi^2(x,y) + \psi^2(x,y)}}{\psi(x,y)} 
& ~~ {\mathrm{if}} ~~ \psi(x,y) \ne 0 \\
0 & ~~{\mathrm{if}} ~~ \psi(x,y) = 0 ~{\mathrm{and}}~ \phi(x,y) > 0 
\end{array} \right. \, .
\label{eq:gminus}
\end{equation}

\begin{property}[Foliations as integral curves]
\label{property:integral}
{\em Within $ \Omega_0 $, a $ \mbox{SORF}_{max}$ curve can be determined by taking an initial point $ (x_0,y_0) $ and then
numerically following 
\begin{equation}
\frac{d y}{d x} = h^+(x,y) ~~~{\mathrm{if}}~~ \left| h^+(x,y) \right| \le 1 ~~~~{\mathrm{and}}~~~~
\frac{d x}{d y} = \, \hdown^+(x,y)~~~{\mathrm{if~else}} \, ,
\label{eq:cogs}
\end{equation}
where we keep switching between the equations depending on the size of $ \left| h^+ \right| $.
This generates a sequence $ (x_i,y_i) $ to numerically approximate an integral curve.  
Similarly, a $ \mbox{SORF}_{min}$ curve can be determined in $ \Omega_0 $ as integral curves of
\begin{equation}
\frac{d y}{d x} = h^-(x,y) ~~~{\mathrm{if}}~~ \left| h^-(x,y) \right| \le 1 ~~~~{\mathrm{and}}~~~~
\frac{d x}{d y} = \, \hdown^-(x,y)~~~{\mathrm{if~else}} \, .
\label{eq:cols}
\end{equation}
}
\end{property}

Property~\ref{property:integral} is an attractive alternative which avoids issues related to the branch cut and vector field discontinuities.  Moreover, it is {\em directly expressed} in terms of the functions $ \phi $ and $ \psi $ via the
straightforward definitions of $ h^\pm $ and $ \hdown^\pm $.  The switching between the $ d y/dx $ and $ dx/dy $ forms
avoids the infinite slopes which may result if only one of these forms is used.  Thus, we can follow a particular curve
as it meanders around $ \Omega_0 $, having vertical and horizontal tangents, and also crossing branch cuts,
with no problem.

%%%%%%%%%%%%%%%%%
\section{Numerical examples of optimal foliations}
\label{sec:numerical}

We will demonstrate applications of the theory to several maps $ \vec{F} $, generated from several applications
of discrete maps, and from sampling flows driven by unsteady velocities.   The examples include situations which are highly disordered
(e.g., maps known to be chaotic under repeated iterations, flows known to possess chaos over infinite times).
Moreover, the maps $ \vec{F} $ need not be area-preserving.

In order to retain sufficient resolution to view relevant features in the many subfigures that we present in this
Section, we will dispense
with axes labels when these are self-evident: $ x $ will be the horizontal axis and $ y $ the vertical as per standard
convention.

%%%%%%%%%%%%%%%%%%
\subsection{H\'enon map}
\label{sec:henon}

\begin{figure}
\centering
 {\includegraphics[width=0.47\textwidth,height=0.21\textheight]{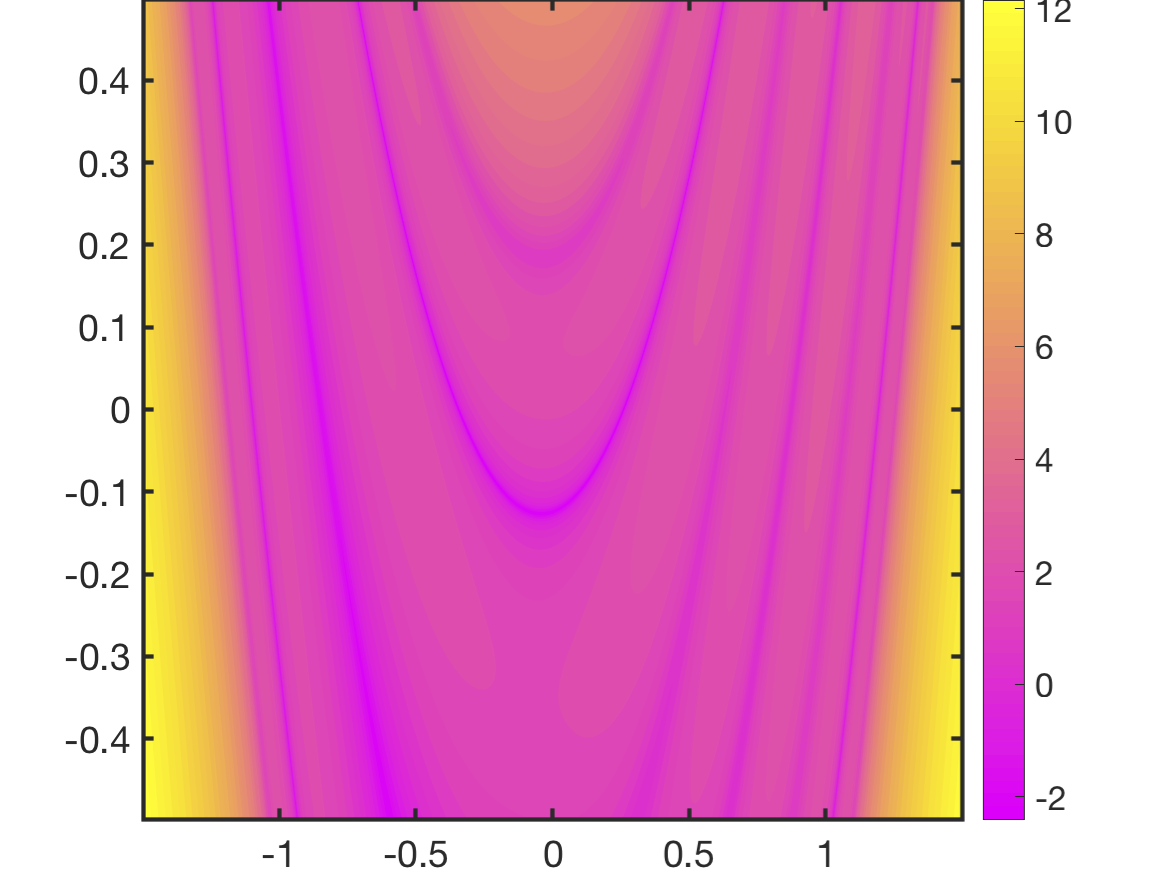}}
 {\includegraphics[width=0.47\textwidth,height=0.21\textheight]{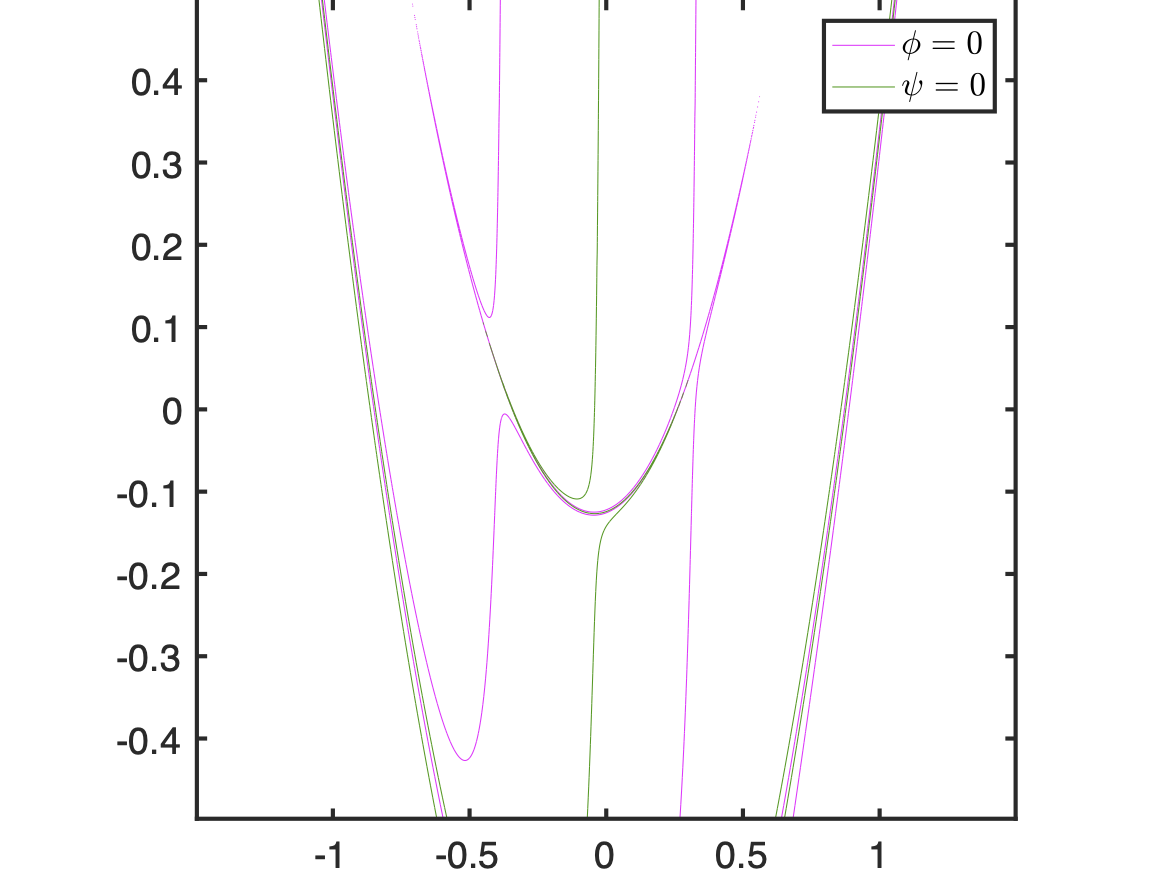}}

\vspace*{0.2cm}
 {\includegraphics[width=0.47\textwidth,height=0.21\textheight]{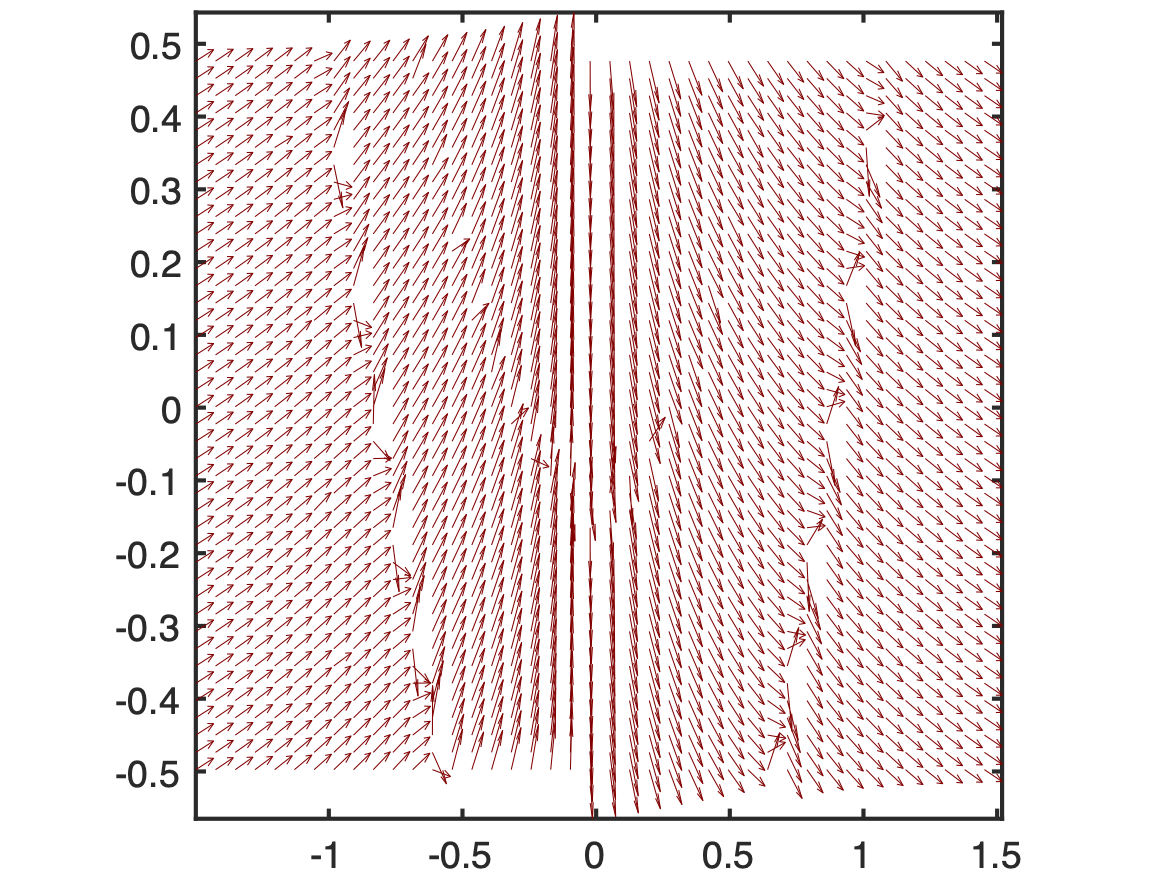}} 
 {\includegraphics[width=0.47\textwidth,height=0.21\textheight]{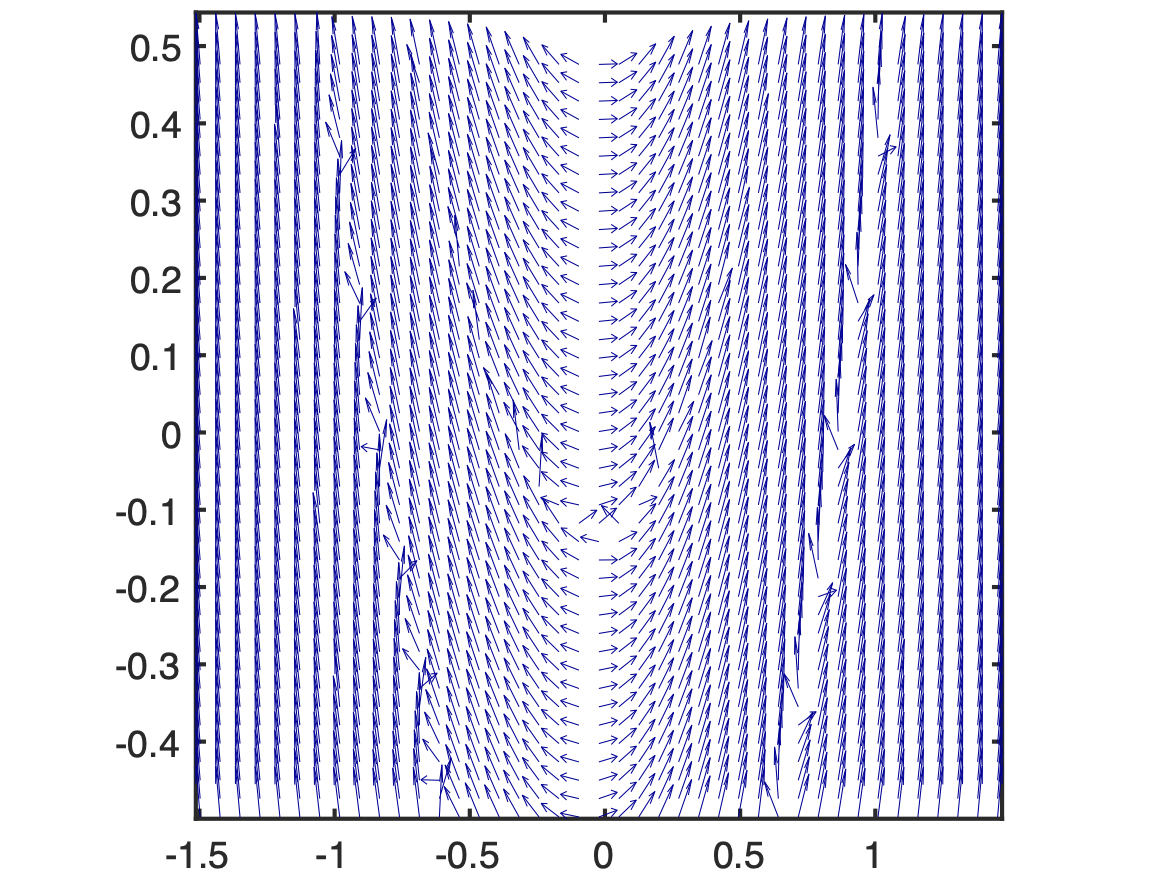}} 

\vspace*{0.2cm}
 {\includegraphics[width=0.47\textwidth,height=0.21\textheight]{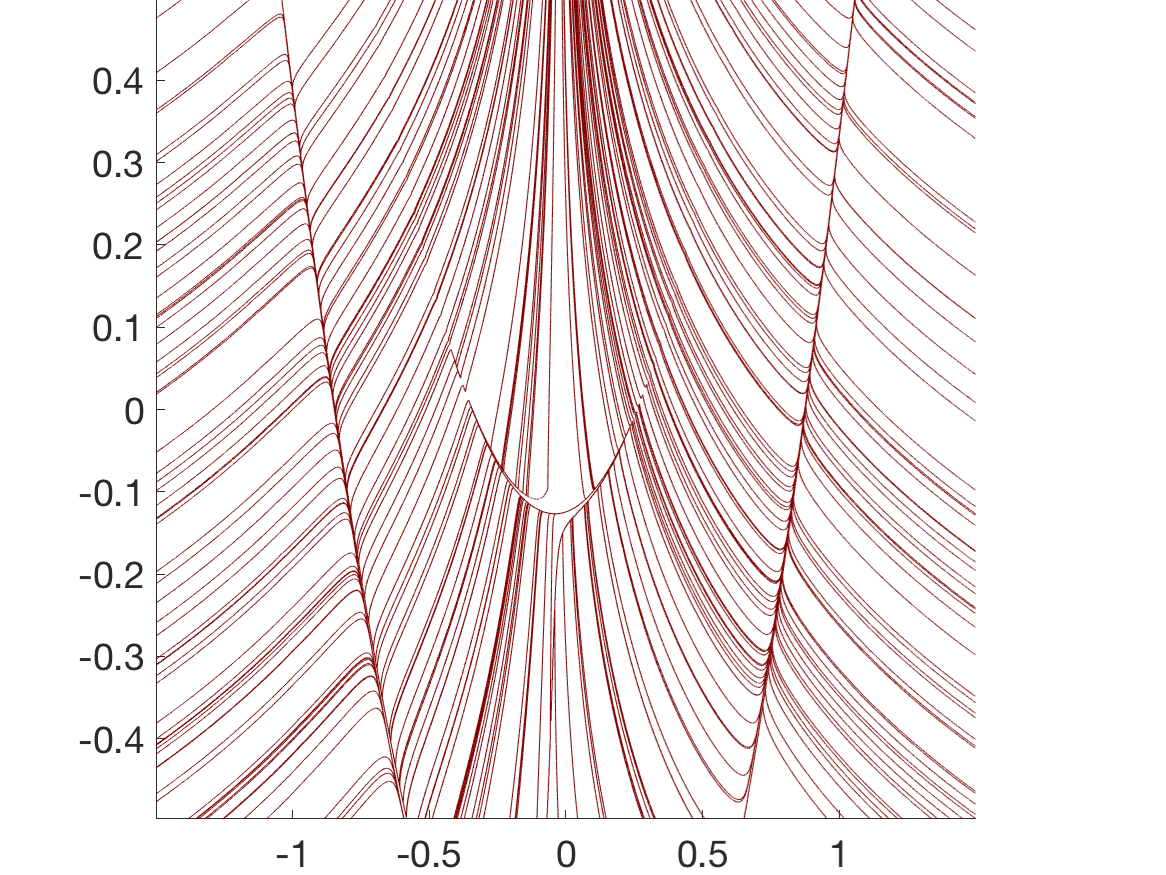}} 
 {\includegraphics[width=0.47\textwidth,height=0.21\textheight]{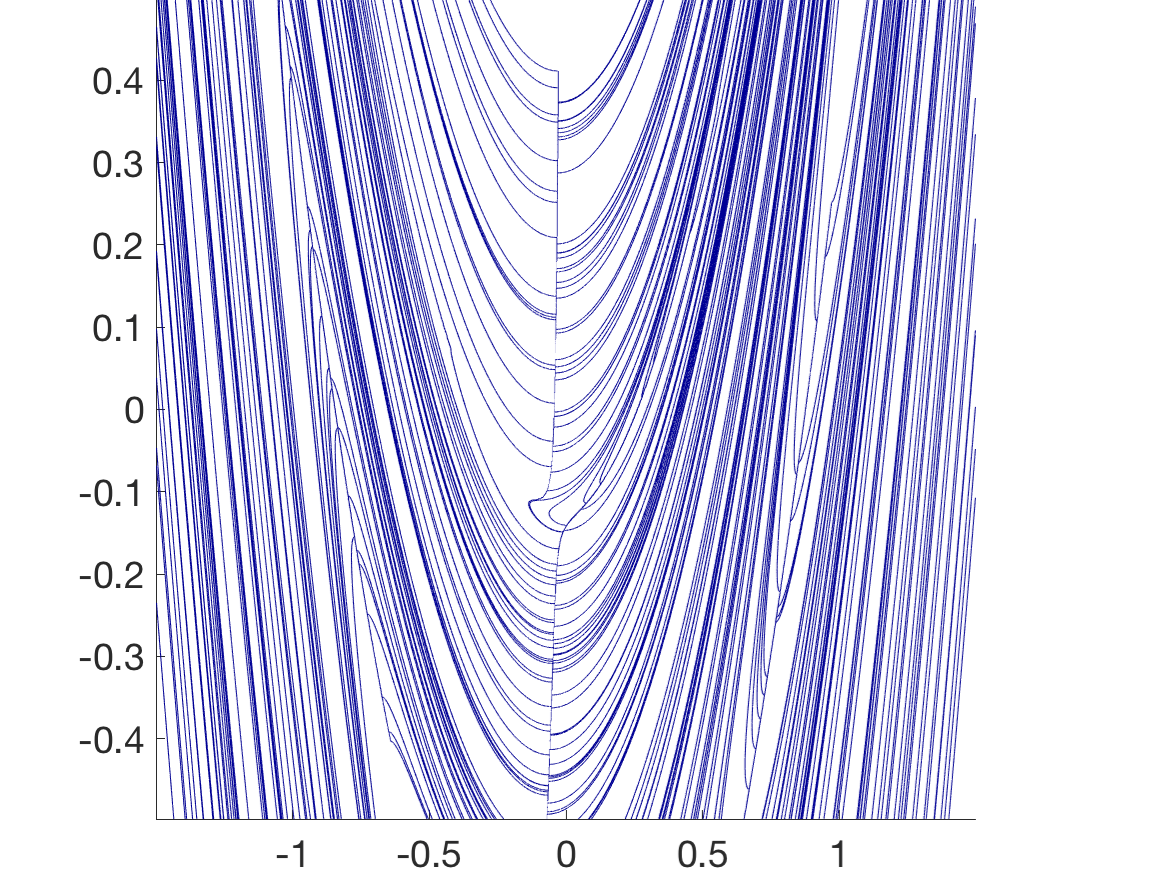}} 

\vspace*{0.2cm}
 {\includegraphics[width=0.47\textwidth,height=0.21\textheight]{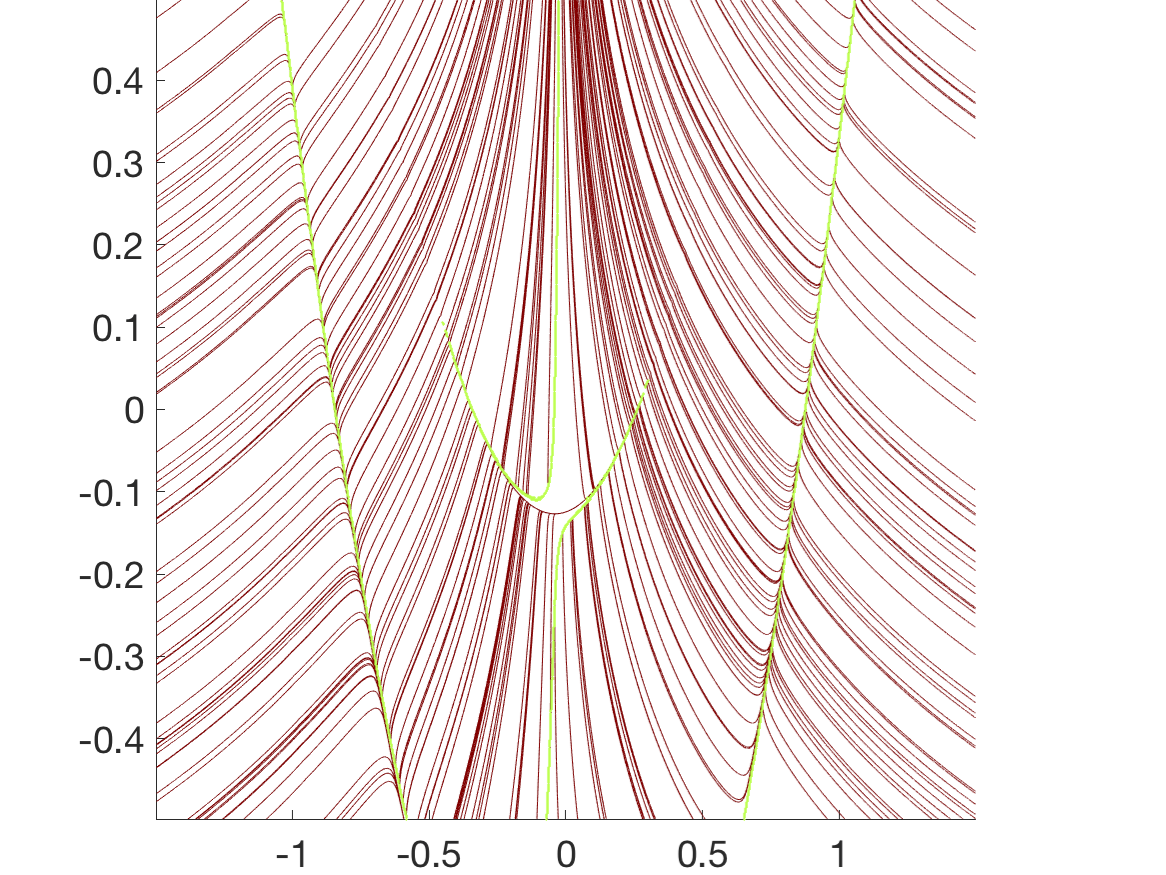}} 
 {\includegraphics[width=0.47\textwidth,height=0.21\textheight]{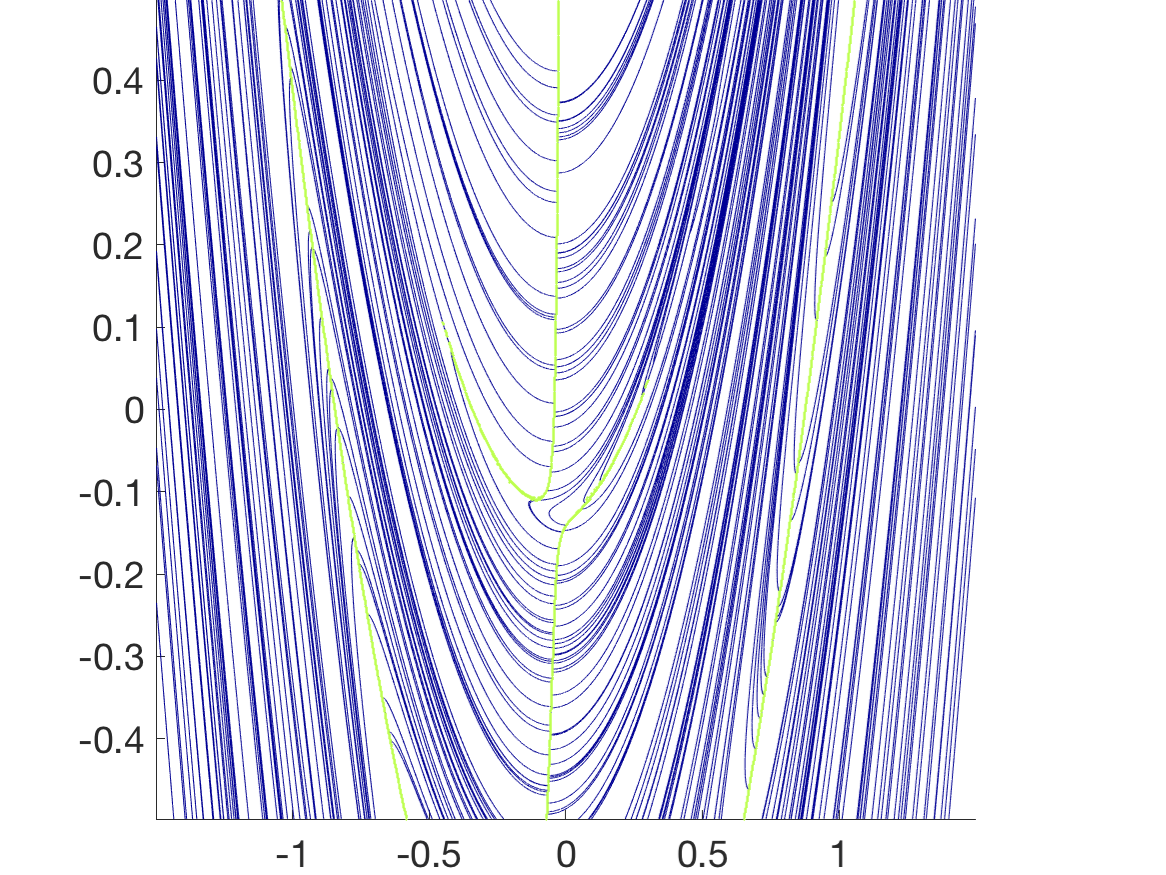}} 

\caption{Optimal foliation computations for $ \vec{F} = \mathfrak{H}^4 $: (a) the logarithm of the maximum
stretching field $ \Lambda_+ $, (b) zero contours of $ \phi $ and $ \psi $, (c)
vector field $ \vec{w}^+  $generated from (\ref{eq:vector_plus}), (d) vector field $ \vec{w}^- $ generated from (\ref{eq:vector_minus}), 
(e) $ \mbox{SORF}_{max}$ by implementing vector field in (c), (f) $ \mbox{SORF}_{min}$ by implementing vector field in (d), (g) 
$ \mbox{SORF}_{max}$ with
branch cut (green), (h) $ \mbox{SORF}_{min}$ with branch cut (green).}
\label{fig:henon}
\end{figure}

As our first example, consider the H\'{e}non map, which is defined by \cite{henon}
\[
\mathfrak{H}(x,y) = \left( \begin{array}{c} 1 - a x^2 + y \\ b x \end{array} \right) 
\]
on $ \Omega = \R^2 $, and where we make the classical parameter choices $ a = 1.4 $ and $ b = 0.3 $.  We choose $ \vec{F} $
to be four iterations of the H\'enon map, i.e., $ \vec{F} = \mathfrak{H}^4 $. Fig.~\ref{fig:henon} demonstrates
the computed foliations and related graphs.  The stretching field $ \Lambda^+ $ is first displayed in Fig.~\ref{fig:henon}(a).  In Fig.~\ref{fig:henon}(b), we show the zero contours of
$ \phi $ and $ \psi $.  In this case, there are no nice transversalities.  Indeed, there are several regions of
almost tangencies, and the fact that several of the zero contours almost coincide in the two outer streaks in the figure, indicate that degenerate foliations are to be expected in their vicinity.  The `squashing together'
that is occurring here is because we are at an intermediate stage in which initial conditions are gradually
collapsing to the H\'enon attractor.  The vector fields $ \vec{w}^\pm $,
computed using (\ref{eq:vector_plus}) and (\ref{eq:vector_minus}) and shown in Figs.~\ref{fig:henon}(c,d)
display discontinuities, which impact the computation of the SORF curves in (e) and (f).  These are 
obtained by seeding 300 initial locations randomly in the domain, and then computing streamlines
generated from (\ref{eq:curvegenerate}) with $ m = 1 $
in forward, as well as backward, $ s $.  Since the $ \phi $ and $ \psi $ fields have large variations at
small spatial scales because of the chaotic nature of the map, finding the branch cut
$ B $ (where where $ \psi = 0 $ and $ \phi < 0 $) as obtained from (\ref{eq:branchcut}) requires care.  We
assess each gridpoint, and color it in
(in green) if it has a different sign of $ \psi $ in comparison to any of its four nearest neighbors, and the
$ \phi $ value at this point is negative.  The lowermost panel overlays the (green) set $ B $ on the SORF curves,
indicating why some of the apparent behavior in (e) and (f) is not representative of the true foliation; the center vertical line in (f), for
example, occurs because of Property~\ref{property:branchcut}(b), while
the $ \mbox{SORF}_{max}$ (resp.\ $ \mbox{SORF}_{min}$) curves stop abruptly on $ B $ if crossing vertically (resp.\ horizontally).

\begin{figure}
\centering

 {\includegraphics[scale=0.2]{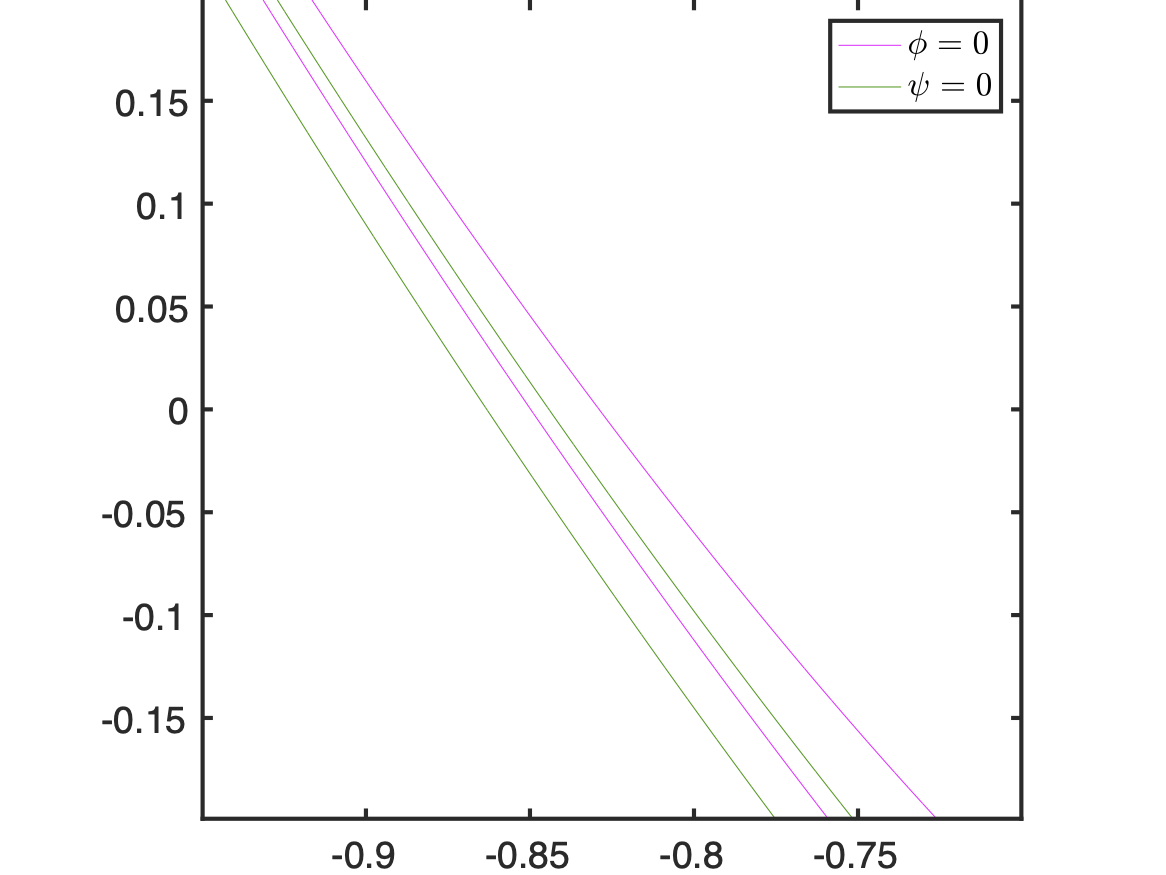}}
 {\includegraphics[scale=0.2]{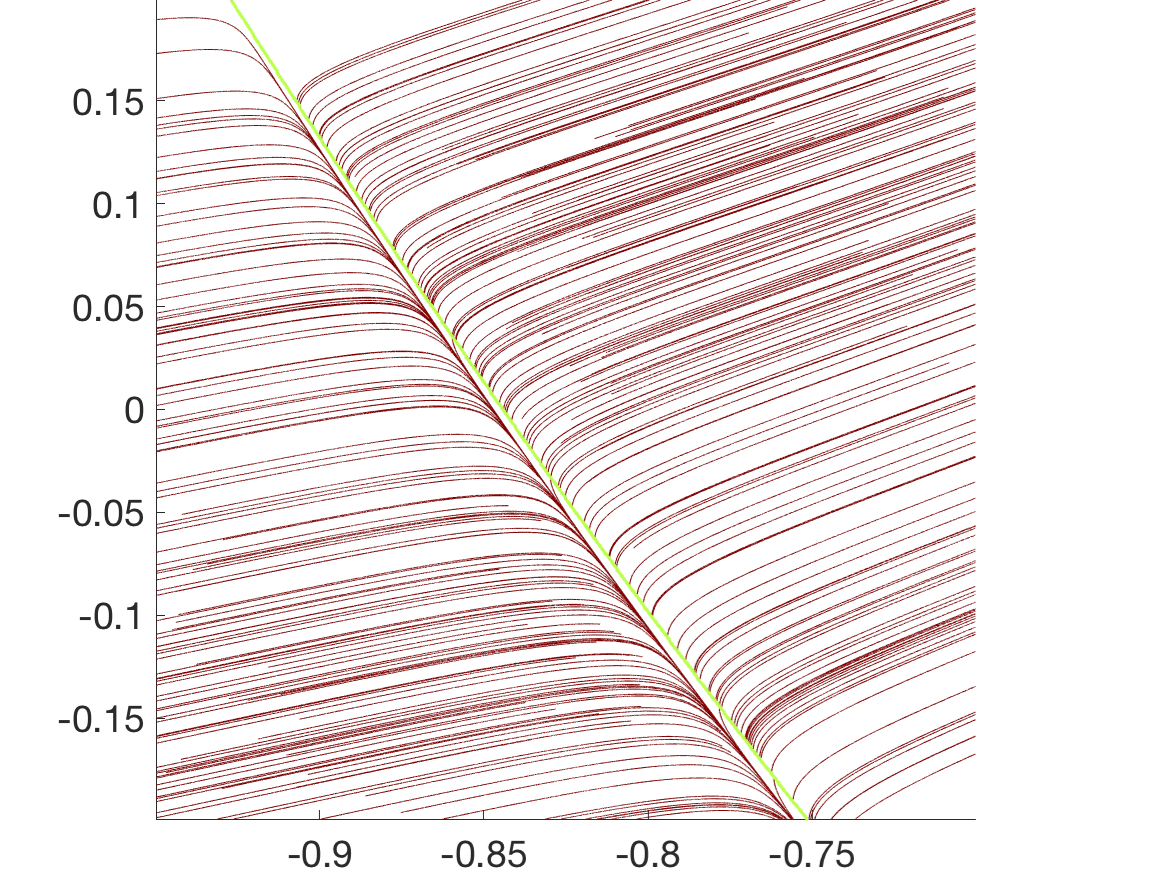}}
 {\includegraphics[scale=0.2]{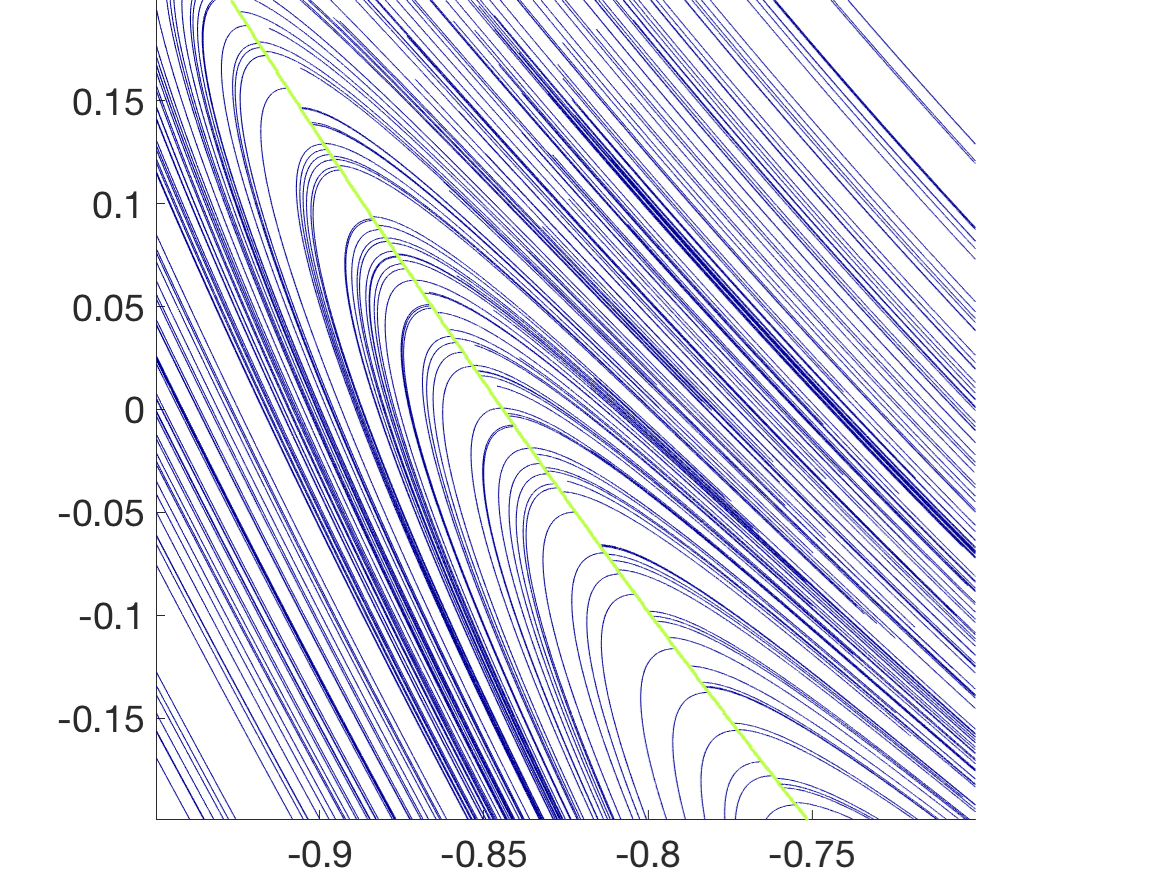}}

\caption{Zooming in to an area associated with the map $ \vec{F} = \mathfrak{H}^4 $ (a) 
the zero contours of $ \phi $ and $ \psi $, (b) the $ \mbox{SORF}_{max}$, and (c) the $ \mbox{SORF}_{min}$.}
\label{fig:henon_zoom}
\end{figure}

On the other hand, Fig.~\ref{fig:henon}(b) indicates that the zero contours of $ \phi $ and $\psi $ almost
coincide on two curves: `outer' and `inner' parabolic shapes.  These are also identified as part of the branch cut set $ B $ because $ \psi \approx 0 $
and $ \phi $ is slightly negative here. These curves are `almost' a curve of $ I $, and we see accumulation of
$\mbox{SORF}_{max}$ curves towards these, indicating---at this level of resolution---potential degeneracy of the foliation.  We zoom in to this in Fig.~\ref{fig:henon_zoom}.  In conjunction with the explanations in Fig.~\ref{fig:fourquadrant},
what occurs here is that the inner green line in Fig.~\ref{fig:henon_zoom}(a) must have a slope field which
is $ - \pi/2 $  (it is in $ \Phi_- = B $ with respect to Fig.~\ref{fig:henon_zoom}), while on the inner pink line it should be $ - \pi/4 $ (corresponding to $ \Psi_- $ in 
Fig.~\ref{fig:fourquadrant}(a)).  The extreme closeness of the contours means that a very sharp change in direction
must be achieved in a tiny region, which then visually appears as a form of degeneracy. 

This example highlights an important computational issue which is very general: even though relevant foliations will exist, in order to
resolve them, one needs a spatial resolution which can resolve the spatial changes in the $ \phi $ and $ \psi $
fields.

%%%%%%%%%%%%%%%%%%
\subsection{Double-gyre flow}
\label{sec:doublegyre}

As an example of when $ \vec{F} $ is generated from a finite-time flow, let us consider the flow map from time
$ t = 0 $ to $ 2 $ generated from the differential equation
\begin{equation}
\renewcommand{\arraystretch}{1.6}
\frac{d}{d t} \left( \begin{array}{c} x \\ y \end{array} \right) = \left( \begin{array}{l}
 - \pi A  \sin \left[ \pi g(x,t) \right] \cos \left[ \pi y \right] \\
 \pi A \cos \left[ \pi g(x,t) \right] \sin \left[ \pi y \right] \frac{\partial g}{\partial x}(x,t)
 \end{array} \right) \, ,
\label{eq:doublegyre}
\end{equation}
in which $ g(x,t) := \eps \sin \left( \omega t \right)  x^2 + \left[ 1 - 2 \eps \sin \left( \omega t \right) \right] x $ 
and $ \Omega = (0,2) \times (0,1) $.   This is the well-studied double-gyre model \cite{shadden}, but we exclude
the boundary of the domain.  We use the parameter values $ A = 1 $, $ \omega = 2 \pi $ and $ \eps =
0.1 $, and the optimal reduced foliations are demonstrate in Fig.~\ref{fig:doublegyre}.

\begin{figure}
\centering
 {\includegraphics[scale=0.32]{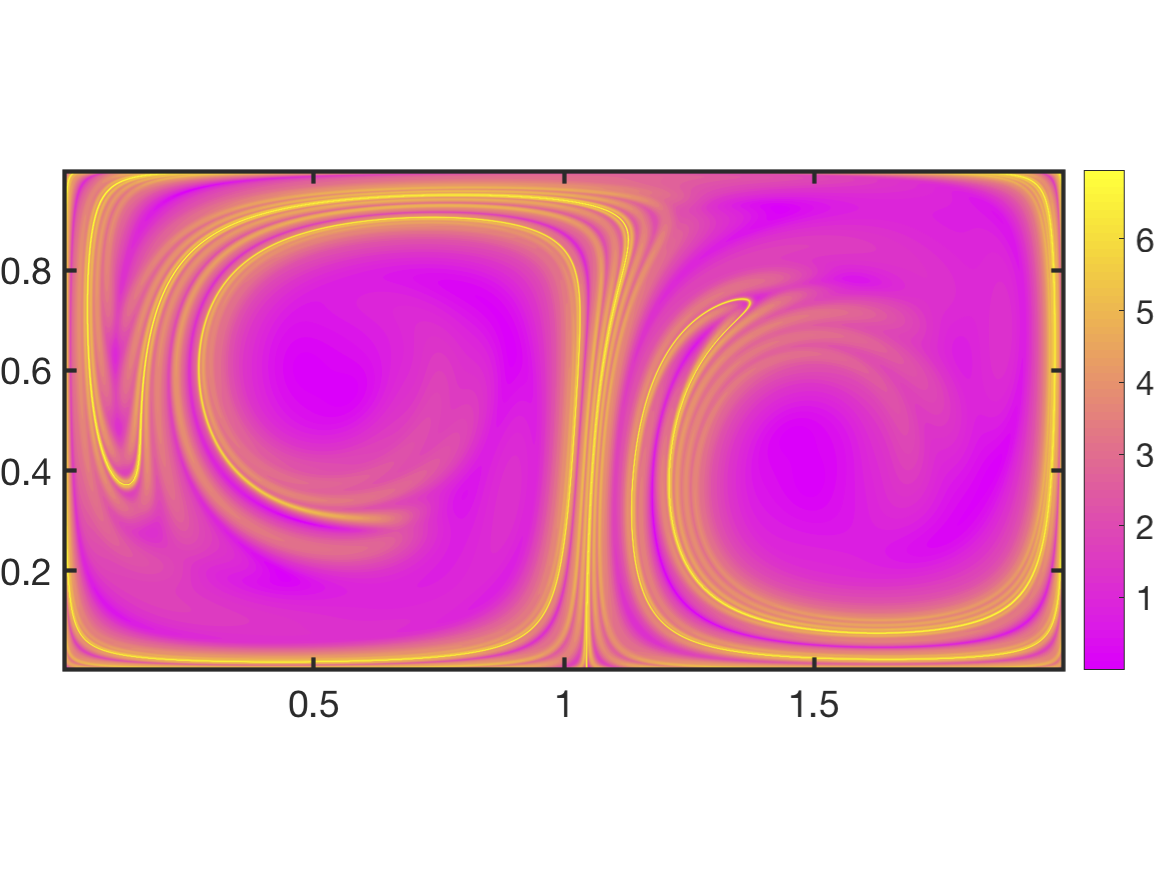}}
 {\includegraphics[scale=0.3]{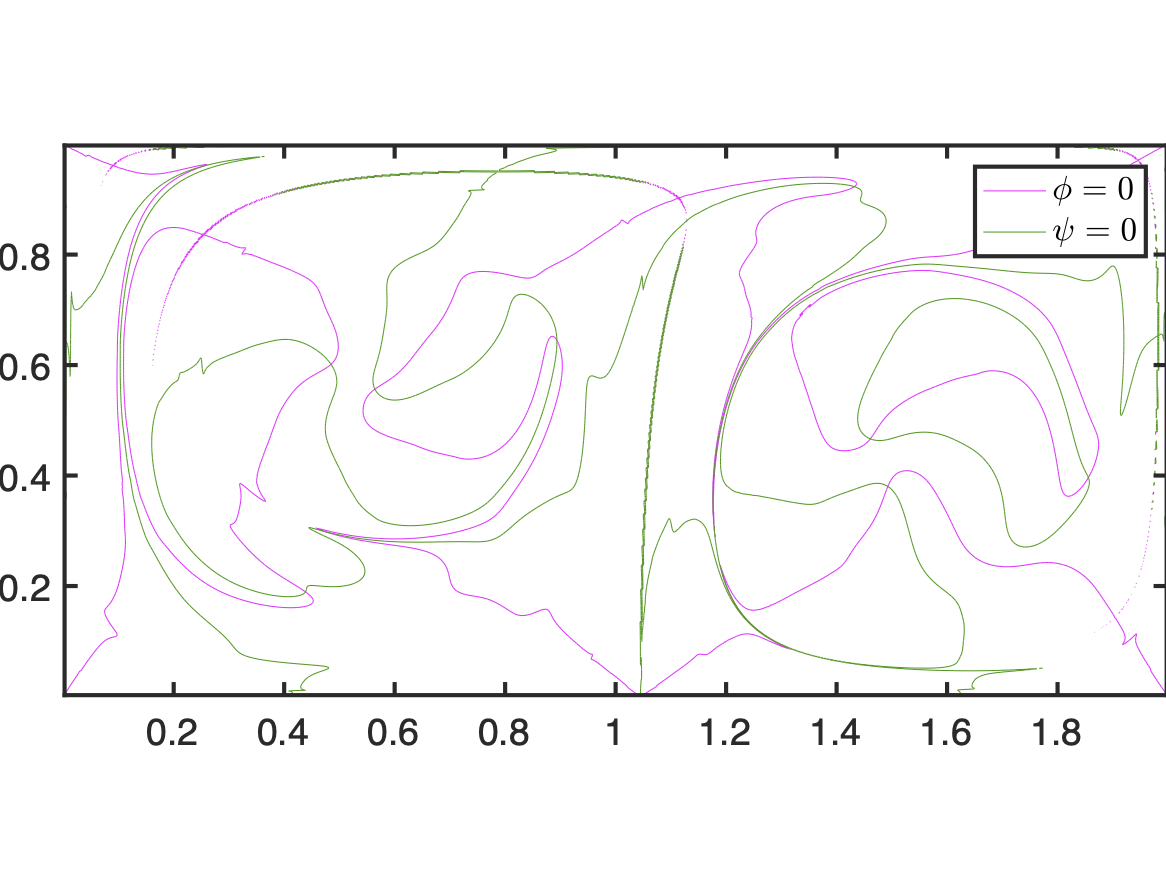}}

 {\includegraphics[scale=0.3]{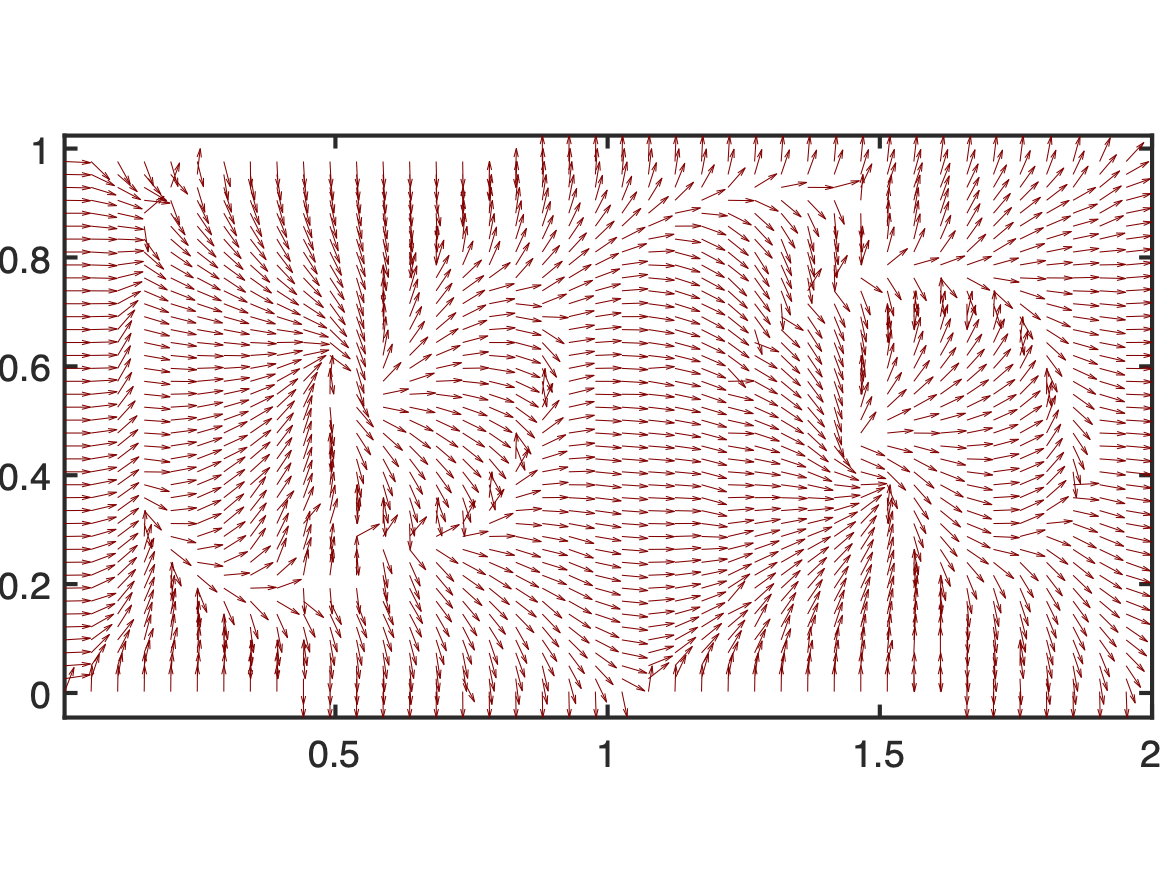}} 
 {\includegraphics[scale=0.3]{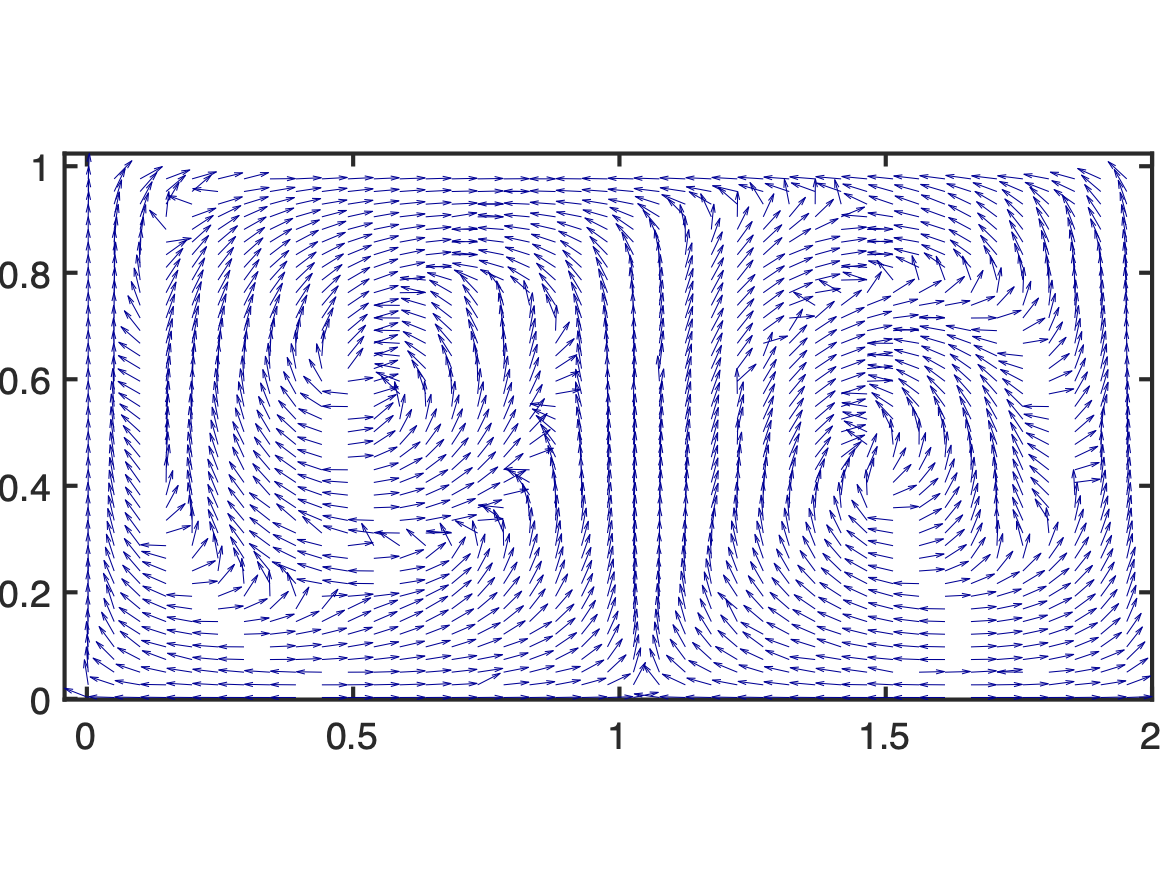}} 

 {\includegraphics[scale=0.3]{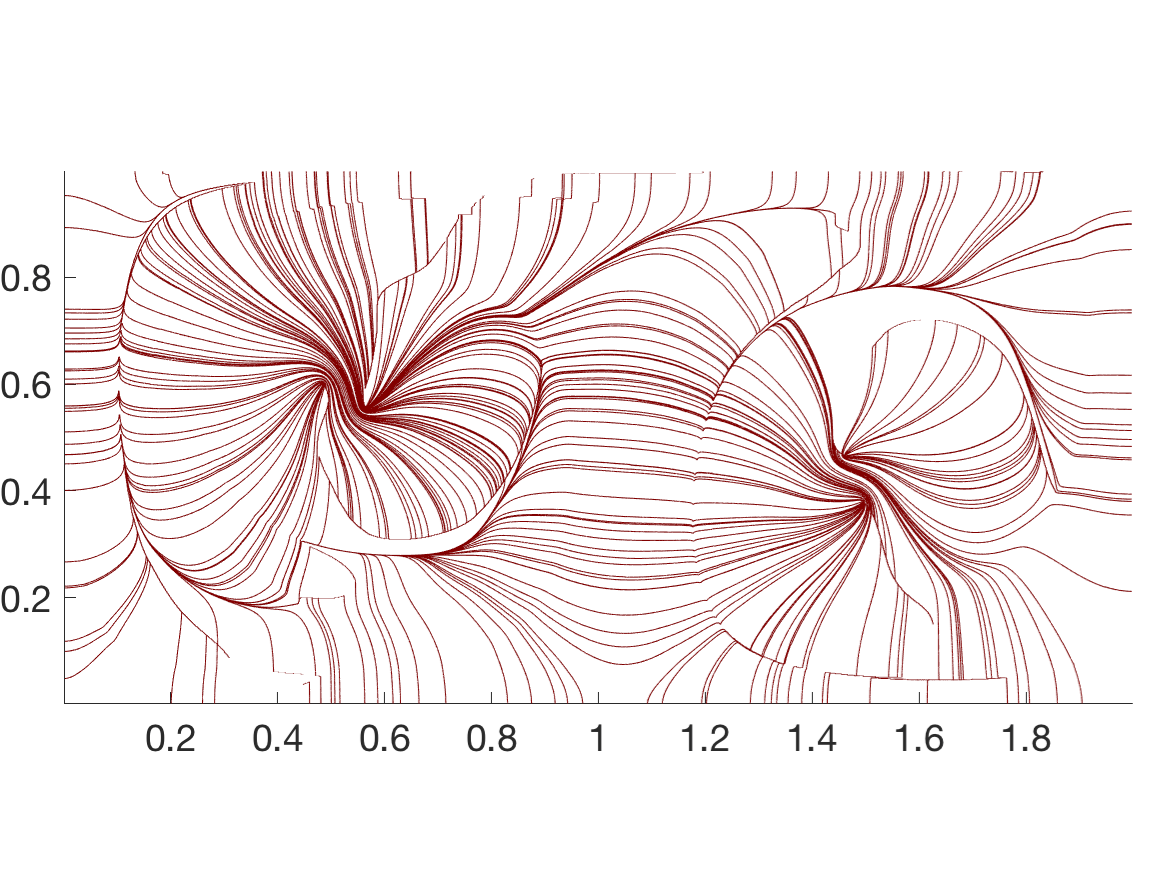}} 
 {\includegraphics[scale=0.3]{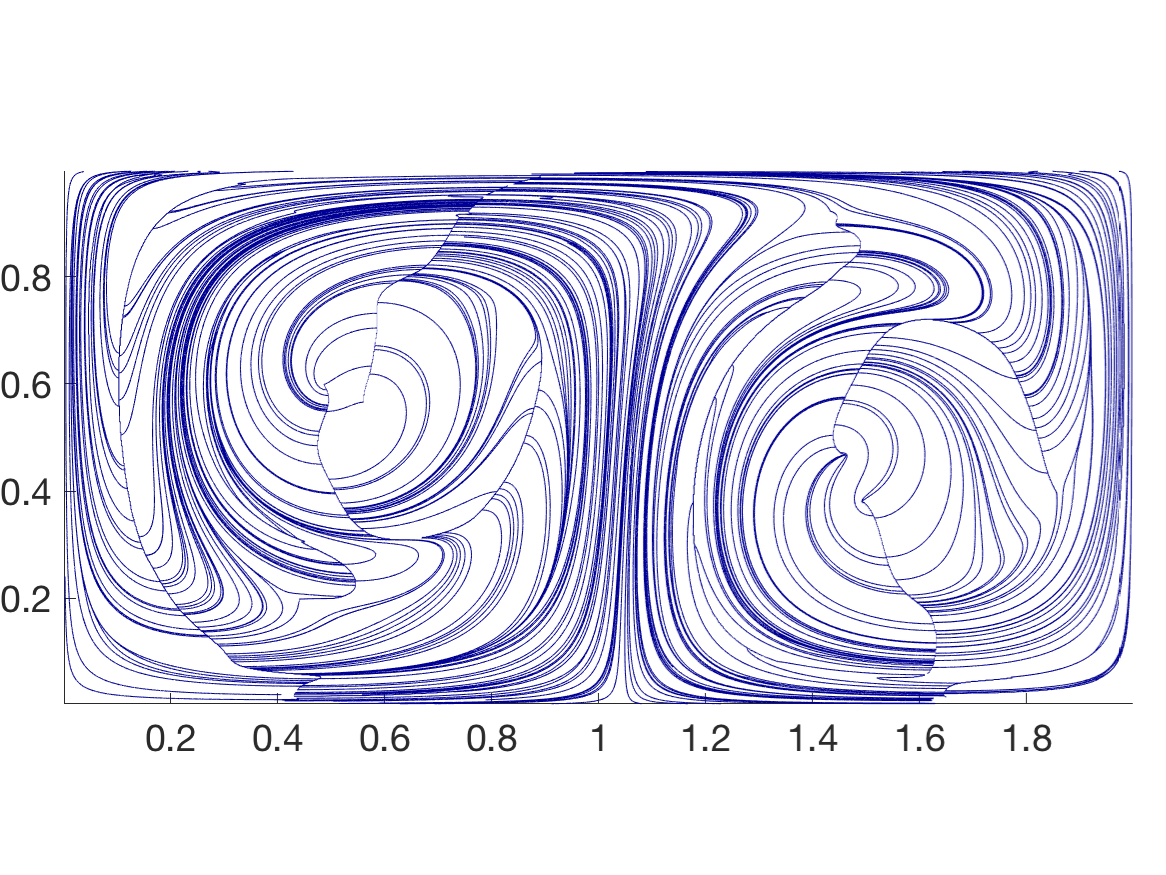}} 

 {\includegraphics[scale=0.3]{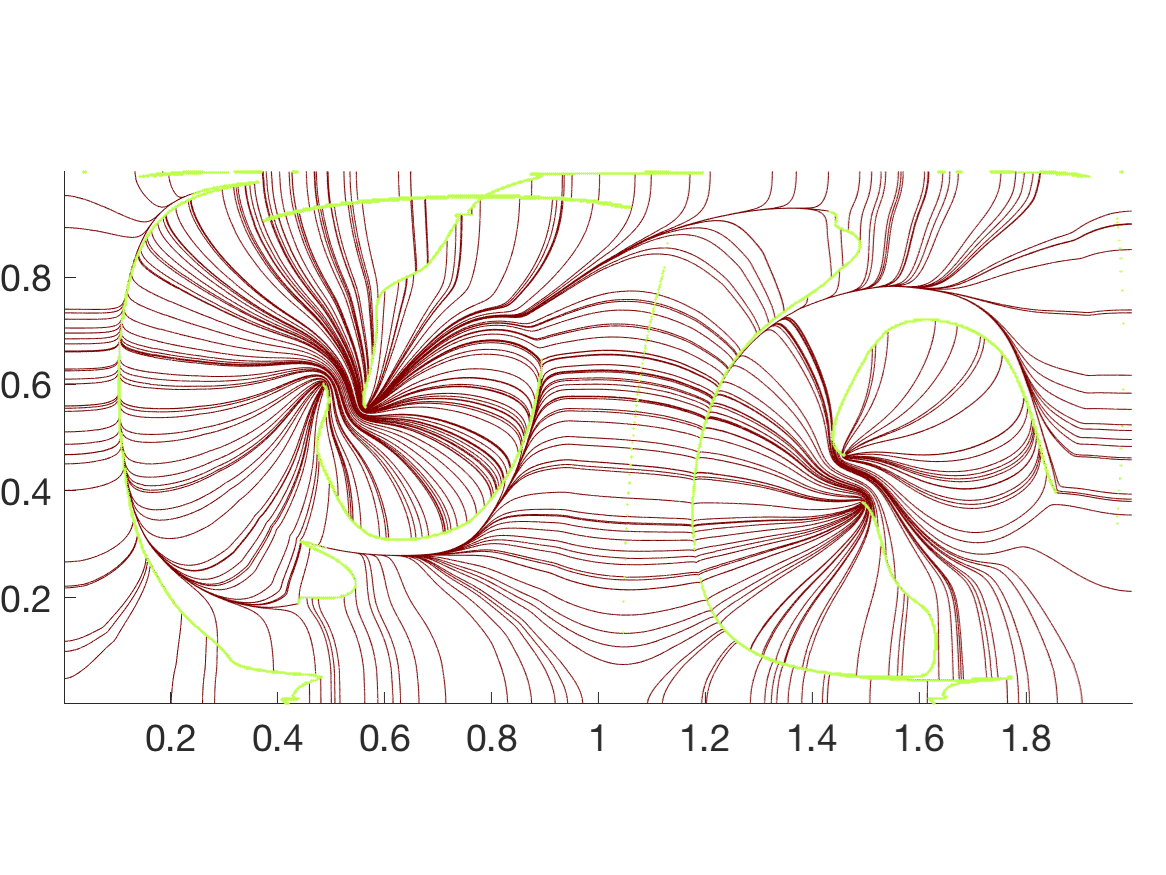}} 
 {\includegraphics[scale=0.3]{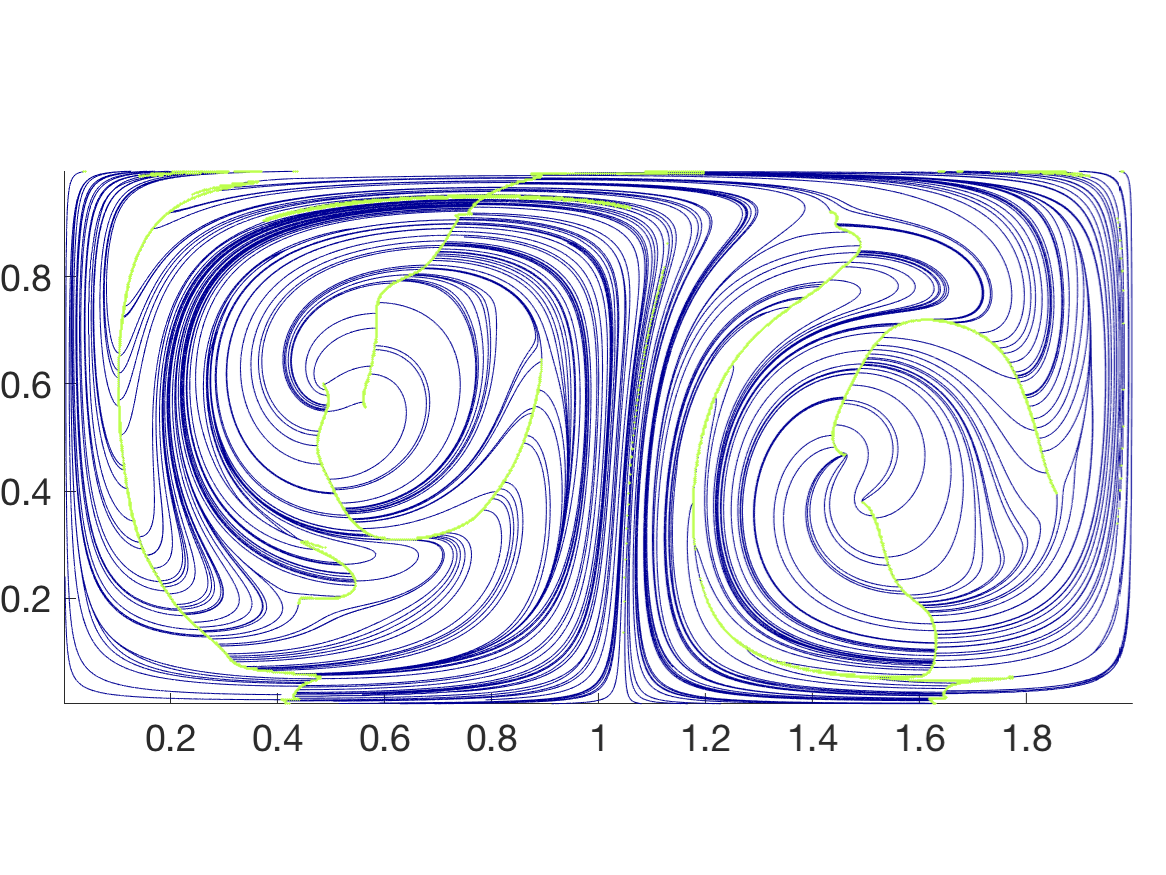}} 

\caption{Optimal foliation computations for the double-gyre flow: (a) The logarithm of the field $ \Lambda^+ $, (b) zero contours of $ \phi $ and $ \psi $, (c)
vector field $ \vec{w}^+  $generated from (\ref{eq:vector_plus}), (d) vector field $ \vec{w}^- $ generated from (\ref{eq:vector_minus}), 
(e) $ \mbox{SORF}_{max}$ by implementing vector field in (c), (f) $ \mbox{SORF}_{min}$ by implementing vector field in (d), (g) $ \mbox{SORF}_{max}$ with
branch cut (green), (g) $ \mbox{SORF}_{min}$ with branch cut.}
\label{fig:doublegyre}
\end{figure}

Fig.~\ref{fig:doublegyre}(a) is a classical figure in this context: the logarithm of the field $ \Lambda^+ $; if divided
by the time-of-flow $ 2 $, this is the finite-time Lyapunov exponent field.  Fig.~\ref{fig:doublegyre}(b)
indicates the $ \phi = 0 $ and $\psi = 0 $ contours, with their intersections defining $ I $.  We use the `standard'
$ \vec{w}^\pm $ unit versions, Eq.~(\ref{eq:vector_plus}), to generate the vector fields in (c) and (d), and the corresponding SORFs
are determined in (e) and (f).  Figs.~\ref{fig:doublegyre}(g) and (h) overlay the branch cuts (green),
which are parts of the green curves in Fig.~\ref{fig:doublegyre}(b) at which $ \phi < 0 $.
As expected, the $ \mbox{SORF}_{max}$ curves fail to cross the branch cut vertically, as do the $ \mbox{SORF}_{min}$ curves
horizontally.  Moreover, foliation curves which do get pushed in towards the branch cuts tend to meander along 
them, giving an impact of spurious accumulations.  We zoom in towards one of these regions in
Fig.~\ref{fig:dg_accumulate}; the $ \mbox{SORF}_{max}$ curves requirements of having slopes $ - \pi/4  $ (resp.\ $ + \pi/2 $)
on $ \Phi_-$ (resp.\ $ \Phi_+ $) result in abrupt curving.  The accumulation is not exactly to $ \Psi_- $, but rather to
a curve which is very close, as seen in Fig.~\ref{fig:dg_accumulate}(b).  Thus, it is {\em not} true that there is a
one-dimensional part of the isotropic set $ I $ along here. The geometric insights of the previous sections allows
us to understand and interpret these issues, while appreciating how resolution may give misleading visual cues.

\begin{figure}
\centering
 {\includegraphics[scale=0.3]{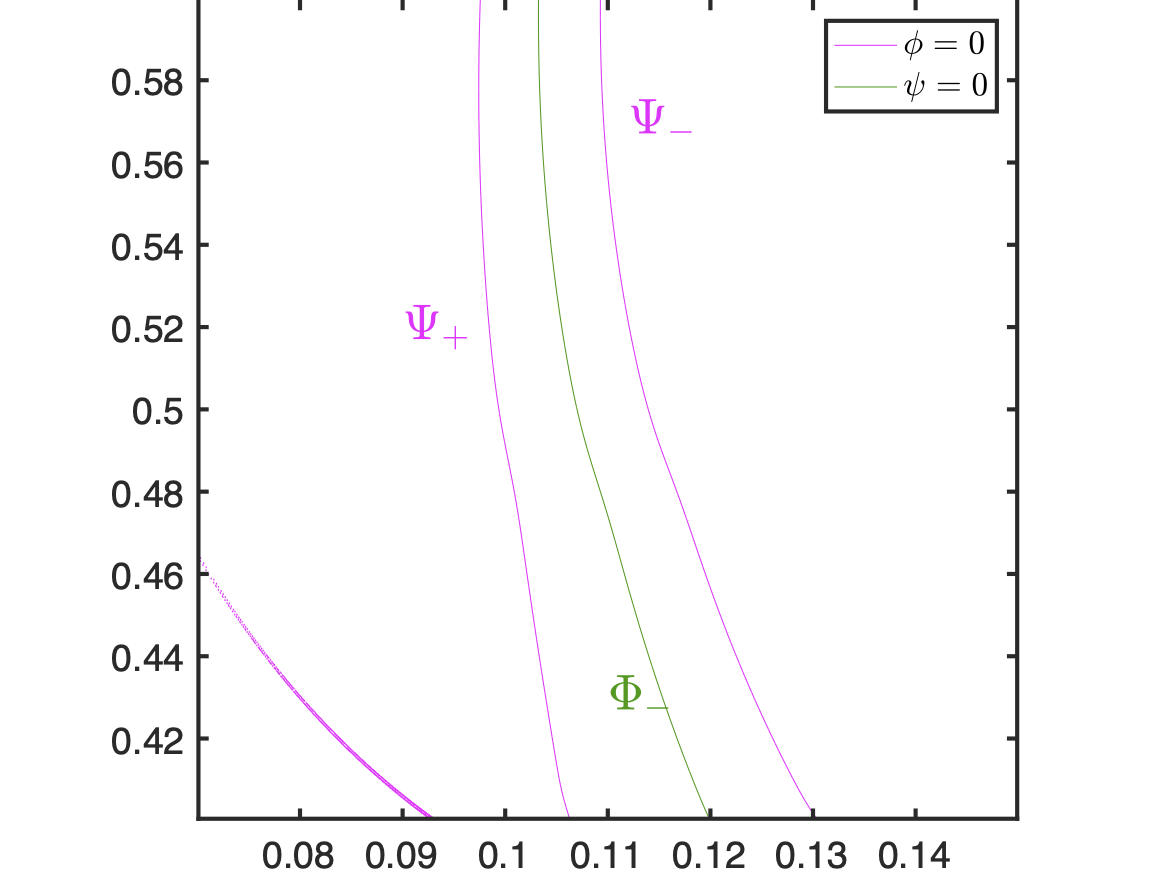}} 
 {\includegraphics[scale=0.3]{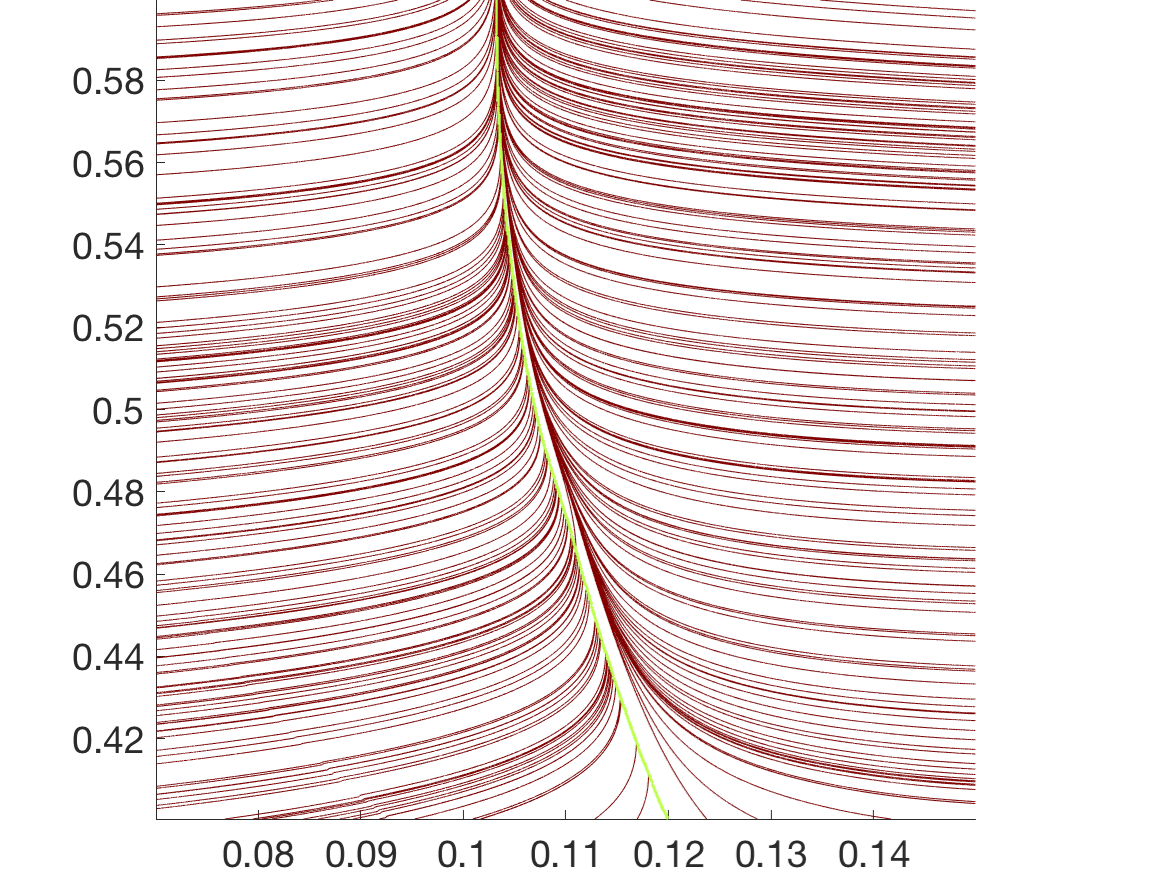}} 
\caption{Zooming in to near an `accumulating' $ \mbox{SORF}_{max}$ from Fig.~\ref{fig:doublegyre}: (a) the relevant zero contours
of $ \phi $ and $ \psi $, and (b) the $ \mbox{SORF}_{max}$.}
\label{fig:dg_accumulate}
\end{figure}

\begin{figure}
\centering
 {\includegraphics[scale=0.3]{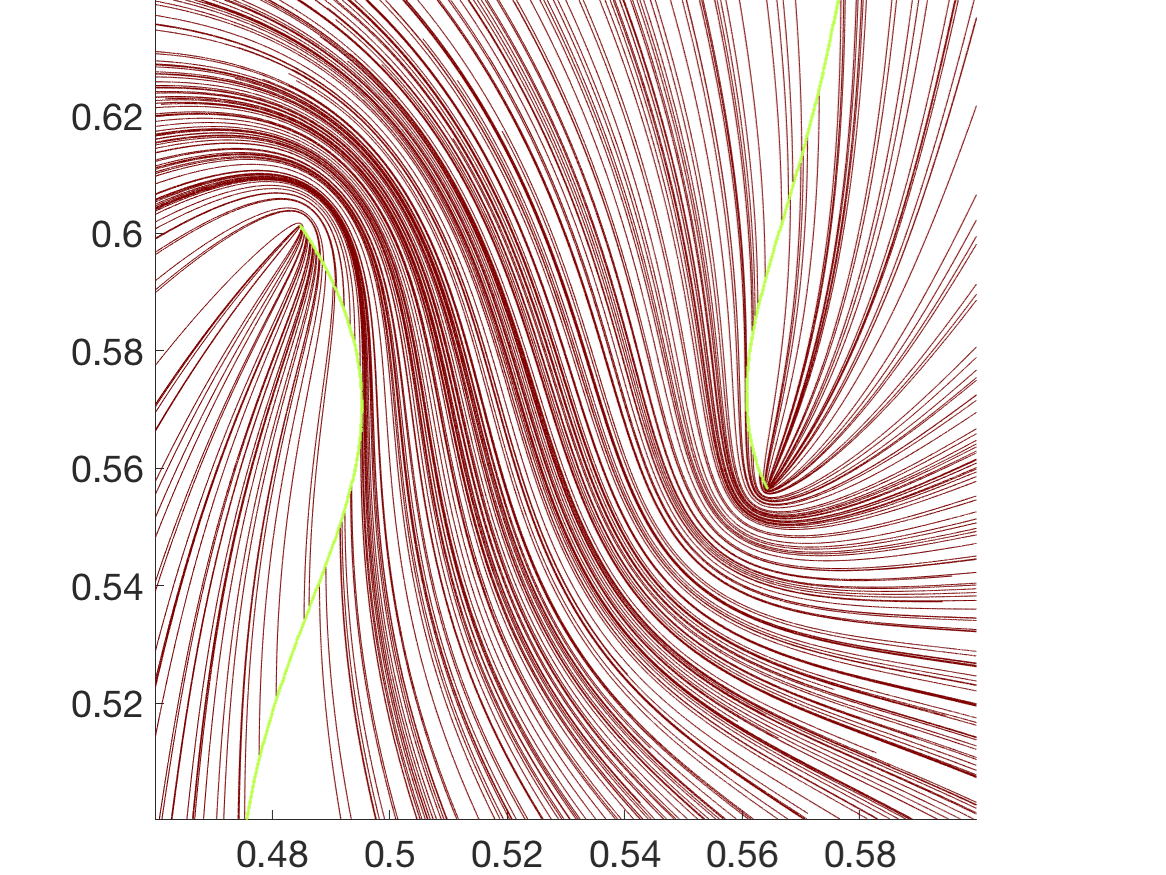}} 
 {\includegraphics[scale=0.3]{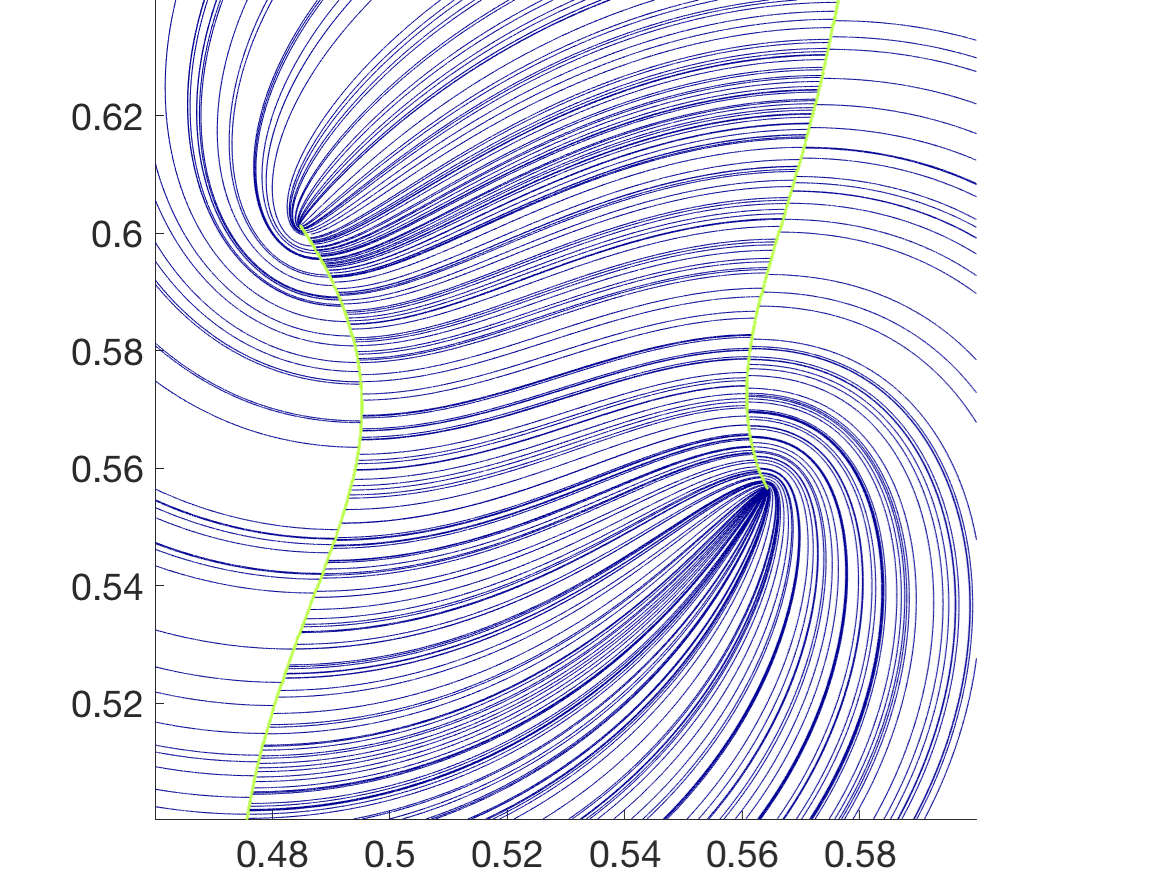}} 

\vspace*{0.2cm}
 {\includegraphics[scale=0.3]{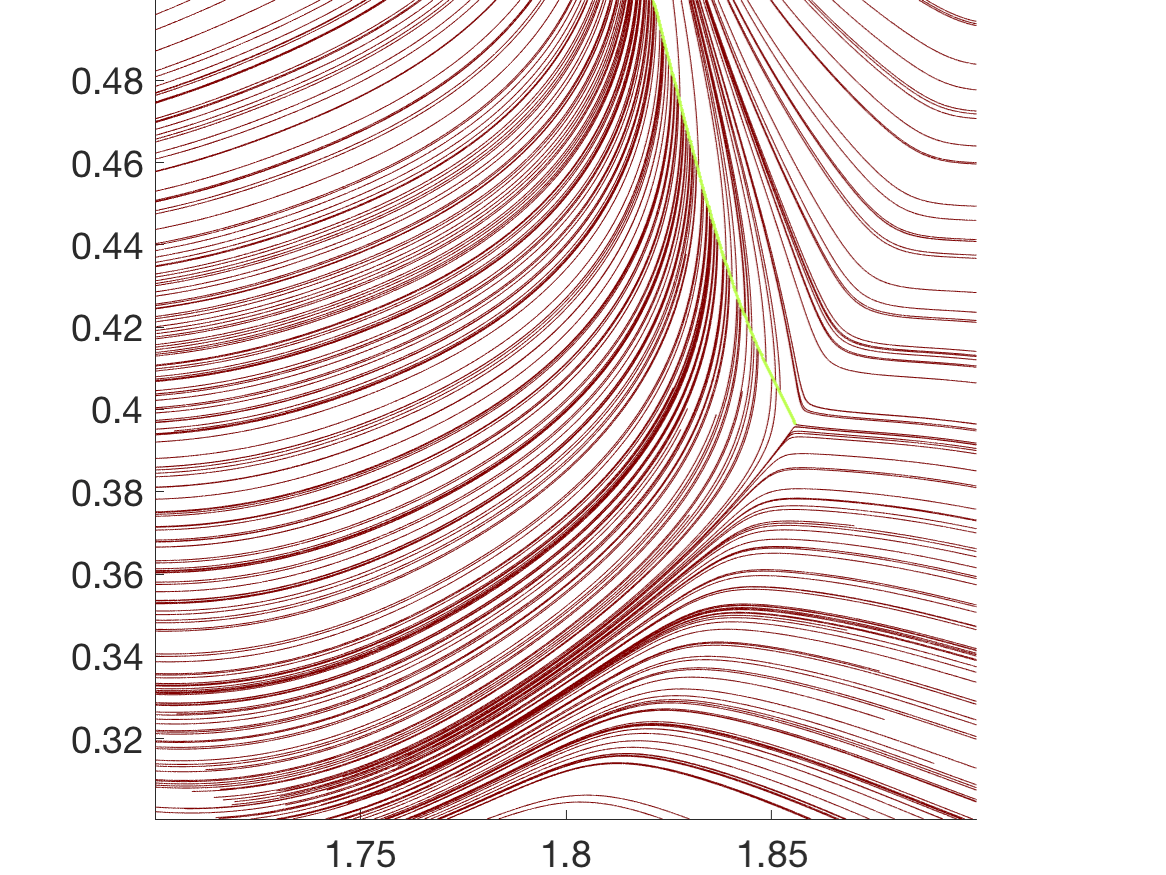}} 
 {\includegraphics[scale=0.3]{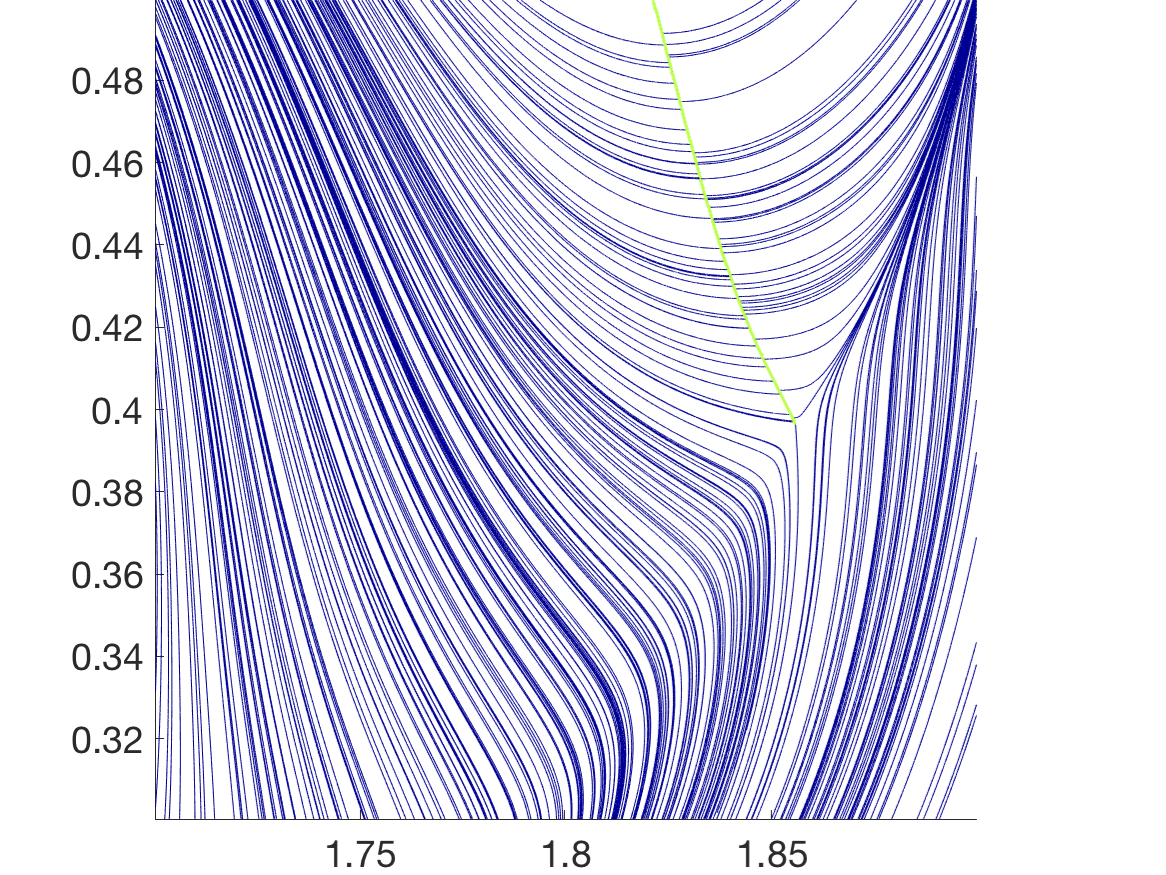}} 

\caption{Zooming in to the $ \mbox{SORF}_{max}$ (left) and $ \mbox{SORF}_{min}$ (right) in the double-gyre.  The top and bottom
panels correspond to different locations, respectively near two adjacent intruding ($ 1 $-pronged) points, and a separating  ($ 3 $-pronged) point.  The branch cut is shown in green.
 Compare to Fig.~\ref{fig:classify} and  Property \ref{property:singularities}.}
\label{fig:doublegyre_zoom}
\end{figure}

In Fig.~\ref{fig:doublegyre_zoom}, we zoom in to two difference locations, chosen by zeroeing in to two
different intersection points of the zero $ \phi $ and $ \psi $-contours.   The top panels illustrate the
$ \mbox{SORF}_{max}$ (left) and the $ \mbox{SORF}_{min}$ (right) curves at the same location.  The theory related to $ 1 $-pronged intruding points is well-demonstrated,
with there being two such points adjacent to each other.  The two orthogonal families
`reverse' the locations of the singularities for the maximizing and minimizing foliations, and the 
branch cut (green) forms vertical/horizontal barriers as appropriate.
In contrast, the bottom figures are of a $ 3 $-pronged separating point; again, the numerics validate the theory.  

%%%%%%%%%%%%
\subsection{Chirikov map}
\label{sec:chirikov}

\begin{figure}
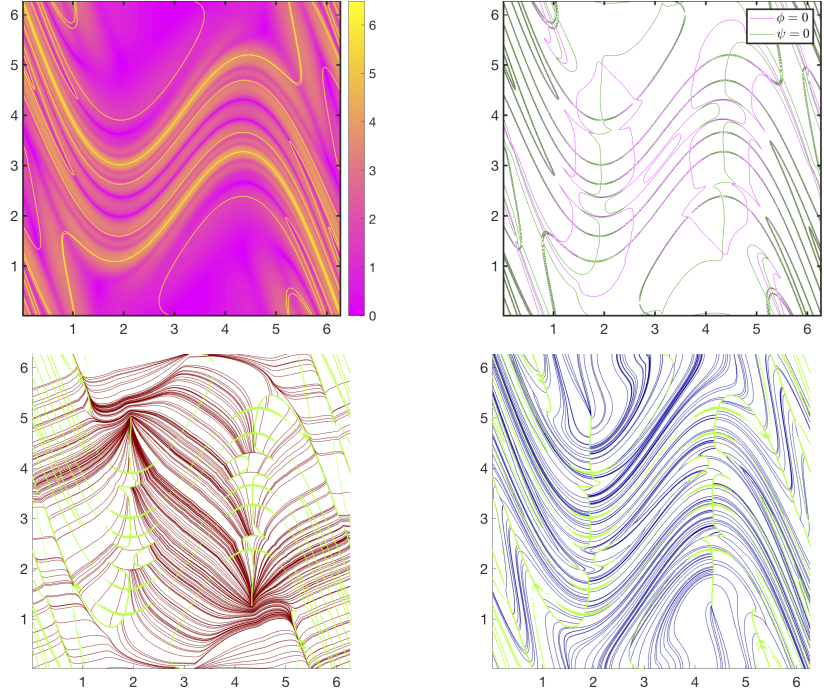

\centering
 {\includegraphics[scale=0.3]{chirikov_k2_T4_stretch.png}}
 {\includegraphics[scale=0.3]{chirikov_k2_T4_phipsi.png}} 

\vspace*{0.2cm}
 {\includegraphics[scale=0.3]{chirikov_k2_T4_fogs.png}} 
 {\includegraphics[scale=0.3]{chirikov_k2_T4_fols.png}} 

\caption{Optimal foliation computations for the Chirikov map $ \vec{F} = \mathfrak{C}_2^4 $: (a) the
logarithm of the field $ \Lambda^+ $, (b) zero contours of $ \phi $ and $ \psi $, 
(c) $ mbox{SORF}_{max}$ with branch cut (green), (d) $ \mbox{SORF}_{min}$ with branch cut (green).}
\label{fig:chirikov}
\end{figure}

The Chirikov (also called `standard') map is defined on the doubly-periodic domain $ \Omega = [0,2\pi) \times [0,2\pi) $ by
\cite{chirikov}
\[
\mathfrak{C}_k(x,y) = \left( \begin{array}{c} x + y + k \sin x ~~~({\mathrm{mod}} \, \, 2 \pi)  \\
y  + k \sin x ~~~({\mathrm{mod}} \, \,  2 \pi) 
\end{array} \right) \, .
\]
We choose $ \vec{F} = \mathfrak{C}_k^n $, that is, $ n $ iterations of the Chirikov map for a given value of the
parameter $ k $.  Increasing $ k $ increases the disorder of the map, as does having $ n $ large. (The map is a
classical example of chaos, with $ \Omega $ consisting of quasiperiodic islands in a chaotic sea, where `chaos/chaotic' must be understood
in the limit $ n \rightarrow \infty $.)  In more disorderly situations,
increasingly fine resolution is required to reveal the structures that we have defined.

Relevant computations for $ k = 2 $ and $ n = 4 $ are shown in Fig.~\ref{fig:chirikov}.  There are significant regions
where the behavior is quite orderly.  There is `greater disorder' in
the region foliated with large values of $ \Lambda^+ $ in (a)---indeed, this region is associated with the `chaotic
sea' when the map is iterated many more times---with the outer parts of low $ \Lambda^+ $ being associated with
quasiperiodic islands and hence order.  All features mentioned in previous examples are reiterated in the pictures.
Moreover, the $ \mbox{SORF}_{min}$ foliation somewhat mirrors the structure expected from classical
Poincar\'e section numerics.  

If we instead consider $ k = 1 $ and $ n = 2 $, an interesting degenerate singularity (corresponding to the $ \psi = 0 $
contour crossing exactly a saddle point
of $ \phi $) is displayed in Fig.~\ref{fig:chirikov_triple}.  The singularity in the $S \mbox{SORF}_{max}$ foliation (b) appears like a degenerate form of a separating point, if thinking
in terms of curves coming from above.  However, if viewed in terms of curves coming in from below, it appears as an intruding point with a sharp (triangular) end.  The $ \mbox{SORF}_{min}$ conforms to this, having elements
of a separating point, and an intruding point, as well. (The numerical issue of $ \mbox{SORF}_{min}$ not crossing 
$ B $ horizontally is displayed in Fig.~\ref{fig:chirikov_triple}(c); in reality, the $ \mbox{SORF}_{min}$ curves should connect smoothly
across.)

\begin{figure}
\centering

 {\includegraphics[scale=0.2]{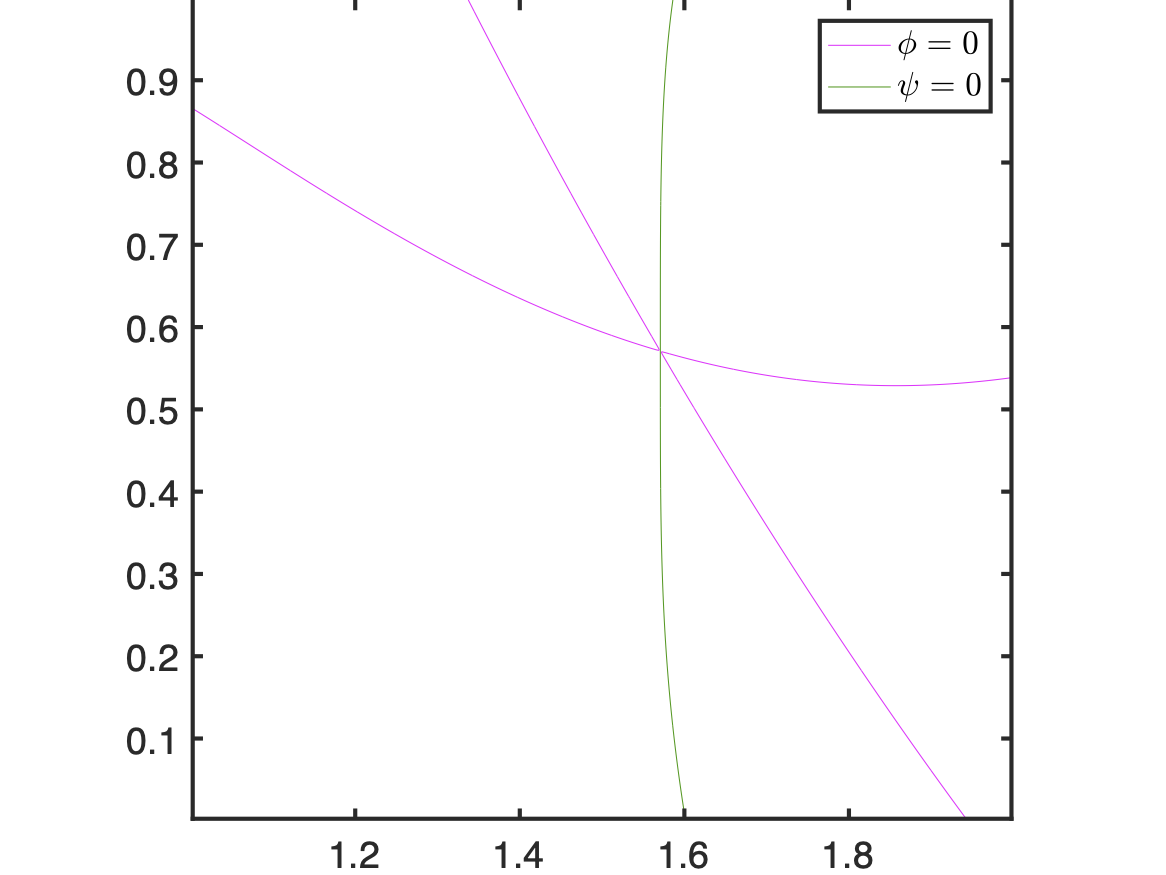}}
 {\includegraphics[scale=0.2]{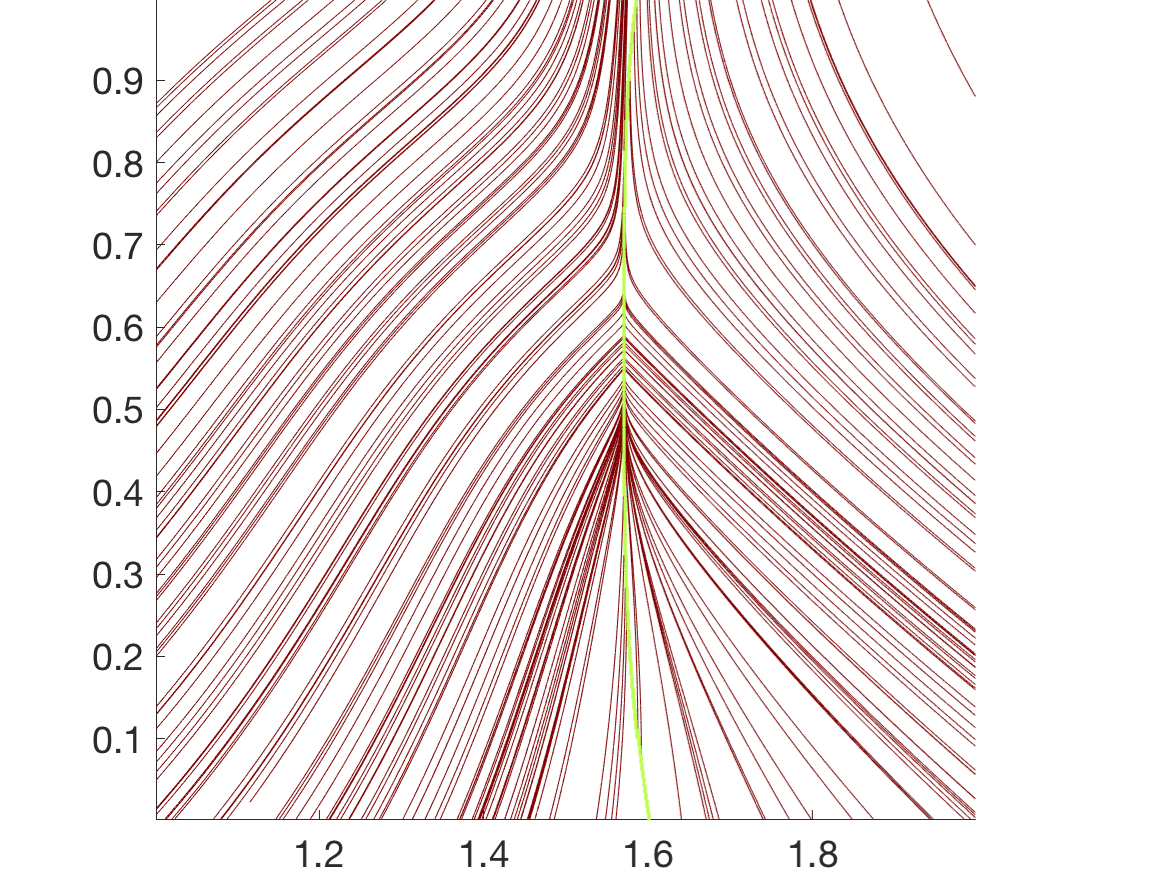}}
 {\includegraphics[scale=0.2]{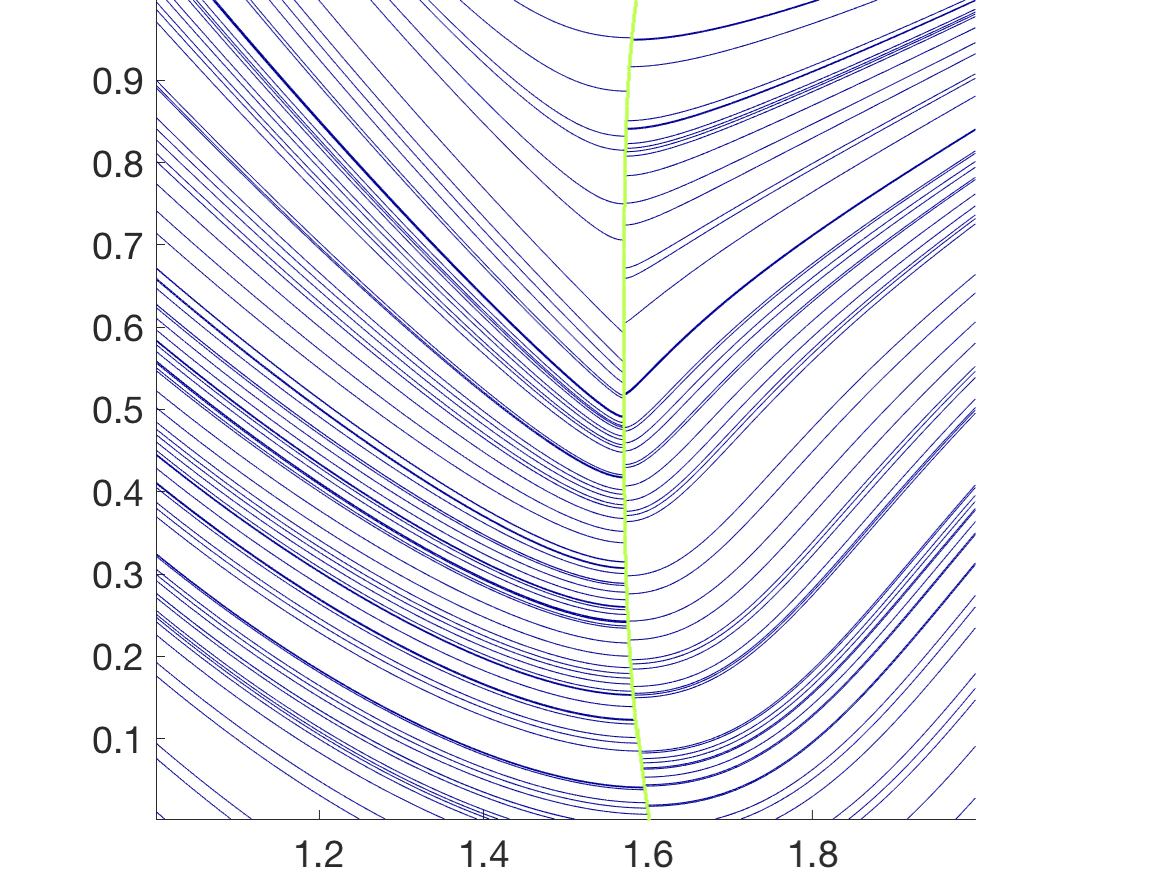}}

\caption{A degenerate singularity of the map $ \vec{F} = \mathfrak{C}_1^2 $, shown zoomed-in: (a) 
the zero contours of $ \phi $ and $ \psi $, (b) $ \mbox{SORF}_{max}$, and (c) $ \mbox{SORF}_{min}$.}
\label{fig:chirikov_triple}
\end{figure}

Next, we demonstrate in Fig.~\ref{fig:chirikov2_k2}, using $ \vec{F} = \mathfrak{C}_2^2 $, the efficacy of using the integral-curve forms (\ref{eq:cogs}) and (\ref{eq:cols}) of the foliations,
rather than using a vector field.  The $ \ln \Lambda^+ $ field in Fig.~\ref{fig:chirikov2_k2}(a) has
several sharp ridges; these are well captured by locations where the
$ \phi $ and $ \psi $ zero-contours in Fig.~\ref{fig:chirikov2_k2}(b) coincide.  The $ \mbox{SORF}_{max/min} $ foliations
in (b) and (c) are computed respectively using the vector fields $ \vec{w}^\pm $ as in previous
situations, and exhibit the usual issues when crossing $ B $.    In contrast, the lower row is generated by using the 
integral-curve forms (\ref{eq:cogs}) and (\ref{eq:cols}), where we have once again started from
$ 300 $ random initial conditions.  For each initial condition $ (x_1,y_1) $, we define the next point $ (x_2,y_2) $ on
a $ \mbox{SORF}_{max}$ curve by $ x_2 = x_1 + \hdown^+(x_1,y_1)  \delta y $ where $ \delta y > 0 $ is the spatial resolution in the $ y $-direction,
and  $dx/dy $ is based on (\ref{eq:cogs}).  Similarly, $ y_2 = y_1 + h^+(x_1,y_1) \delta x $ using (\ref{eq:cogs}), and
where $ \delta x > 0 $ is the resolution chosen in $ x $-direction.  This
initializes the process.  Next, we check the value of $ h_+(x_2,y_2) $, thereby deciding which of the equations in 
(\ref{eq:cogs}) to implement.  If the $ dy/dx $ equation, we take $ x_3 = x_2 + {\mathrm{sign}} \left( x_2 - x_1 \right) \delta x $,
and thus find $ y_3 $ using the ODE solver.  Having now obtained $ (x_3,y_3) $, we again use the last two points to
make decisions on which of the two equations to use, and continue in this fashion for a predetermined number of steps.
Next, we go back to $ (x_1,y_1) $ and now set $ x_2 = x_1 -  \hdown^+(x_1,y_1) \delta y $ and $ y_2 = y_1 -  h^+(x_1,y_1) \delta y $, thereby going in the opposite direction. Having initiated this process, we can then continue this curve using
the same continuation scheme.  The $ \mbox{SORF}_{min}$ are obtained similarly, using the two equations in (\ref{eq:cols}).
There
is sensitivity in the process to locations where $ \phi $ and $ \psi $ change rapidly (they are each of the order 
$ 10^5 $ in this situation), and in particular where zeros are near.  The resolution scales $ \delta x $ and $
\delta y $ need to be reduced sufficiently to not capture spurious effects.  Notice that there are no branch-cut
problems in the resulting foliations obtained using the integral-curve approach, since we do not have to worry about a discontinuity in a vector field.  Neither are there any abrupt stopping of curves.

%\note{**beyond the scope here - but between us - what happens in the infinite time limit?  The stable and unstable manifolds of henon for example are part of the foliations?  How about the SORFs - do the SORF-3 type become very narrow forks with increasing n?  How about when there are tangencies of stable and unstable manifolds - what do those look like in the foliations?** SB: I have no idea, Erik!  Would be fascinating if globally optimizing foliations had relationshipsto this.  The $ \mbox{SORF}_{min} $ in general seems to have some indications of partitioning into coherent regions such as a chaotic sea, etc.  I'm not sure in what way this connects to invariant manifolds, which are specialized curves whereasthe foliation gives a full set of curves across the domain...}
%\note{EB2: I love this part of the discussion - but we need to remove it from this paper - so between us - ...  1)  what happens as n goes to infinity to these structures?2) does SORFmin partition into coherent regions and chaotic sea regions?  So cut this all out or make a throw away sentence in the conclusion?**}

\begin{figure}
\centering
 {\includegraphics[scale=0.3]{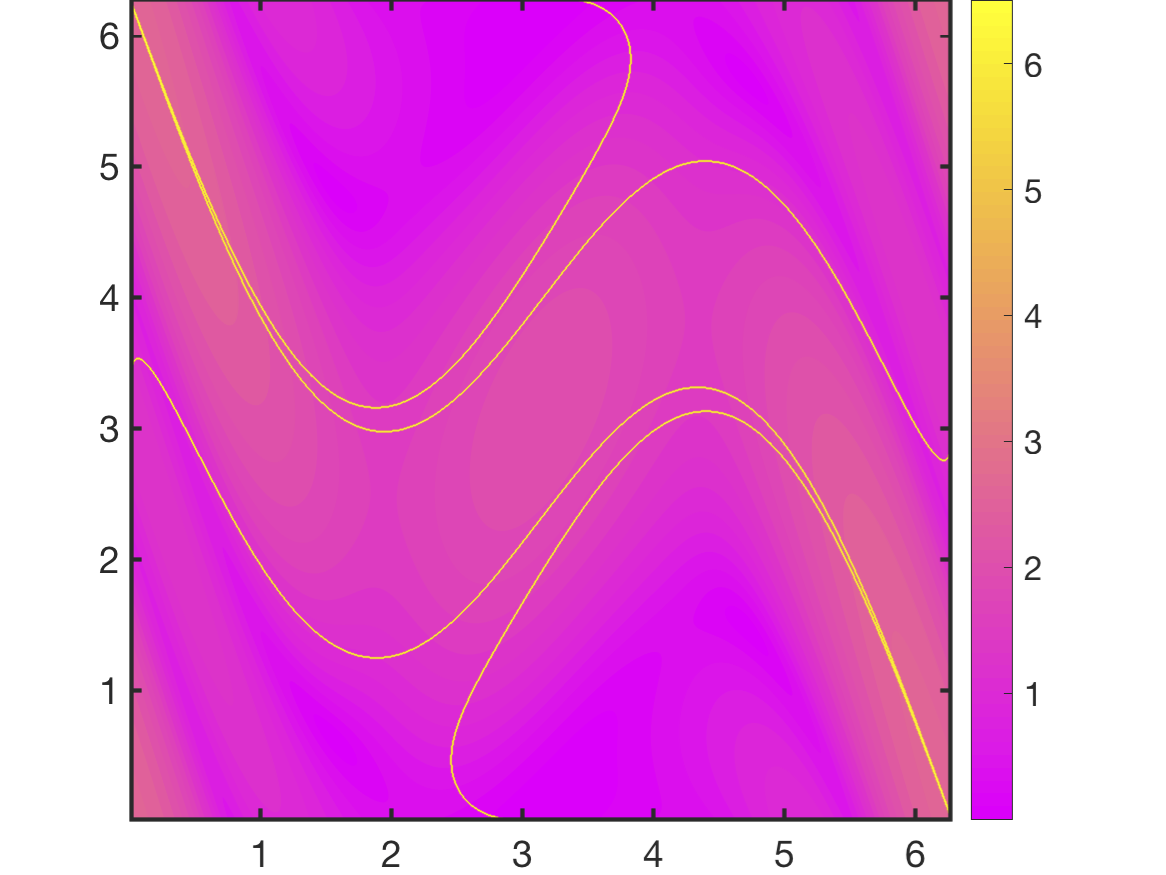}} 
 {\includegraphics[scale=0.3]{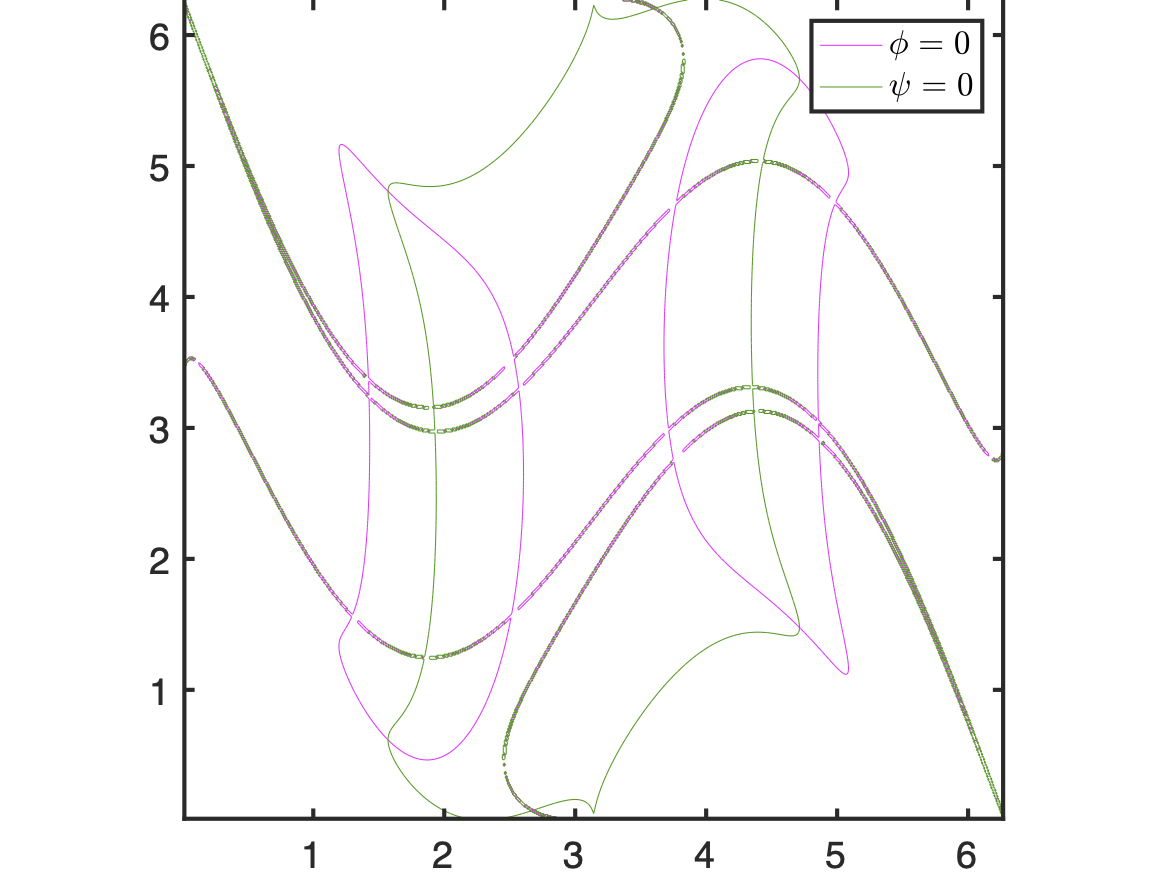}} 

\vspace*{0.2cm}
 {\includegraphics[scale=0.3]{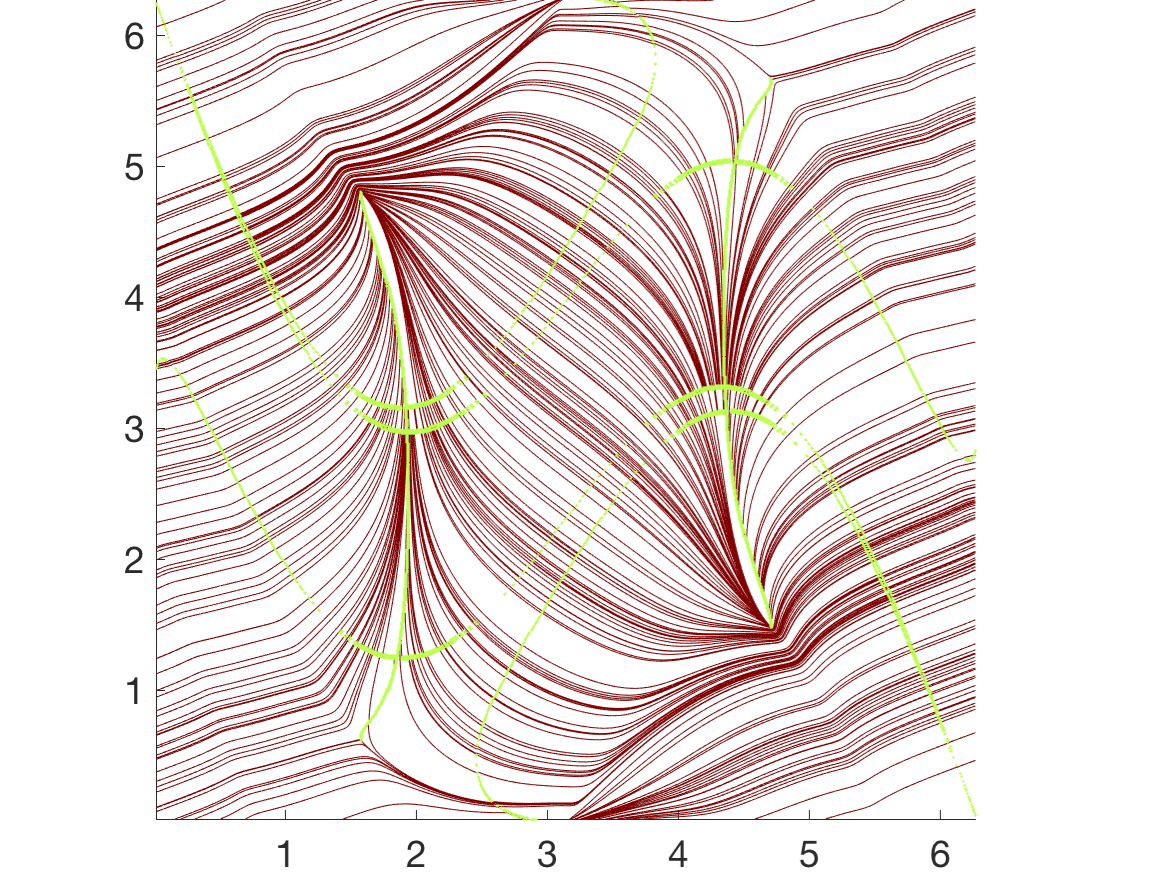}} 
 {\includegraphics[scale=0.3]{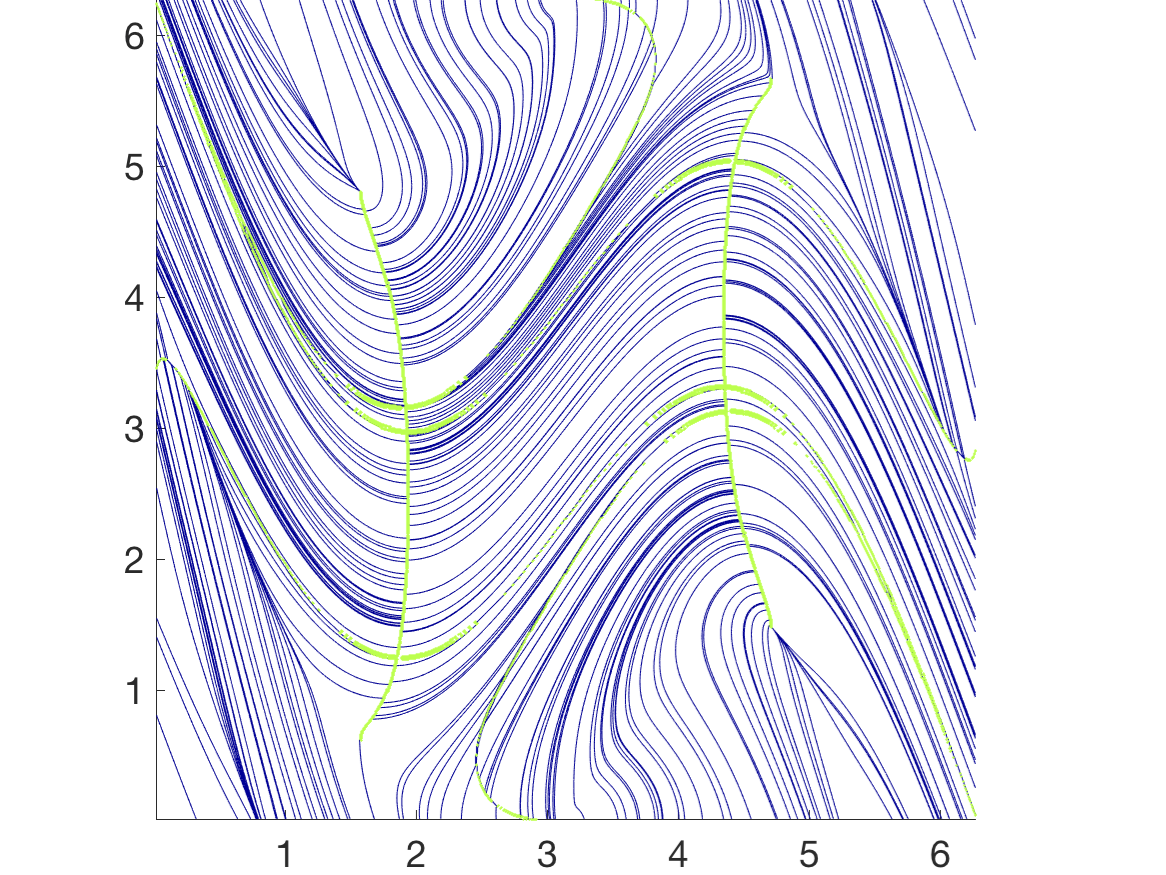}} 

\vspace*{0.2cm}
 {\includegraphics[scale=0.3]{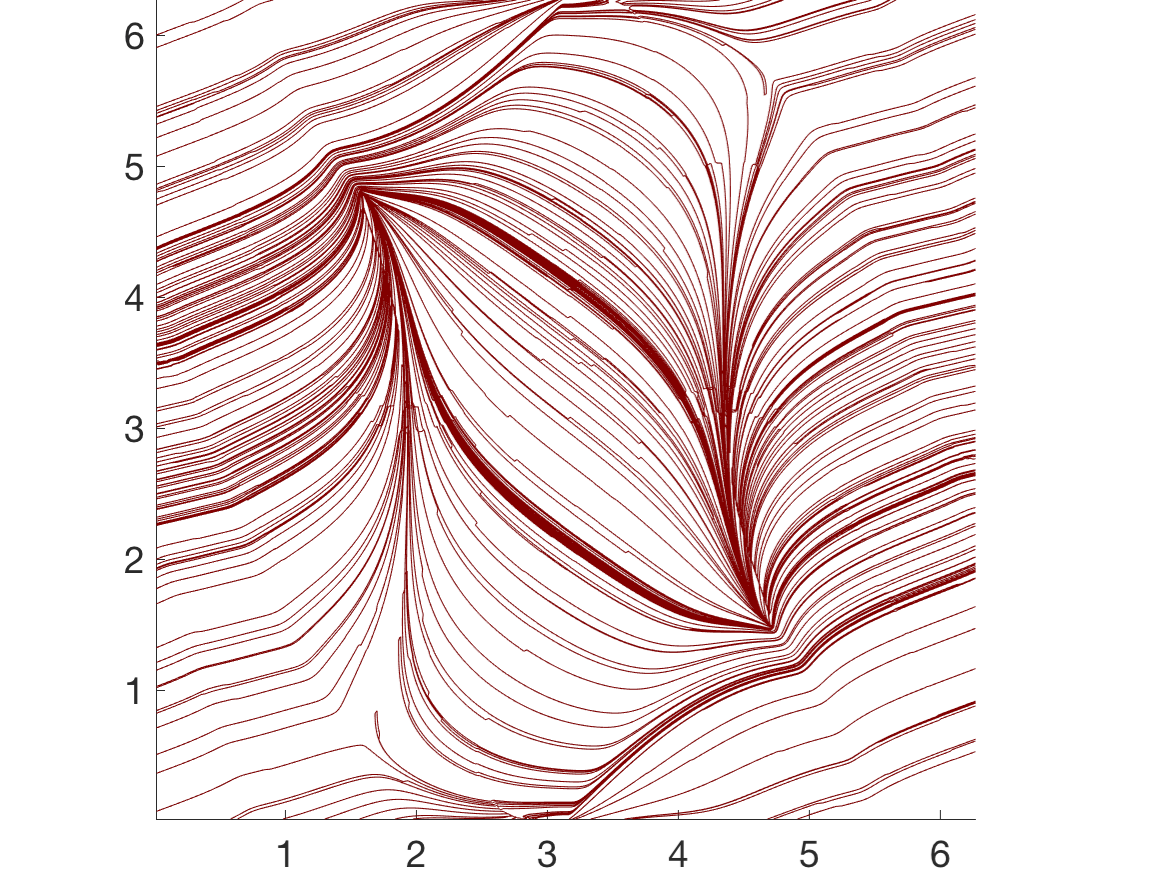}} 
 {\includegraphics[scale=0.3]{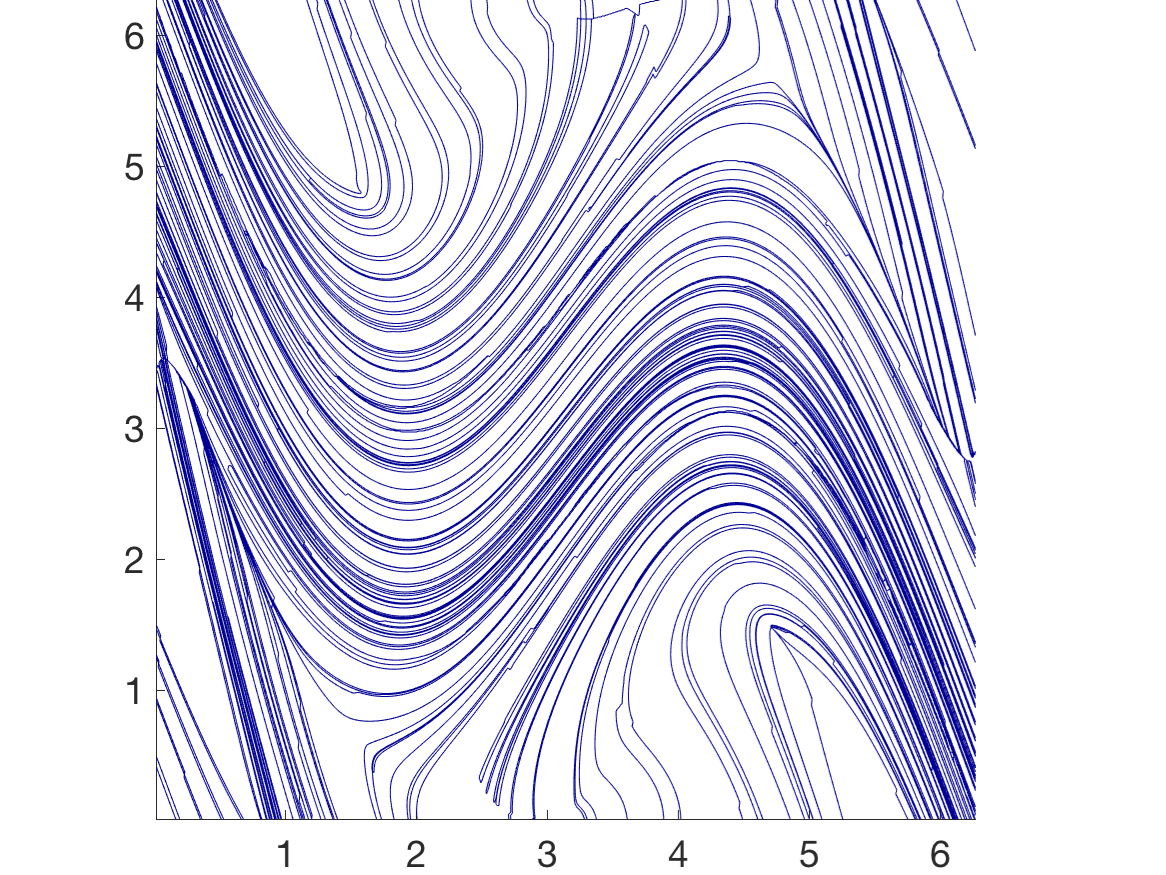}} 

\caption{Comparison between using the integral-curve forms (\ref{eq:cogs}) and (\ref{eq:cols}) and the
vector field forms for $ \vec{F} = \mathfrak{C}_2^2 $: (a) $ \ln \Lambda^+ $ field, (b) zero contours of
$ \phi $ and $ \psi $, (c) $ \mbox{SORF}_{max}$ using the vector field (\ref{eq:vector_plus}), (d) $
\mbox{SORF}_{min}$ using the vector field (\ref{eq:vector_minus}), 
(e) $ \mbox{SORF}_{max}$ using the integral curve form (\ref{eq:cogs}), and (f) $ \mbox{SORF}_{min}$ using the form (\ref{eq:cols}).}
\label{fig:chirikov2_k2}
\end{figure}

%%%%%%%%%%%%%%%%%%
\section{Concluding remarks}

In this paper, we have examined the issue of determining foliations which {\em globally} maximize and minimize
stretching associated with a two-dimensional map, where the map can be defined in terms of a finite sequence
of discrete maps, or a finite-time flow of a differential equation.  Our formulation establishes a connection to the
well-known {\em local} optimizing issue, and provides new insights into the resulting foliations and their singularities.
In particular, an easy criterion for classifying the nature of generic singularities is expressed.  Some numerical 
artefacts arising when computing these foliations in standard ways are characterized in terms of a `branch cut' phenomenon, and 
a methodology of avoiding these is developed.  We have expressed connections with a range of related and highly
studied concepts (Cauchy--Green tensor, Lyapunov vectors, singularities of vector fields), and demonstrated
computations in both discretely- and continuously-derived maps.  

We expect these results to help researchers interpret, and improve, numerical calculations in related situations.
In particular, misinterpretations of numerics can be mitigated via the understandings presented here.  Regions
of high sensitivity towards spatial resolutions are also identifiable in terms of the near-zero sets of the $ \phi $ and $ \psi $
functions.  

We wish to highlight from our numerical results the role of $ \mbox{SORF}_{min}$ restricted foliations as being effective demarcators of complication flow regimes.  These curves---observable for example in blue in Figs.~\ref{fig:henon}, 
\ref{fig:doublegyre}, \ref{fig:chirikov} and \ref{fig:chirikov2_k2}---indicate curves along which there is minimal stretching.
Consequently, there is maximal stretching in the orthogonal direction to these curves.  This indicates that the $ \mbox{SORF}_{min}$
curves are barriers in some senses: disks of initial conditions positioned on such a curve experience sharp stretching
orthogonal to them.  That is, initial conditions on one side of such a curve get separated quickly from those on the other
side, with the curve positioned optimally to maximize the separation.  Our methodology enables
this intuitive idea to be put into a global optimizing foliation framework.   Looking at this another way, the dense regions of the $ \mbox{SORF}_{min}$ (blue) foliations in Figs.~\ref{fig:henon}, \ref{fig:doublegyre}, \ref{fig:chirikov} and  \ref{fig:chirikov2_k2} are reminiscent of separation curves which attempt to demarcate chaotic from regular regions.
We emphasize, though, that `chaotic' has no proper meaning in the finite-time context since it must be
understood in terms of infinite-time limits; in this case, the separation one may try to obtain is between more `disorderly'
and `orderly' regions.  The ambiguity of defining these is reflected in the Figures, in which the $ \mbox{SORF}_{min}$ foliation 
nonetheless identifies coherence-related topological structures in $ \Omega $ which are strongly influenced by the
nature of the singularities in the foliation.

%\note{**EB: Here's the riff text I suggest:
Note that the interaction of $\phi=0$ and $\psi=0$ level sets as seen in Fig.~\ref{fig:henon}(b) bear a striking resemblance to Figures regarding zero angle between stable and unstable foliations of Lyapunov vectors such as in Fig.~1 for the H\'{e}non map from \cite{jaeger1997structure} that was part of a search for primary heteroclinic tangencies when developing symbolic dynamic generating partitions of the Henon map, \cite{grassberger1989symbolic, bollt2001symbolic,bollt2000validity,christiansen1997guidelines}.  Indeed this analysis likely bears a relationship, in that in a infinite time limit, the Lyapunov vectors suggested come to the same point as those much earlier stories underlying the topological dynamics of smooth dynamical systems.  What is clear in the finite time discussion here is that when we see a coincidence between the stretching and folding, that in successively longer time windows, these properties repeat in progressively smaller regions.  As suggested by Fig.~\ref{fig:henon}, e.g.\ (h), any point of tangency would in turn be infinitely repeated in the long time limit.  The perspective of this current work may further understanding of what has always been the intricate topic of why and how hyperbolicity is lost in nonuniformly hyperbolic systems wherein seemingly paradoxically, errors can grow along the directions related to stable manifolds, such as highlighted by Fig.~5 in
\cite{jaeger1997homoclinic}.
%}

{\bf Acknowledgements:} SB acknowledges with thanks partial support from the Australian Research Council via
grant DP200101764.  EB acknowledges with thanks the Army Research Office (N68164-EG) and also DARPA.

%%%%%%%%%%%%%%%%%
%%%%%%%%%%%%%%%
\appendix

%%%%%%%%
\section{Proof of Lemma~\ref{lemma:puncture}}
\label{sec:puncture}

Given a general point $ (x,y) \in \Omega_0 $, let $ \theta
 \in [-\pi/2,\pi/2) $.
The local stretching (\ref{eq:stretching}) associated with this point and direction is 
\[
\Lambda \left( x,y,\theta \right) = \sqrt{ \left( u_x \cos \theta + u_y \sin \theta \right)^2 + \left( v_x \cos \theta +
v_y \sin \theta \right)^2} \, .
\]
where the $ (x,y) $-dependence on $ u_x $, $ u_y $, $ v_x $ and $ v_y $ has been omitted 
from the right-hand side for brevity. Hence, 
\[
\Lambda^2 = \frac{u_x^2 + v_x^2 - u_y^2 - v_y^2}{2} \cos 2 \theta + \left( u_x u_y + v_x v_y \right) \sin 2 \theta
+ \frac{u_y^2 + v_y^2 + u_x^2 + v_x^2}{2} \, .
\]
Using the definitions for the functions $ \phi $ and $ \psi $ from (\ref{eq:phi}) and (\ref{eq:psi}),
\begin{equation}
\Lambda^2 = \phi \cos 2 \theta + \psi \sin 2 \theta + \frac{\left| \vecnabla u \right|^2 + \left| \vecnabla v \right|^2}{2} \, .
\label{eq:Lambdasquared}
\end{equation}
Given the linear independence of the sine and cosine functions, the value of $ \Lambda^2 $ at $ (x,y) $ is
 independent of $ \theta $ if and only if $ \phi $ and $ \psi $ are both zero.  Thus, the isotropic set is
characterized as the intersection of the zero sets of the functions $ \phi $ and $ \psi $.

%%%%%%%%%
\section{Proof of Lemma~\ref{lemma:theta}}
\label{sec:theta}

\begin{figure}
 {\includegraphics[width=0.47\textwidth,height=0.21\textheight]{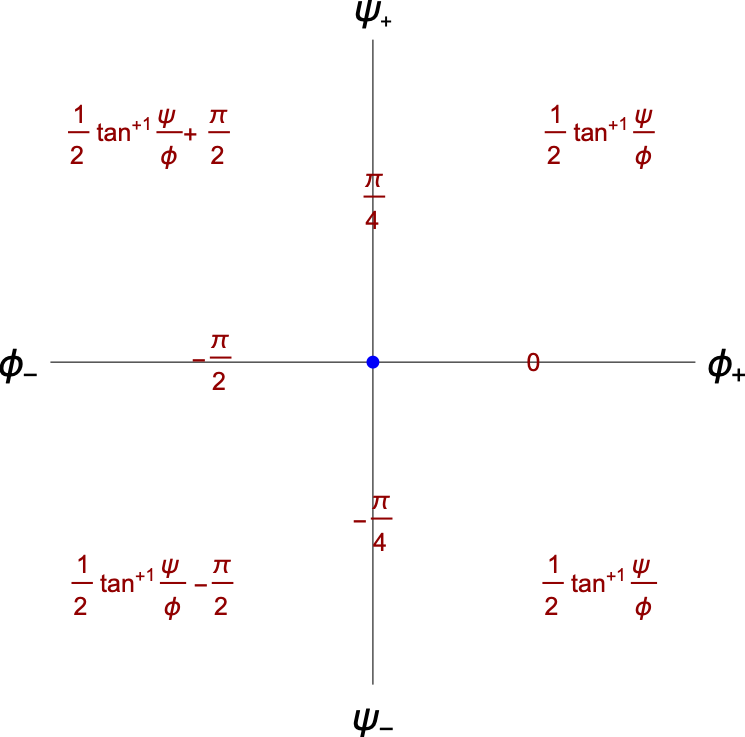}} 
 {\includegraphics[width=0.47\textwidth,height=0.21\textheight]{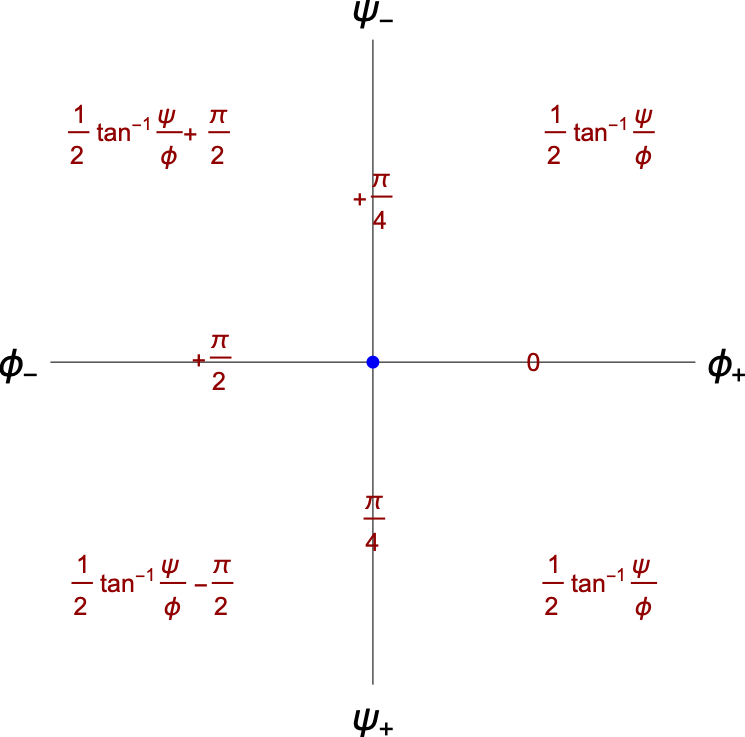}} 

\vspace*{0.2cm}
 {\includegraphics[width=0.47\textwidth,height=0.17\textheight]{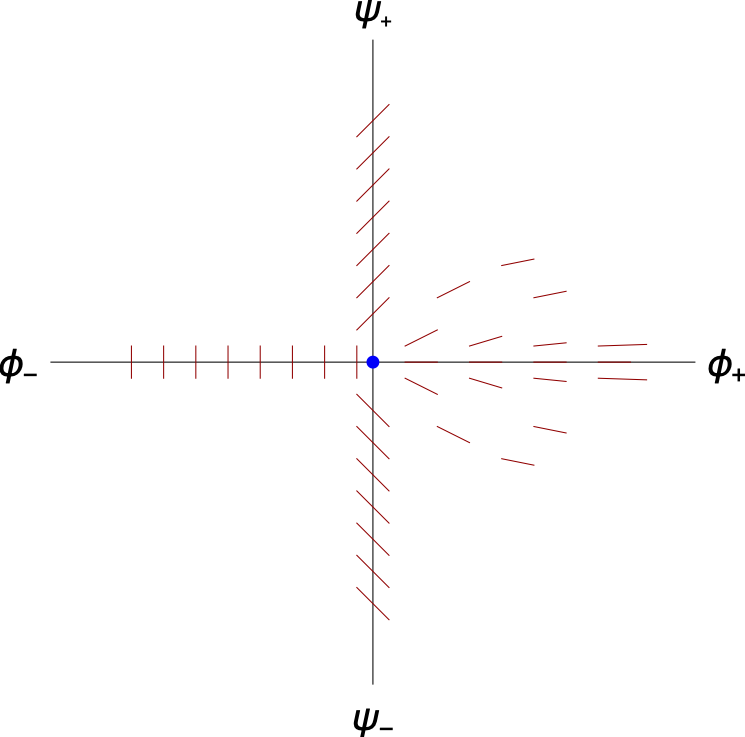}} 
 {\includegraphics[width=0.47\textwidth,height=0.17\textheight]{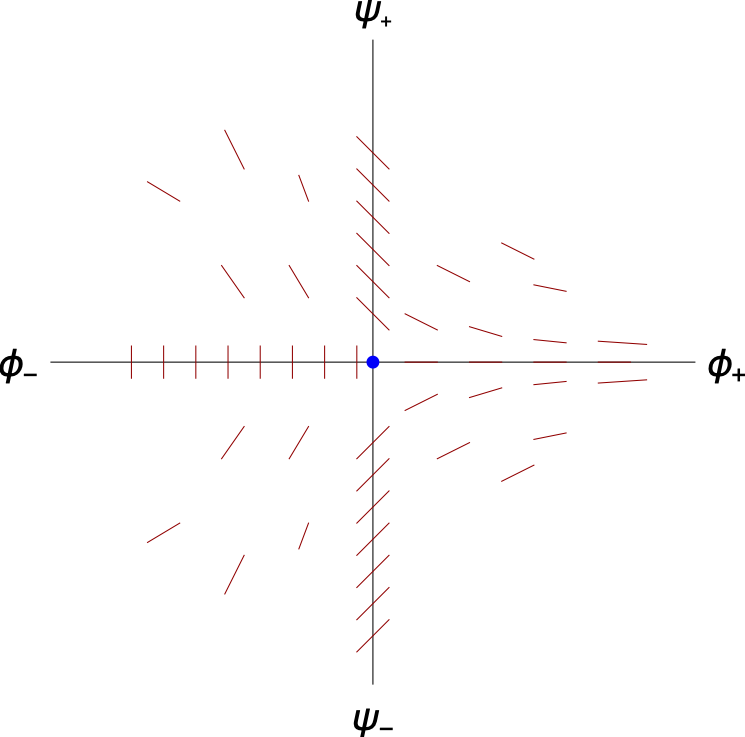}} 

\vspace*{0.2cm}
 {\includegraphics[width=0.47\textwidth,height=0.17\textheight]{intruding.png}} 
 {\includegraphics[width=0.47\textwidth,height=0.17\textheight]{separating.png}} 

\vspace*{0.2cm}
 {\includegraphics[width=0.47\textwidth,height=0.17\textheight]{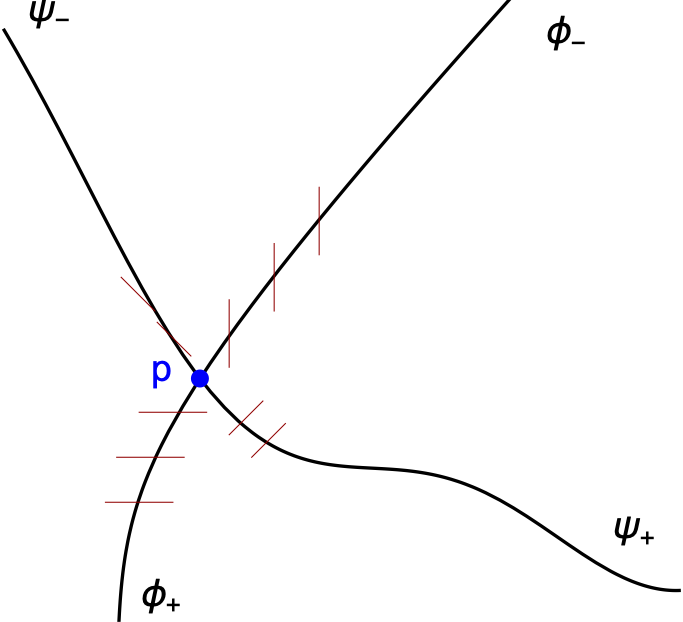}} 
 {\includegraphics[width=0.47\textwidth,height=0.17\textheight]{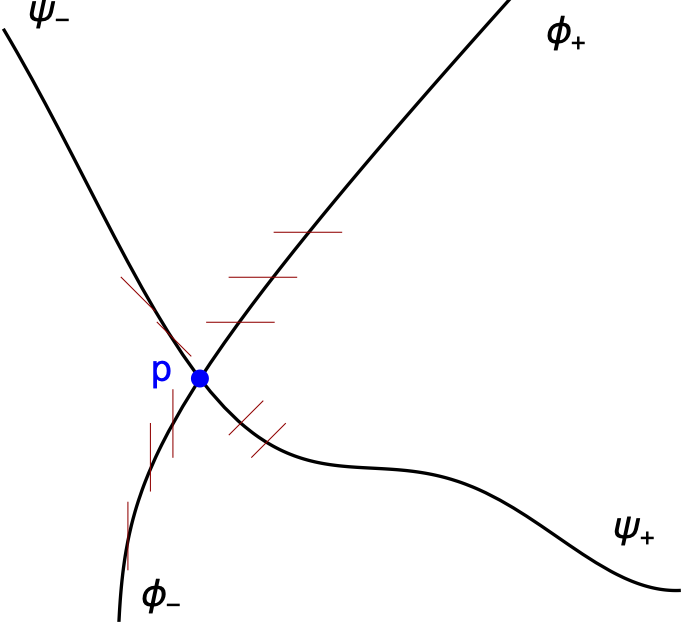}}
\caption{$ \mbox{SORF}_{max}$ near a nondegenerate singularity: (a) Value of $ \theta^+ \in [-\pi/2, \pi/2) $ in $ (\phi,\psi) $-space
using (\ref{eq:thetaplus}),  (b) as in (a), but shown in a left-hand system, (c) and (d) qualitative slope fields for (a) and (b); (e) $ 1 $-pronged `intruding
point' associated with the structure (c); (f) $ 3 $-pronged `separating' point associated with the structure (d); (g) 
intruding point
when axes are tilted; (h) separating point when axes are tilted. Compare to Fig.~\ref{fig:classify} and  Property \ref{property:singularities}.}
\label{fig:fourquadrant}
\end{figure}

We begin with (\ref{eq:thetaplus}), and obtain (\ref{eq:thetaplus2}).   Assuming for now that both $ \phi $ and $ \psi $ are not zero, we use  the double-angle formula to obtain 
\[
\frac{2 \tan \theta^+ }{ 1 -
\tan^2 \theta^+ }  = \tan 2 \theta^+ = \frac{\psi}{\phi}  \, . 
\]
Solving the quadratic for $ \tan \theta^+  $, we see that
\begin{equation}
\tan \theta^+ = \frac{-1 \pm \sqrt{ (\psi/\phi)^2 + 1}}{\psi/\phi} = 
\frac{ - \phi \pm \sqrt{\phi^2 + \psi^2}}{\psi}
\label{eq:tanthetaLambda}
\end{equation}
We now need to choose the sign in this expression, bearing in mind the usage of the four-quadrant inverse 
tangent as used in (\ref{eq:thetaplus}).  The four quadrants here are in the $ (\phi,\psi) $-space, 
which is indicated in Fig.~\ref{fig:fourquadrant}(a).  If $ \phi > 0 $ and $ \psi > 0 $, this implies that $ 2 \theta^+ $
is in the first quadrant, and thus so is $ \theta^+ $.  This means that $ \tan \theta^+> 0 $, and
consequently the positive sign must be chosen.  If $ \phi > 0 $ and $ \psi < 0 $, $ 2 \theta^+ $ is in fourth
quadrant, or $ 2 \theta^+ \in (-\pi/2,0) $.  Thus, $ \tan \theta^+  < 0 $, and so the positive sign must be chosen
in  (\ref{eq:tanthetaLambda}) to ensure that the division by $ \psi < 0 $ leads to an eventual negative sign.  Next,
if $ \phi < 0 $ and $ \psi > 0 $, $ 2 \theta^+  \in (\pi/2,\pi) $, and $ \theta^+ \in (\pi/4,\pi/2) $, leading
to $ \tan \theta^+ > 0 $ and the necessity of choosing the positive sign in (\ref{eq:tanthetaLambda}).  Finally, if
$ \phi < 0 $ and $ \psi<0 $, $ 2 \theta^+ \in (-\pi, -\pi/2) $ and $ \theta^+ \in (-\pi/2,-\pi/4) $, and thus $ 
\tan \theta^+ < 0 $ and the positive sign in the numerator of (\ref{eq:tanthetaLambda}) must be chosen.  
Thus, all cases lead to a positive sign, and so
\[
\tan \theta^+ = \frac{-\phi+ \sqrt{\phi^2+\psi^2}}{\psi} \, , 
\]
whence (\ref{eq:thetaplus2}) when neither $ \phi $ nor $ \psi $ is zero.  

Next, we rationalize the fact that (\ref{eq:thetaplus2}) arises from (\ref{eq:thetaplus}) even if one or the other of $ \phi $ or $ \psi $ is zero.  The arguments to follow are equivalent to considering the four emanating axes in Fig.~\ref{fig:fourquadrant}(a).  If $ \phi = 0 $ and $ \psi \ne 0 $,
(\ref{eq:thetaplus}) tells us that $ 2 \theta^+ = \, (\pi/2) \,  {\mathrm{sign}} \left( \psi \right) $ and thus $ 
\tan \theta^+ = \tan (\pi/4) \, {\mathrm{sign}} \left( \psi \right) = {\mathrm{sign}} \left( \psi \right) $.  This is consistent 
with what
(\ref{eq:thetaplus2}) gives when $ \phi = 0 $ is inserted. 
If $ \psi = 0 $ and $ \phi \ne 0 $, (\ref{eq:thetaplus}), which tells us that $ 2 \theta^+ = - \pi $ if $ \phi < 0 $,
or $  2 \theta^+ = 0 $ if $ \phi > 0 $.  Thus if $ \psi = 0 $,  $ \theta^+ = - \pi/2 $ if $ \phi < 0 $, and $ \theta^+ = 0 $ 
if $ \phi > 0 $.  This verifies that (\ref{eq:thetaplus2}) is equivalent to (\ref{eq:thetaplus}) in 
$ \Omega_0 $.

Now, $ \theta^- $ in (\ref{eq:thetaminus}) is defined specifically to be orthogonal to $ \theta^+ $.  There
is only one angle in $ [-\pi/2,\pi/2) $ which obeys this condition.  It is straightforward to verify from (\ref{eq:thetaplus2}) and (\ref{eq:thetaminus2}) that 
\[
 \left( \tan \theta^+ 
\right) \left( \tan \theta^- \right) = - 1 
\]
in $ \Omega_0 $.  Thus, $ \theta^- $ as defined in (\ref{eq:thetaminus2}) is at right-angles to $ \theta^+ $
as defined in (\ref{eq:thetaplus2}), which has been established to be equivalent to (\ref{eq:thetaplus}).

%%%%%%%%
\section{Proofs of Theorems~\ref{theorem:max} and \ref{theorem:min}}
\label{sec:proof}

First, we tackle Theorem~\ref{theorem:max}, related to maximizing the global stretching.
Let $ f $ be a restricted foliation on $ \Omega $, and $ \theta_f $ be the unique angle field in $ \Omega_0 $ associated with it.  From
(\ref{eq:Lambdasquared}) from the proof of Lemma~\ref{lemma:puncture}, we have that the local stretching $ \Lambda $ at a point $ (x,y) \in \Omega_0 $ 
related to the angle $ \theta_f $ obeys
\begin{eqnarray}
\Lambda^2
& = & \sqrt{\phi^2 + \psi^2} \left[ \frac{\phi}{\sqrt{\phi^2 + \psi^2}} \cos 2 \theta_f + \frac{\psi}{\sqrt{\phi^2+\psi^2}}
\sin 2 \theta_f \right] +  \frac{\left| \vecnabla u \right|^2 + \left| \vecnabla v \right|^2}{2} \nonumber \\
& = & \sqrt{\phi^2 + \psi^2} \left[ \cos 2 \theta^+ \cos 2 \theta + \sin 2 \theta^+ \sin 2 \theta_f 
\right] +  \frac{\left| \vecnabla u \right|^2 + \left| \vecnabla v \right|^2}{2} \nonumber \\
& = & \sqrt{\phi^2 + \psi^2} \cos \left[ 2 \left( \theta^+ - \theta_f \right) \right] + 
 \frac{\left| \vecnabla u \right|^2 + \left| \vecnabla v \right|^2}{2}
\label{eq:stretching2}
\end{eqnarray}
in which $ \theta^+ = \theta^+(x,y) $ satisfies
\begin{equation}
\cos 2 \theta^+ = \frac{\phi}{\sqrt{\phi^2 + \psi^2}} \quad {\mathrm{and}} \quad
\sin 2 \theta^+ = \frac{\psi}{\sqrt{\phi^2 + \psi^2}} \, .
\label{eq:sincostheta}
\end{equation}
Thus, $ \tan 2 \theta^+ = \psi/\phi $.  If applying the inverse tangent to determine $ 2 \theta^+ $ from
this, we need to take the two equations (\ref{eq:sincostheta}) into account in choosing the correct branch.  
This clearly depends on the signs of $ \phi $ and $ \psi $, which is automatically dealt with if the four-quadrant
inverse tangent is used.  Consequently, (\ref{eq:sincostheta}) implies that
\[
\theta^+(x,y) = \frac{1}{2} \, \tilde{\tan}^{-1} \left( \psi(x,y), \phi(x,y) \right) \, , 
\]
which is chosen modulo $ \pi $ because of the premultiplier of $ 1/2 $ (the four-quandrant inverse tangent is
modulo $ 2 \pi $).  Thus, $ \theta^+ $ as defined here is identical to that given in (\ref{eq:thetaplus}), which
by Lemma~\ref{lemma:theta} is equivalent to (\ref{eq:thetaplus2}).

Next, given that the cosine function is always between $ - 1 $ and $ 1 $, we see that the local stretching must obey
\[
 \left[    - \sqrt{\phi^2 + \psi^2} + \frac{\left| \vecnabla u \right|^2 + \left| \vecnabla v \right|^2}{2}\right]^{1/2} 
 \le \Lambda \le
  \left[    \sqrt{\phi^2 + \psi^2} + \frac{\left| \vecnabla u \right|^2 + \left| \vecnabla v \right|^2}{2}\right]^{1/2} \, , 
  \]
and consequently the global stretching (\ref{eq:globalstretching})  satisfies
\begin{eqnarray}
\Sigma_f & \ge &  \int \! \! \! \! \int_\Omega \left[    - \sqrt{\phi^2 + \psi^2} + \frac{\left| \vecnabla u \right|^2 + \left| \vecnabla v \right|^2}{2}\right]^{1/2} \, \d x \, \d y \quad {\mathrm{and}} 
\label{eq:minpossible} \\
\Sigma_f & \le &  \int \! \! \! \! \int_\Omega \left[    \sqrt{\phi^2 + \psi^2} + \frac{\left| \vecnabla u \right|^2 + \left| \vecnabla v \right|^2}{2}\right]^{1/2} \, \d x \, \d y \, .
\label{eq:maxpossible}
\end{eqnarray}
for any choice of foliation.

Let $ f^+ $ be the foliation identified with the angle field $ \theta^+(x,y) $ at every location in $ \Omega_0 $.  Inserting this into (\ref{eq:stretching2}) renders the cosine term
$ 1 $, and thus the right-hand side of (\ref{eq:maxpossible}) is achieved for this foliation.  There can be no foliation
with gives a larger value of $ \Sigma_f $.  This foliation is equivalent to pointwise maximizing $ \Lambda $ in 
$ \Omega_0 $. 

Can there be a different acceptable foliation, $ \tilde{f} $, which also attains this maximum value for $ \Sigma_f $
(i.e., that $ \Sigma_{\tilde{f}} = \Sigma_{f^+} $)?  If so, there must
be a point $ \left( \tilde{x}, \tilde{y} \right) \in \Omega_0 $ where the induced slopes $ \theta_{\tilde{f}} $ 
and $ \theta^+ $ of the two
different foliations are different.  Given that foliations must be smooth, this implies the presence of an open
neighborhood $ N_\eps $ (with positive measure) around this point such that $ \cos 2 \left( \theta_{\tilde{f}} - \theta^+ 
\right) < 1 -  \eps $, for any given 
$ \eps > 0 $.  Thus the integrated local stretching in $ N_\eps $ for $ \tilde{f} $ is strictly less than that of $ f^+ $.
Since it is not possible to obtain a greater integrated stretching outside of $ N_\eps $ (because $ f^+ $, by forcing
the cosine term to take its maximum possible value, cannot be bettered), this would imply that the integrated stretching
of $ \tilde{f} $ over $ \Omega_0 $ is strictly less than that of $ f^+ $.  Given that the contribution to
the integral in $ I $ is independent of the foliation, this provides a contradiction.  Therefore, the foliation $ f^+ $,
corresponding to the choice of angle field $ \theta^+ $ as given in (\ref{eq:thetaplus}), maximizes $ \Sigma_f $, 
and is uniquely defined in $ \Omega_0 $.

The proof of Theorem~\ref{theorem:min} related to minimizing the global stretching is similar.  We use (\ref{eq:minpossible}), which corresponds to choosing $ \theta_f $ such that the term $ \cos 2\left( \theta^+ - \theta_f
\right) $ is always $ - 1 $. This tells us that $ \theta_f $ must be chosen perpendicular to $ \theta^+ $.  Thiis is
exactly the characterization used to determine $ \theta^- $ in (\ref{eq:thetaminus}), and the equivalence
to (\ref{eq:thetaminus2}) has been established in Lemma~\ref{lemma:theta}.

%%%%%%%%
\section{Local stretching connections related to Remark~\ref{remark:local}}
\label{sec:localstretching}

Given a location $ (x,y) $, suppose we wanted to determine the direction (encoded by an angle $ \theta $)
to place an infinitesimal line segment such that it stretches the most under $ \vec{F} $.  From (\ref{eq:stretching}),
we need to solve
\[
\sup_{\theta} \left\| \vecnabla \vec{F}(x,y) \left( \begin{array}{c} \cos \theta \\ \sin \theta \end{array} \right) \right\| :=
\left\| \vecnabla \vec{F} \right\| \,, 
\]
where the right-hand side is the operator norm of $ \vecnabla \vec{F} $.  This is computable
by the
square-root of the larger eigenvalue of $ \left[ \vecnabla \vec{F} \right]^\top \vecnabla \vec{F} $, i.e., of the Cauchy--Green
tensor $ C $ as defined in (\ref{eq:cauchygreen}).  
Given the map (\ref{eq:map}), since
\[
\vecnabla \vec{F} = \left( \begin{array}{cc} u_x & u_y \\ v_x & v_y \end{array} \right) \, , 
\]
it is clear that the Cauchy--Green strain tensor (as defined in (\ref{eq:cauchygreen})) is 
\[
\vec{C} := \left[ \vecnabla \vec{F} \right]^\top \, \vecnabla \vec{F} = \left( \begin{array}{cc}
u_x^2 + v_x^2 & u_x u_y + v_x v_y \\
u_x u_y + v_x v_y & u_y^2 + v_y^2 \end{array} \right) \, .
\]
Accordingly, the eigenvectors $ \lambda $ of the Cauchy--Green tensor (i.e., the singular values of $ \vecnabla \vec{F} $)
obey
%\[
%\left[ \left( u_x^2 + v_x^2 \right) - \lambda \right] \left[ u_y^2 + v_y^2 - \lambda \right] - \left[ u_x u_y + v_x v_y \right]^2 = 0 \, .
%\]
%Rearranging, we write this as
\[
\lambda^2 - \left( \left| \vecnabla u \right|^2 + \left| \vecnabla v \right|^2 \right) \lambda + \left[
(u_x^2 + v_x^2)(u_y^2 + v_y^2) - \left( u_x u_y + v_x v_y \right)^2 \right] = 0 \, , 
\]
and thus
\begin{eqnarray}
\lambda 
%& = & \frac{\left( \left| \vecnabla u \right|^2 + \left| \vecnabla v \right|^2 \right)  \pm \sqrt{ \left( u_x^2 + u_y^2 + v_x^2 + v_y^2 \right)^2 - 4 \left[ (u_x^2 + v_x^2)(u_y^2 + v_y^2) -  \left( u_x u_y + v_x v_y \right)^2 \right]}}{2} \
%\nonumber \\
& = & \frac{\left| \vecnabla u \right|^2 + \left| \vecnabla v \right|^2 }{2} \pm
\sqrt{ \phi^2 + \psi^2}
\label{eq:cauchygreenev}
\end{eqnarray}
by using the definitions for $ \phi $ and $ \psi $ in (\ref{eq:phi}) and (\ref{eq:psi}).  

We assume that $ \phi $ and $ \psi $ are not simultaneously $ 0 $ (in our framework, that we are not in 
$ I $).  Clearly, the larger value of $ \lambda $ is obtained by taking
the positive sign, and the square-root of this is the matrix norm of $ \vec{C} $.  This gives exactly the pointwise
maximized local
stretching of $ \Lambda^2 $ as defined in (\ref{eq:stretching2}), which satisfies
\[
\left( \Lambda^+\right)^2 = \frac{\left| \vecnabla u \right|^2 + \left| \vecnabla v \right|^2 }{2} + 
\sqrt{ \phi^2 + \psi^2} \, .
\]
The quantity $ \Lambda^+ $ defined above (and also given in the main text as (\ref{eq:ftle})), is related to 
the  {\em finite-time Lyapunov exponent} or simply the {\em Lyapunov exponent}. We note that for defining $ \theta^+ $ (for optimizing stretching) we required that $ (x,y) \ne I $, but $ \Lambda^+ $
can be thought of as a field on all of $ \Omega $.

{\em Obtaining} the eigenvector of the Cauchy--Green tensor $ C $ corresponding to the $ \lambda $ associated with (\ref{eq:ftle}) is somewhat unpleasant.  However, our
equation for $ \theta^+ $ in (\ref{eq:thetaplus2}) indicates that eigenvector---modulo a nonzero scaling---can be written as
\[
\tilde{\vec{w}}^+ = \left( \begin{array}{c} \psi \\
- \phi + \sqrt{\phi^2 + \psi^2} \end{array} \right) \, ,
\]
as long as this value is not zero (which is when $ \psi = 0 $ and $ \phi > 0 $, in which case $ \vec{w}^+ = \left( 1 \, \, 
0 \right)^\top $).  Tedious calculations reveal that 
\begin{eqnarray*}
C \, \tilde{\vec{w}}^+ & = & \left( \begin{array}{cc} u_x^2 + v_x^2 & \psi \\ \psi & u_y^2 + v_y^2 \end{array} \right) \, 
\left( \begin{array}{c} \psi \\
- \phi + \sqrt{\phi^2 + \psi^2} \end{array} \right) \\
& =& \ldots = \left( \Lambda^+\right) ^2 \, \tilde{\vec{w}}^+ \, , 
\end{eqnarray*}
verifying that our expression does indeed give the relevant eigenvector.  
The situation of $ \psi = 0 $ and $ \phi > 0 $ is easy to check as well.  Using $ \tilde{\vec{w}}^+ = \left( 1 \, \, \, 0 \right)^\top $,
we once again get
\[
C \tilde{\vec{w}}^+  = \left( \phi + \frac{\left| \vecnabla u \right|^2 + 
\left| \vecnabla v \right|^2}{2} \right)  \tilde{\vec{w}}^+ = \left( \Lambda^+\right)^2 \tilde{\vec{w}}^+ \, .
\]
The eigenvector field $ \tilde{\vec{w}}^+ $ of $ C $ (or a scalar multiple of it) is only defined on $ \Omega_0 $.  In the literature, this is variously
referred to as the {\em Lyapunov} \cite{wolfesamelson,ramasubramaniansriram} or {\em Oseledec} \cite{oseledec} vector field, related to the local direction (in the domain of 
$ \vec{F} $) in which the stretching due to the application of $ \vec{F} $ will be the most.  If $ \vec{F} $ were a
flow map derived from a flow over a finite-time, then these would depend both on the initial time $ t_0 $ and a
time $ t $ at the end. In other words, $ \vec{F} $ would be the flow map from time $ t_0 $ to $ t $.  In this situation,
the variation of the vector field with respect to both $ t_0 $ and $ t $ is to be noted.

The smaller eigenvalue of the Cauchy--Green tensor is obtained by taking the negative sign in (\ref{eq:cauchygreenev}),
which gives 
\[
\left( \Lambda^-\right)^2 = \frac{\left| \vecnabla u \right|^2 + \left| \vecnabla v \right|^2 }{2} - 
\sqrt{ \phi^2 + \psi^2} \, .
\]
This is clearly the local stretching minimizing choice, corresponding to choosing $ \theta = \theta^- $ (i.e., making
the cosine term equal to $ - 1 $).  The corresponding eigenvector $ \tilde{\vec{w}}^- $ can be verified (as above) to be in the direction
specified by $ \theta^- $.  However,  given that $ \sqrt{\phi^2 + \psi^2} \ne 0 $, we have distinct eigenvalues for the
symmetric matrix $ C $, and thus the two eigenvectors must be orthogonal by standard spectral theory.  Hence we can
easily conclude that 
$ \theta^- $ corresponds to $ \tilde{\vec{w}}^- $, the eigenvector of $ C $ corresponding to the smaller eigenvalue.

The situation in which the eigenvalues of $ \vec{C} $ coincide corresponds to `singularities,' in particular because 
this means that an orthogonal eigenbasis may not exist.  This can only occur when the eigenvalues are
repeated, and from (\ref{eq:cauchygreenev}) this occurs only when $ \phi^2 + \psi^2 = 0 $.  Thus, {\em both} $ \phi $ and $ \psi $ must be
zero.  Thus corresponds exactly to the isotropic set $ I $, in Definition \ref{def:isotropic} and Lemma \ref{lemma:puncture}.

We note that Haller \cite{haller} uses streamlines of the eigenvector fields from the Cauchy--Green tensor
in his theories of variational Lagrangian coherent structures, looking for example for curves to which there is extremal
attraction or repulsion due to a flow over a given time period.  Our foliations obtained here,
corresponding to globally maximizing and minimizing stretching, are generated from fibers of the {\em same} fields.  Therefore, our insights
into singularities and branch-cut discontinuities are therefore relevant to these 
approaches as well.

%%%%%%%%%%%%%%%%%%
\section{Singularity classification}
\label{sec:fourquadrant}

This section provides explanations for the nondegenerate singularity classification of Property~\ref{property:singularities}.
  Given
the transverse intersection of the $ \phi = 0 $ and $ \psi = 0 $ contours at a singularity $ \vec{p} $, we examine nearby contours
not in standard $ (x,y) $-space, but in $ (\phi,\psi) $-space, in which $ \vec{p} $ is at the origin.  The angle fields
$ \theta^\pm $ are the defining characteristics of the foliation, and thus we show in Fig.~\ref{fig:fourquadrant}(a) 
a schematic of the maximizing angle field $ \theta^+ $.  A nonstandard labelling of the $ \phi $ 
and $ \psi $ axes is used here because the relative orientations of the positive axes $ \phi_+ $ and $ \psi_+ $
(the directions in which $ \phi>0 $ and $ \psi>0 $ resp.) and negative axes $ \phi_- $ and $ \psi_- $ is
related to whether $ \vec{p} $ is right- or left-handed.  Thus, Fig.~\ref{fig:fourquadrant}(a) corresponds to $ \vec{p} $ being right-handed. The slope fields and expressions indicated are based on the four-quadrant inverse tangent
 (\ref{eq:thetaplus}), expressed 
 in terms of the {\em regular} inverse tangent in each quadrant.   We also express the
values of $ \theta^+ $ on each of the axes in Figs.~\ref{fig:fourquadrant}(a), along which $ \theta^+ $ is seen to be constant.  

In Figs.~\ref{fig:fourquadrant}(c), just below, we indicate the angle field $ \theta^+ $ by drawing tiny lines which
have the relevant slope.  
What happens when we `connect these lines' to form a foliation is shown underneath in
Figs.~\ref{fig:fourquadrant}(e).  The foliation bends around the origin (shown as the blue point $ \vec{p} $), 
effectively rotating around it by $ \pi $.   However, it must be cautioned that while Fig.~\ref{fig:fourquadrant}(e) seems to indicate that the fracture ray lies
along $ \phi_+ $, this is in general not the case.  The angle fields shown in Figs.~\ref{fig:fourquadrant}(c) and (e) 
display directions in {\em physical} ($ \Omega $) space, in which the $ \phi = 0 $ and $ \psi = 0 $ contours
intersect in some slanted way.  We show one possibility in Fig.~\ref{fig:fourquadrant}(g), in which the fracture ray
will be approximately from the northwest.  We identify
$ \vec{p} $ in this case an {\em intruding point} or a {\em $ 1 $-pronged singularity}. The
nearby $ \mbox{SORF}_{max}$ curves rotate by $ \pi $ around it.

\begin{figure}
 {\includegraphics[width=0.3\textwidth,height=0.21\textheight]{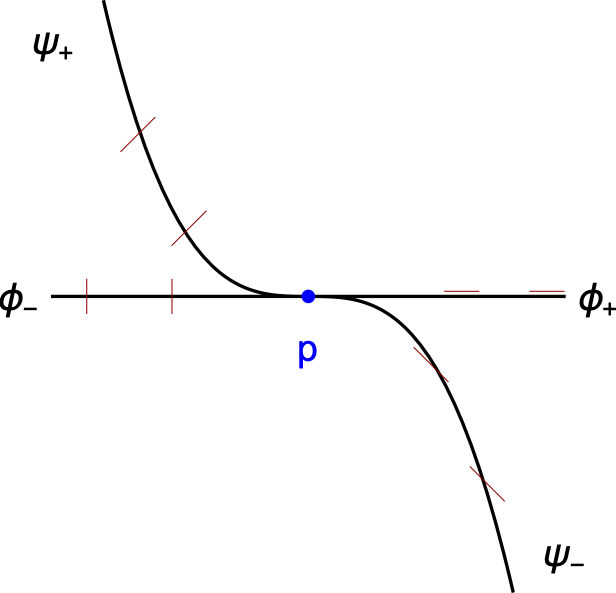}} 
 {\includegraphics[width=0.3\textwidth,height=0.21\textheight]{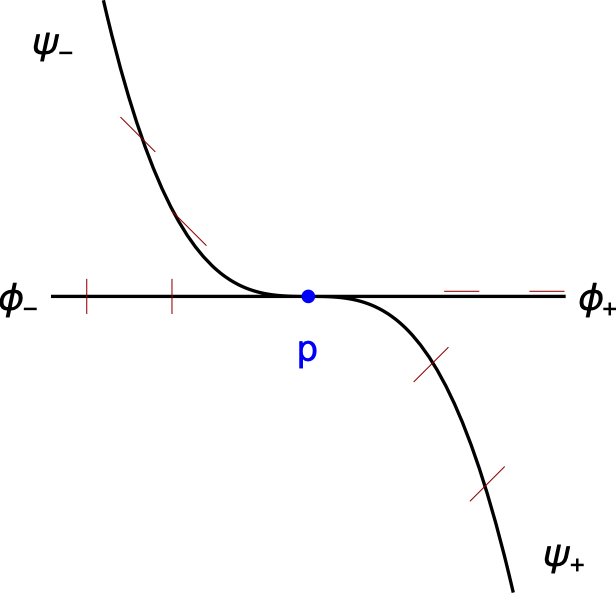}} 
 {\includegraphics[width=0.3\textwidth,height=0.15\textheight]{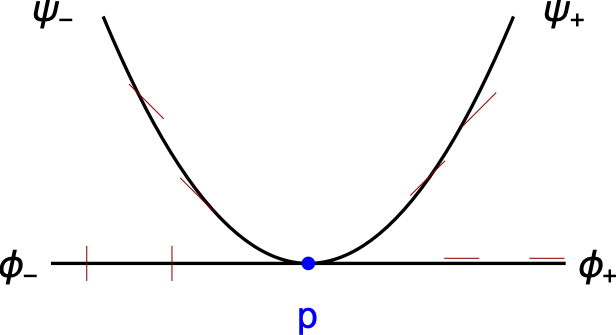}} 
\caption{$ \mbox{SORF}_{max}$ near $ \vec{p} $ when transversality is relaxed: (a), (b) and (c)  show different possibilities
for axes to intersect, and the corresponding $ \mbox{SORF}_{max}$ topologies are illustrated in Fig.~\ref{fig:degenerate}.}
\label{fig:nontransverse}
\end{figure}

In the right-hand panels of Fig.~\ref{fig:fourquadrant} we examine the other possibility of $ \vec{p} $ being left-handed.
This is achieved
in Fig.~\ref{fig:fourquadrant}(b) by simply flipping the $ \psi_- $ and $ \psi_+ $ axes, and retaining the information
that we have already determined in Fig.~\ref{fig:fourquadrant}(a).  The corresponding slope field is displayed
in Fig.~\ref{fig:fourquadrant}(d).  The fracture ray (also along the $ \phi_+ $-axis in this case) now {\em separates
out} curves coming from the right, rather than causing them to turn around the origin.  
Fig.~\ref{fig:fourquadrant}(f) demonstrates this
behavior, obtained by connecting the angle fields into curves.  There are two other fracture rays generated by this
process of separation, because curves in the $ \phi_- $ region are forced to rotate away from the origin without
approaching it.  Fig.~\ref{fig:fourquadrant}(h) is an orientation-preserving rotation of the axes in Fig.~\ref{fig:fourquadrant}(f), which highlights that the directions of the three fracture rays are based on the orientations of the axes in
physical space.  Based on the topology of the foliation, when $ \vec{p} $ is left-handed, we thus have
a {\em separating point} or {\em $ 3 $-pronged singularity}. 

Suppose next that the nondegeneracy of $ \vec{p} $ is relaxed mildly by allowing
the $ \phi = 0 $
and $ \psi = 0 $ contours (both still considered to be one-dimensional) to intersect {\em tangentially} at $ \vec{p} $.  
To achieve this, imagine bending the $ \psi $-axis
in Figs.~\ref{fig:fourquadrant}(a) and (c) so that it becomes tangential to the $ \phi $-axis, but the axes still
cross each other.  This degenerate situation is shown in Fig.\ref{fig:nontransverse}(a), and we note that
the orientation of the axes remains right-handed despite the tangency.  Connecting the angle field lines 
gives the relevant topological structure of Fig.~\ref{fig:degenerate}(a). The topology is very close to the nondegenerate intruding point, but there
is an accumulation of curves towards the fracture ray from one side.  
It is easy to verify (not shown) that there is no change in this topology if the tangentiality
shown in Fig.~\ref{fig:nontransverse}(a) goes in the other direction, with $ \psi_+ $ becoming tangential to $ \phi_+ $
and $ \psi_- $ to $ \phi_- $.  Fig.~\ref{fig:nontransverse}(b) examines the impact on the degenerate left-handed situation; Fig.~\ref{fig:degenerate}(b) indicates that the fracture ray acquires a similar one-sided accumulation effect, while
the remainder of the portrait remains essentially as it was.  So this is a degenerate separation point.  Finally, 
in Fig.~\ref{fig:nontransverse}(c) we consider the case where the tangentiality is such that the $ \phi $- and $ \psi $-axes
do not cross one another.  In this case, drawing connecting curves reveals that the topology is a combination of degenerate intruding and separating points,
and is illustrated in Fig.~\ref{fig:degenerate}(c).
Testing the other possibilities (interchanging the $ \psi_- $ and $ \psi_+ $ axes locations, and doing the same
analysis with them below the $ \phi $-axis) yields no new topologies.  One way to 
rationalize this is that the relative (degenerate) orientation between the negative axes and that between the positive axes is in
this case exactly opposite; one is as if there is a right-handed orientation, while the other is left-handed.  

%%%%%%%%%%%
\section{Proof of Theorem~\ref{theorem:nonsmooth}}
\label{sec:nonsmooth}

We have established via Fig.~\ref{fig:branchcut} that if there exists a nondegenerate singularity $ \vec{p} $, 
then $ \vec{w}^+ $ is not continuous across the branch cut $ B $.  This vector field is `the' Lyapunov
vector field, generated from the eigenvector field corresponding to the larger eigenvalue 
of the Cauchy--Green tensor field, where this is well-defined (i.e., in $ \Omega_0 $).  However,
a vector field associated with the angle field  $ \theta^+ $ is not unique, as is reflected in the 
presence of the arbitrary function $ m $ in 
(\ref{eq:curvegenerate}).  The nonuniqueness is equivalent to the potential of 
scaling Lyapunov vectors in a nonuniform way in $ \Omega
\setminus I $, by multiplying by a nonzero scalar.  The question is: is it possible to remove the discontinuity that
$ \vec{w}^+ $ has across $ B $ by choosing a scaling function $ m $? 

From Fig.~\ref{fig:branchcut}, we argue that the answer is no.  Imagine going around the black dashed curve,
$ C $, 
and attempting to have $ \vec{w}^+ $ be continuous while doing so.  Since $ \vec{w}^+ $ has a jump
discontinuity across $ B $, it will therefore be necessary to choose $ m $ to have the opposite jump
discontinuity for $ m \vec{w}^+ $ to be smooth.  So $ m $ must jump from $ + 1 $ to $ - 1 $
in a certain direction of crossing.  However, since $ \vec{w}^+ $ is continuous
on $ C \setminus B $, to retain this continuity $ m $ must also remain continuous along $ C \setminus B $.
This implies that $ m $ must cross zero at some point in $ C \setminus B $.  Doing so would render the
Lyapunov vector $ \vec{w}^+ $ invalid.  We have therefore established Theorem~\ref{theorem:nonsmooth}
using elementary
geometric means.  We remark that this theorem is analogous to the classical ``hairy ball'' theorem due to
Poincar\'{e} \cite{poincare}.

%%%%%%%%%%%%
\section{Branch cut effects on computations}
\label{sec:branchcut}

If $ \vec{p} $ is a nondegenerate singularity, then the vector field of (\ref{eq:curvegenerate}) with $ m = 1 $ and
the choice of the positive sign ($ \mbox{SORF}_{max}$) will locally have the behavior as shown in Fig.~\ref{fig:branchcut}.  
Now, in general, in finding a $ \mbox{SORF}_{max}$ which passes through $ (x_0,y_0) $, we can implement (\ref{eq:curvegenerate})
for the choice of $ m = 1 $, in  {\em both} directions (increasing and decreasing $ s $), thereby obtaining the
curve which crosses the point.  An equivalent viewpoint is that we implement (\ref{eq:curvegenerate}) with $ m = 1 $,
and $ s > 0 $, and then implement it with $ m = - 1 $ while retaining $ s > 0 $. 

If using (\ref{eq:curvegenerate}) with $ m = +1 $ (globally) and $ \vec{w}^+ $ to generate a $ \mbox{SORF}_{max}$ curve, the vector field in Fig.~\ref{fig:branchcut}(a) must be followed.  
However, it is clear that anything approaching the branch cut $ B $ gets pushed away in the vertical direction.
Thus, $ \mbox{SORF}_{max}$ curves near $ B $ will in general be difficult to find.

The solution appears to be to set $ m = - 1$, which reverses the vector field.  However, this is essentially the diagram in Fig.~\ref{fig:branchcut}(b), corresponding to a left-handed $ \vec{p} $.  
This is of course equivalent to implementing (\ref{eq:curvegenerate}) with $ m = + 1 $ but in the $ s < 0 $ direction.
Curves coming in to $ B $ now get stopped abruptly, because the vector field on the other side of $ B $ directly
opposes the vertical motion.  Thus, curves will not cross $ B $ vertically.  However, since any incoming curve
will in general have a vector field component tangential to $ B $, this will cause a veering along the curve  $ B $.
The curve will continue along $ B $, because the vector field pushes in on to $ B $ vertically, preventing departure from it.
Thus when numerically finding $ \mbox{SORF}_{max}$ curves, curves which appear to tangentially approach the branch cut
$ B $ will be seen.  These curves are {\em not} real $ \mbox{SORF}_{max}$ curves because, as is clear from Fig.~\ref{fig:branchcut}, the actual vector field
is {\em not} necessarily tangential to $ B $.  That is, the branch cut is not necessarily a streamline of the direction
field $ \theta^+ $.

A similar analysis (not shown) indicates that if using $ \vec{w}^- $ (as suggested via Theorem~\ref{theorem:min}) to generate $ \mbox{SORF}_{min}$ curves, then these curves will not cross $ B $ {\em horizontally}, and
also have the potential for tangentially approaching $ B $ in a spurious way.  Notice moreover that, while we
have discussed the branch cut locally near $ \vec{p} $, these objects extend through $ \Omega_0 $,
potentially connecting with several singularities.

Finally, suppose there are parts of  $ B $ that are two-dimensional regions.  In such regions, 
Fig.~\ref{fig:fourquadrant}(a) indicates that the angle field $ \theta^+ $ is
vertical; alternatively, see (\ref{eq:thetaplus_zero}).  Consequently, $ \theta^- $ is horizontal everywhere.  
However, numerical issues as above will occur when crossing the one-dimensional boundary $ \bar{B} \setminus B $,
due to the inevitable issue of the reversal of the vector field along at least one part of this boundary.

%%%%%%%%%%%%
\bibliographystyle{elsarticle-num} 
 \bibliography{foliation}

\end{document}